\documentclass[11pt]{article} %[referee]{svjour3}%
\usepackage{graphics}
\usepackage{amsfonts}
\usepackage{amsmath}
\usepackage{amssymb}
\usepackage{mathrsfs}
\usepackage{graphicx}
\usepackage{amsthm}
\usepackage{soul} % \ul

\renewcommand \appendix{
\setcounter{section}{0}
\setcounter{subsection}{0}
\setcounter{figure}{0}
\setcounter{table}{0}
\setcounter{equation}{0}
\renewcommand \thesection{\Alph{section}}
\renewcommand \theequation{\Alph{section}\arabic{equation}}
\renewcommand \thefigure{\Alph{section}\arabic{figure}}
\renewcommand \thetable{\Alph{section}\arabic{table}}
}

\newtheorem{theorem}{Theorem}[section]
\newtheorem{corollary}{Corollary}[theorem]
\newtheorem{proposition}{Proposition}[theorem]
\newtheorem{lemma}[theorem]{Lemma}
\newtheorem{definition}{Definition}[section]

\begin{document}

\title{{\Large \bf Isochrony in 3D radial potentials}\\
{\small \bf From Michel H{\'e}non ideas to isochrone relativity:}\\
{\small \bf classification, interpretation and applications}\\
~}

\author{Alicia Simon-Petit$^1$\\
J\'{e}r\^{o}me Perez$^1$\\
 Guillaume Duval $^2$\\
{\small $^1$ Ensta ParisTech, Paris Saclay University, Applied
Mathematics Laboratory, France }\\ 
{\small $^2$ INSA Rouen, Mathematics \& Informatics Laboratory, France}}

\maketitle

\begin{abstract}
Revisiting and extending an old idea of Michel H\'{e}non, we geometrically
and algebraically characterize the whole set of isochrone potentials. Such potentials
are fundamental in potential theory. They appear in spherically
symmetrical systems formed by a large amount of charges (electrical or
gravitational) of the same type considered in mean-field theory. Such potentials are defined by the fact that
the radial period of a test charge in such potentials, provided that it exists, depends only on its energy and not on its angular momentum. Our characterization of 
the isochrone set is based on the action of a real affine subgroup 
on isochrone potentials related to parabolas in the $\mathbb{R}^2$ plane. Furthermore, any isochrone orbits are mapped onto associated Keplerian 
elliptic ones by a generalization of the Bohlin transformation. This mapping allows us to understand the isochrony property of a given potential 
as relative to the reference frame in which its parabola is represented. We detail this isochrone relativity in the special relativity formalism.
We eventually exploit the completeness of our characterization and the relativity of isochrony to propose a deeper
understanding of general symmetries such as Kepler's Third Law and
Bertrand's theorem. 
\end{abstract}

%\keywords{Theoretical Astrophysics; Potential Theory (Mathematics); Integrable Systems; Differential Equations; Classical Gravitation.}

\tableofcontents

\section{Introduction}

Macroscopic properties of self-gravitating systems can be derived from the
orbits of their components, e.g. stars. These orbits are designed by the
potential -- density pair ($\psi$ -- $\rho$) involved in Poisson's
equation $\Delta \psi=4\pi G\rho$. This pair forms a steady-state model for
such astrophysical systems and there are essentially two ways to produce a physically
relevant model --- one depending on empirical input, the other on theoretical input.

By compiling observational data, one can look for the emergence of an
\emph{empirical} model. For example, consider 
de Vaucouleur's law for elliptical
galaxies in the middle of the twentieth century \cite{r1/4}. In that paper,
the author remarks that the projection $I\left(  R\right)  \ $of the
luminosity onto the plane of the sky of elliptical galaxies varies as a
function of an apparent distance $R$ from the center as $I\left(  R\right)
\propto \exp \left(  -R^{1/4}\right)  $. From $I\left(  R\right)  $, assuming a
given mass--to--light ratio, one can build the mass density $\rho$ of the
system and, solving Poisson's equation, obtain a gravitational potential for
elliptical galaxies. This problem is generally ill-posed: as a matter of
fact, after the projection, a lot of ``good" potentials (Jaffe \cite{jaffe},
Hernquist \cite{hernquist}, Dehnen \cite{Dehnen} or NFW \cite{NFW})
produce $R^{1/4}$--compatible luminosity profiles. Apart from this empirical
property all these famous models are poorly justified physically.

The reverse approach is much less investigated. The mass density is a 
marginal velocity law of the one-particle distribution function $f$ associated with a
self-gravitating system. This function $f\left(  t,\mathbf{r},\mathbf{p}%
\right)$ describes the statistical properties of a test particle of mass $m$,
position $\mathbf{r}$ and momentum $\mathbf{p}$ in the mean field
gravitational potential $\psi \left(  t,\mathbf{r}\right)$. These two
functions satisfy the Collisionless-Boltzmann and Poisson system%
\[
\left \{
\begin{array}
[c]{l}%
\frac{\partial f}{\partial t}+\left \{  f,E\right \}  =0,\\
\Delta \psi=4\pi G\rho=4\pi mG\int fd\mathbf{p},%
\end{array}
\right.
\]
where $E=\frac{\mathbf{p}^{2}}{2m}+m\psi$ is the total energy of the test
particle and $\left \{  ,\right \}  $ stands for the Poisson bracket. Using basic
properties of these brackets, one can see that the simplest steady states are
described by $f\left(  E\right)$: this is the simplest case of Jeans'
theorem (see e.g. \cite{BT08}). Involving Gidas-Ni-Nirenberg theorem
\cite{GNN}, one can show (\cite{PA} sect. 2) that, if their total mass is finite, the corresponding
self-gravitating systems are spherical and isotropic and thus their gravitational
potentials are radial, $\psi=\psi \left(  r\right)  $ with $r=\left \vert
\mathbf{r}\right \vert $. Stability analysis can restrict possible steady
states to decreasing and positive $f$ but nothing general can be said
anymore about the choice of an equilibrium in this context. Adding thermodynamic
considerations, Lynden-Bell \cite{VioRel} has initiated a long debate. Based on the fact that
in three spatial dimensions there is no regular isothermal steady states with finite mass,
this debate is often summarized by the fact that isolated self-gravitating systems could settle
down in a truncated isothermal state with a core-halo density distribution. The size
of the core and the slope of the halo depend on structural dissipation
which can occur in the system. This point will be discussed in a forthcoming
paper.

In a singular and seminal paper in French, H\'{e}non \cite{Henon58} (for an
English translation see \cite{BinneyHEnon}) followed  
another way to address
this problem. Radial potentials confer to any of their confined test particles
the property to have a periodic radial distance from the center of the system.
This radial period $\tau_{r}\ $depends generically on the two physical
parameters of this test particle: its energy $E$ and the modulus $L^{2}$ of
its squared angular momentum. H\'{e}non remarks that orbits confined around
the center of the system (which evolve generically in a harmonic potential)
and orbits confined to the outer parts (which evolve generically in a
Keplerian potential) have a radial period that depends only on $E$. He then
proposed looking for a general potential which could be characterized by this
property. He succeeded by finding his famous isochrone potential.
Although his potential gives a mass density in pretty good accordance with
some of the observed globular clusters at the time, history has decided to
follow another direction. In his conclusion, Michel H{\'e}non proposed a mechanism 
based on resonances that could lead to the formation of an isochrone. This mechanism
needed to be considered more accurately and proved substantially~\cite{Henon58,BinneyHEnon}, but it has not been further investigated.
In addition to Lynden-Bell's work on violent relaxation and the
above-mentioned debate that followed, the observational data refinement and the
development of numerical simulations revealed a great variety of profiles for
self-gravitating systems and H\'{e}non's isochrone became one among them.
Recent works (in a paper in preparation by Simon-Petit, A., Perez, J, and Plum, G.) 
reveal that, as suggested by H\'{e}non \cite{Henon58}
in his conclusion, isochrony could in fact be inherited from the formation
process of isolated self-gravitating systems. Hence there could be a
fundamental initial state from which, after the initial collapse, the observed
diversity could arise.

For all of these reasons, we have decided to revisit in detail
isochrony in radial potential-governed systems. Inspecting 
H\'{e}\-non ideas we have found that his work is far from exhaustive in a
mathematical sense even if the potential he has found might be one of the most
important for physical applications. We propose in this paper to characterize
the whole set of isochrone potentials in a rigorous way. This characterization
will help for a global understanding of the importance of the isochrone
property and will clarify some important physical symmetries occurring in
gravitation like Kepler's Third law or Bertrand's theorem.

The paper is organized as follows. In the spirit of Michel H\'{e}non, section
\ref{sec2} is devoted to geometry. In sec.~\ref{subsecHenonparabola} we first
recall the basics of the problem of potential isochrony, general definitions
and the H\'{e}non link between isochrony and parabolas. In addition, we call
for a rigorous proof of this parabola property which is given in
appendix~\ref{appendixa}. Taking into account very general physical properties, we
introduce in sec. \ref{subsec:GalPropIsoParabolas} appropriate
transformations 
and prove three lemmas
which allow us to restrict the study to parabolas passing through the origin
with a vertical or a horizontal tangent. As these transformations 
leave invariant vertical lines, these three lemmas will
make clear the decomposition of parabolas into four families: ones with a
vertical symmetry axis (straight parabolas) and others (tilted parabolas) that
are classified in three different types depending on the parabola orientation and vertical tangent position.
These transformations and hence this distinction were not identified by
H\'{e}non who also missed some elements of the isochrone set. Thanks to this
geometric decomposition, we deduce in
sec. \ref{subsec:isochroneclassification} the whole set of isochrone potentials. It is a
modulus space in which each point is a potential from one of the four classes 
of equivalence of parabolas under the action of the previous
transformations. 
This algebraic representation classifies the four isochrone potential types but 
separates them in a partition of four equivalence classes. 
However, isochrone potentials are unified, linking orbits together.

In section~\ref{sec3}, we focus on the isochrone orbits. Based on the fundamental differential orbital equation, we present in 
sec.~\ref{subsec_LBidea} the most general transformation which preserves isochrony and angular momentum 
when applied to a given orbit. This linear application is then identified as a generalization of the well-known 
Bohlin transformation (\cite{Bohlin,ArnoldBol}) as well as the brilliant idea of Donald Lynden-Bell~\cite{LyndenBell}. 
It continuously maps isochrone orbits onto their Keplerian associates. 
This Keplerian character of a given isochrone orbit is developed in sec.~\ref{subsec:isochronerelativity}.
Adapting the time, energy and angular momentum of a given isochrone orbit 
 in an isochrone potential, it is shown in detail how to map this orbit onto its associated Keplerian one in the appropriate frame.

The last section is devoted to physical applications of this isochrony classification and interpretation. 
We first present in sec.~\ref{subsec_physapp} the physical
properties of systems associated with isochrone potentials. In particular, we
give in table \ref{table_tau_n} the explicit formulation of $\tau_{r}\left(
E\right)  $ and $n_{\varphi}\left(  L^{2}\right)  $ for all isochrone
potentials. The properties of $\tau_{r}\left(  E\right)  $ allow us to give a
generalization of Kepler's third law in sec.~\ref{subsec_kep_third_law}.
Eventually we show in sec.~\ref{BTh} that the famous Bertrand's theorem about closed orbits in
radial potentials is just a corollary of a general property of isochrone orbits.
 
Four appendices detail important results for isochrony used in the paper. 

\section{The isochrone geometry\label{sec2}}

\subsection{H\'{e}non's parabola\label{subsecHenonparabola}}

We consider a stellar system described by a gravitational potential
$\psi \left(  \mathbf{r}\right)  =\psi \left(  r\right)$, where $\mathbf{r}$ is
the position vector of a test particle of mass $m$ confined in this system. The orbit of
this test particle is contained in a plane. In this plane, the two parameters
of this orbit are its energy $E=m\xi$ and the norm of its angular momentum
$L=m\Lambda$. Both these two parameters contribute to the definition of the
gravitational potential of the cluster $\psi \left(  r\right)  $ and to the
computation of the distance $r$ between the star and the center of mass of the
cluster at each time $t$. This contribution is summarized in the definition of
the energy of the star,%
\begin{equation}
\xi=\frac{1}{2}\left(  \frac{dr}{dt}\right)  ^{2}+\frac{\Lambda^{2}}{2r^{2}%
}+\psi \left(  r\right)  =\mathrm{cst.} \label{defE}%
\end{equation}

We are interested in increasing\footnote{This restriction characterizes the
gravitational interaction for which Gauss' theorem in spherical symmetry
indicates that $\frac{d\psi}{dr}=\frac{GM\left(  r\right)  }{r^{2}}>0$.}
potentials $\psi \left(  r\right)  \ $for which\ the \textsc{ode} \eqref{defE}
admits periodic solutions, named hereafter Periodic Radial Orbits
(\textsc{pro}s). The effective potential$\  \psi_{e}(r)=\frac{\Lambda^{2}%
}{2r^{2}}+\psi \left(  r\right)  $ then reaches a global minimum and diverges to
$+\infty$ when $r\rightarrow0$ as shown in figure \ref{poteff2}. When
they exist, the apoastron at distance $r_{a}$ and periastron at $r_{p}$ of a \textsc{pro}
are given by the two intersections of the graph of $\psi_{e}$ with constant
$\xi-$lines. For a given energy $\xi_{c}$ corresponding to the minimum of $\psi_e$, the distance $r_a=r_p$ and the orbit is circular.

\begin{figure}[ptb]
\centering
\resizebox{0.7\textwidth}{!}{\includegraphics{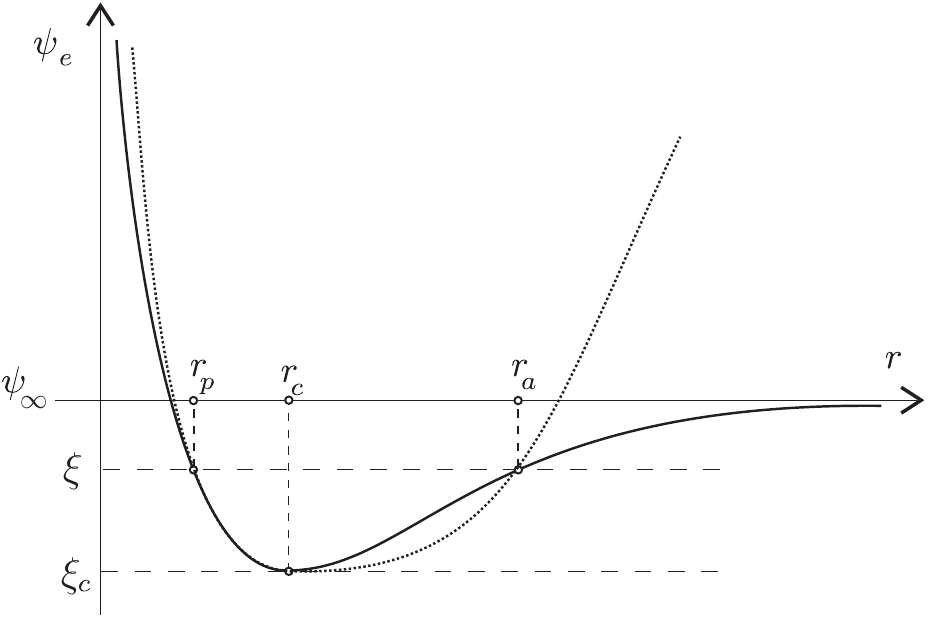}}
\centering
\caption{Physical effective potentials that allow Periodic Radial Orbits. The
solid curve corresponds to a finite $\psi_{\infty}$ and the dashed curve
corresponds to an infinite $\psi_{\infty}$.}%
\label{poteff2}%
\end{figure}

\bigskip In order to clarify the vocabulary we will use, let us define two
fundamental potentials in this context.

\begin{definition}\label{def:kepetharm}
The harmonic potential is defined by $\psi_{\mathrm{ha}}\left(
r\right)  =\frac{1}{2}\omega^{2}r^{2}$ with $\omega \neq0$. We call
the potential $\psi_{\mathrm{ke}}\left(  r\right)  =-\dfrac{\mu}{r} $ with
$\mu>0$ a Keplerian potential.
\end{definition}

To get the existence of a global minimum of the effective potential $\psi_e$ and hence of
\textsc{pro}'s, we have to specify the behavior of the potential $\psi$ when
$r\rightarrow0$. This is the objective of the following lemma.

\begin{lemma}
\label{lemme0}If for some $\Lambda \geq0$, the effective potential $\psi
_{e}\left(  r\right)  \rightarrow+\infty$ when $r\rightarrow0$, then
$\lim \limits_{r\rightarrow0}r^{2}\psi \left(  r\right)  =\ell<\infty$.
Conversely, if $\lim \limits_{r\rightarrow0}r^{2}\psi \left(  r\right)
=\ell<\infty$, then for any $\Lambda \geq0$ the effective potential $\psi
_{e}\left(  r\right)  \rightarrow+\infty$ when $r\rightarrow0$ provided that
$\ell>-\Lambda^{2}$.
\end{lemma}

\begin{proof}
\bigskip The converse claim is obvious since $\lim \limits_{r\rightarrow0}%
r^{2}\psi_{e}\left(  r\right)  =\Lambda^{2}+\ell>0$ if $\ell>-\Lambda^{2}$.
For the first claim, let us assume that $\lim \limits_{r\rightarrow0}r^{2}%
\psi \left(  r\right)  $ is infinite. Then $\lim \limits_{r\rightarrow0}%
\psi \left(  r\right)  $ is also infinite. But since $r\mapsto \psi \left(
r\right)  $ is increasing we must have $\lim \limits_{r\rightarrow0}\psi \left(
r\right)  =-\infty$. So for any $\Lambda>0$, by choosing $r$ close enough to
$0$, we would get
\[
r^{2}\psi \left(  r\right)  <-\Lambda^{2}\implies \psi \left(  r\right)
<-\frac{\Lambda^{2}}{r^{2}}\implies \psi_{e}\left(  r\right)  <-\frac
{\Lambda^{2}}{2r^{2}}%
\]
which implies $\lim \limits_{r\rightarrow0}\psi_{e}\left(  r\right)  =-\infty$. The
claim follows by contraposition. \qed
\end{proof}

These restrictions allow a \textsc{pro} provided that $\xi \in \left[  \xi
_{c},\psi_{\infty}\right)$, where $\psi_{\infty}=\lim \limits_{r\rightarrow
+\infty}\psi \left(  r\right)  $ may be infinite. The total and/or the central
mass of such systems could be infinite but the radial period
\begin{equation}
\tau_{r}={2\int_{r_{p}}^{r_{a}}}\frac{dr}{\  \sqrt{2\left[  \xi-\psi \left(
r\right)  \right]  -\dfrac{\Lambda^{2}}{r^{2}}}} \label{radialperiod}%
\end{equation}
is always finite. This period corresponds to the total duration of the transfer from
$r_{a}$ to $r_{p}$ and back, and it is also related to the $\xi-$derivative
of the radial action $\mathcal{A}_{r}$, which gives the radial pulsation (see
for example \cite{BT08} p. 221)
\begin{equation}
\Omega_{r}^{-1}=\frac{\tau_{r}}{2\pi}=\frac{\partial \mathcal{A}_{r}}{\partial \xi
}\text{ with }\mathcal{A}_{r}=\frac{1}{\pi}{\int_{r_{p}}^{r_{a}}}%
\sqrt{2\left[  \xi-\psi \left(  r\right)  \right]  -\dfrac{\Lambda^{2}}{r^{2}}%
}dr\text{.} \label{djde}%
\end{equation}

This radial action also generates the increment of the azimuthal angle
$\Delta \varphi$ during the transfer from $r_{a}$ to $r_{p}$ and back given by
\begin{equation}
\frac{\Delta \varphi}{2\pi}=n_{\varphi}=-\frac{\partial \mathcal{A}_{r}%
}{\partial \Lambda}. \label{djdl}%
\end{equation}

Both $\tau_{r}$ and $\Delta \varphi$ are clearly two functions of the two
variables $\xi$ and $\Lambda$. In May 1958, Michel H\'{e}non pointed out that
two fundamental potentials, i.e. the Keplerian and harmonic ones, have $\tau_{r}$
which only depends on $\xi$. One year later, in a seminal article in French \cite{Henon58} (for
an English translation see \cite{BinneyHEnon}), he found a family of physical
potentials for which this property remains valid. We propose to complete this
characterization of isochrony by an equivalent property on the azimuthal angle: 
$\Delta \varphi$ only depends on $\Lambda$, see theorem~\ref{thm:characisopot} 
in appendix~\ref{appendix:isocharac}. This family is known as
H\'{e}non's Isochrone. We propose now to follow his steps to recover his
result and eventually extend it by exhibiting the whole set of possible isochrones.

Introducing H\'{e}non's variables,%
\begin{equation}
x=2r^{2}\text{ \ and }Y\left(  x\right)  =x\psi \left(  \sqrt{x/2}\right),
\label{defxY}%
\end{equation}
one can see that the corresponding $x-$values of the periastron and apoastron,
namely $x_{a}$ and $x_{p}$, are the roots of the equation $Y\left(  x\right)
=\xi x-\Lambda^{2}$. As it is detailed in figure~\ref{parabola}, for a fixed
value $\xi$ of the energy, the set of all points $\left(  P_{a,i}%
;P_{p,i}\right)  $ on the lines $y_{i}(x)=$ $\xi x-\Lambda_{i}^{2}$ with
corresponding abscissa $x_{a,i}$ and $x_{p,i}$ form the graph of $Y$. Using a
clever analysis Michel H\'{e}non shows that $\tau_{r}$ only depends on $\xi$
if and only if $P_{0}I$ is proportional to $\left(  x_{p,1}-x_{a\,1}\right)
^{2}$ when $\Lambda^2$ is varying. After a much more involved analysis H\'{e}non was able to prove that
this property characterizes parabolas. This original proof is very technical
and we give a new version of it in theorem \ref{newhenontheo} of 
appendix~\ref{appendixa} highlighting the analytical property of the potentials.

\begin{figure}[h]
\begin{center}
\resizebox{0.8\textwidth}{!}{\includegraphics{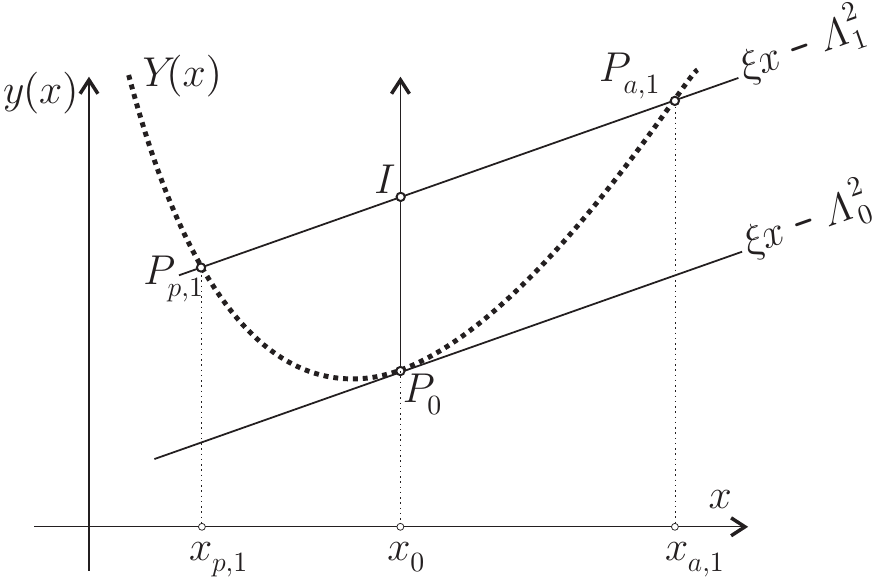}}
\end{center}
\caption{Geometric view of H\'{e}non's variables.}%
\label{parabola}%
\end{figure}

\subsection{General properties of isochrone
parabolas\label{subsec:GalPropIsoParabolas}}

The general equation for a parabola in H\'{e}non's variables is written as%
\begin{equation}
\left(  ax+bY\right)  ^{2}+cx+dY+e=0. \label{parabolaequation}%
\end{equation}
The expressions of the constants $a$, $b$, $c$, $d$ and $e$ in terms of the
problem pa\-ra\-me\-ters are given in the original H\'{e}non paper \cite{Henon58},
where $a$ and $b$ cannot simultaneously vanish. From now on, the
function defined by $\psi : x \mapsto Y(x)/x$ represents an isochrone
potential according to the previous result. Two remarks allow us to simplify
the parabola equation. First, any potential is defined up to a constant
$\epsilon$ which enables us to map $\psi \rightarrow \psi+\epsilon$ or $Y\rightarrow
Y+\epsilon x$ without loss of generality. This transformation is named an
$\epsilon-$\emph{transvection} $\left(  x,Y\right)  \rightarrow \left(x,Y+\epsilon x\right)$.
%\left(
%\begin{array}
%[c]{cc}%
%1 & 0\\
%\alpha & 1
%\end{array}
%\right)  \left(  x,y\right)  $.
 Second, by inspection of equation
\eqref{defE}, one can see that if $\psi$ is isochrone then the potential
$\psi^{\ast}\left(  r\right)  =\psi \left(  r\right)  +j_{\lambda}\left(
r\right)  $ where $j_{\lambda}\left(  r\right)  =\frac{\lambda}{2r^{2}}$ is also
isochrone with a new value of the angular momentum $\Lambda^{\prime2}%
=\Lambda^{2}+\lambda>0$. In terms of $Y$ this allows the transformation
$Y\rightarrow Y^{\ast}=Y+\lambda$. Let us call this translation of the parabola
a $\lambda-$\emph{gauge} transformation of an isochrone potential. The action of
a $\lambda-$gauge or $\epsilon-$transvection could be synthesized in an affine
transformation which is denoted as
\[
J_{\epsilon,\lambda}:%
\begin{array}
[c]{ccc}%
\mathbb{R}^{2} & \rightarrow & \mathbb{R}^{2}\\
(x,Y) & \mapsto & (x,Y+\epsilon x+\lambda) .
\end{array}
\]
If we denote by $\mathbb{A}$ the set of these affine transformations $J_{\epsilon,\lambda}$ and
by observing that $J_{\epsilon,\lambda}\circ J_{\epsilon^\prime,\lambda^\prime} = J_{\epsilon+\epsilon^\prime,\lambda+\lambda^\prime}$,
we see that it is a subgroup of affine transformations of the real plane, isomorphic to $(\mathbb{R}^2,+)$. Affine
transformations in H\'{e}non's variables correspond to physical
transformations which preserve the isochrone property. From this, we arrive at three short lemmas to organize the discussion.

\begin{lemma}
\label{lem1} With a vertical translation $J_{0,\lambda}:\ Y\rightarrow Y^{\ast
}=Y+\lambda$, the H{\'e}non Parabola can be reduced to a non-degenerate parabola
passing through the origin of the $\left(  x,y\right)$-axis.
\end{lemma}

\begin{proof}
According to lemma \ref{lemme0}, $\ell=\lim \limits_{r\rightarrow0}r^{2}%
\psi \left(  r\right)  $ is a real number. By plugging the potential in
equation \eqref{parabolaequation} with the H{\'e}non change of variables, we
get
\[
\left[  2ar^{2}+2br^{2}\psi \left(  r\right)  \right]  ^{2}+2cr^{2}
+2dr^{2}\psi \left(  r\right)  +e=0.
\]
Taking the limit as $r\rightarrow0$ we get%
\begin{equation}
4b^{2}\ell^{2}+2d\ell+e=0. \label{eq:origin}%
\end{equation}
Now, the $\lambda-$translation $Y\rightarrow Y^{\ast}=Y+\lambda$ changes the
potential to $\psi^{\ast}\left(  r\right)  =\psi \left(  r\right)  +\frac
{\lambda}{2r^{2}}$ and hence $\ell^{\ast}=\lim \limits_{r\rightarrow0}r^{2}\psi
^{\ast}\left(  r\right)  =\ell+\frac{\lambda}{2}$. So by taking $\lambda=-2\ell$
we have $\ell^{\ast}=0$. Therefore, according to \eqref{eq:origin}, we have
$e^{\ast}=0$ for the new parabola. In other words, the translated parabola
passes through the origin of the $\left(  x, y\right)$-axis. The degenerate
cases of parabolas, where $a/b$ (resp. $b/a$) is proportional to $c/d$ (resp.
$d/c$) or $d=c=0$ or $a=b=0$, are not of interest in our study since they lead
to constant potentials up to a gauge.\qed

\end{proof}

Considering the result of lemma \ref{lem1}, it is now possible to consider the
asymptotic behavior of the isochrone potential $\psi$ associated with $Y$, which
is given by the relation
\begin{equation}
(A+B\psi)^{2}=\frac{C\psi+D}{2r^{2}}. \label{eq5}%
\end{equation}
Let $\mathcal{D}_{\psi}\subset\mathbb{R}^+$ be the domain on which the potential is defined physically. Then, let us introduce 
\begin{equation}
\mathcal{R}=\displaystyle \sup_{\bar{\mathbb{R}}}\left[  \mathcal{D}_{\psi}\right], \label{definitionR}
\end{equation}
where a priori $\mathcal{R}$ is finite and positive if $\mathcal{D}_{\psi}$ is bounded or $\mathcal{R}=+\infty$ if not. We additionally define $\psi_{\infty}=\displaystyle \lim_{r\rightarrow
\mathcal{R}}\psi(r)$. We now have the following lemma:

\begin{lemma}
\label{lem2} In the context of the above reduction given by \eqref{eq5} we
have the following equivalences: $\psi_{\infty}$ is infinite if and only if
$B=0$ if and only if $\psi$ is harmonic up to an additive constant.
\end{lemma}

\begin{proof}
\begin{itemize}
\item If $B=0$ then, according to lemma \ref{lem1}, $C\neq0$ and from
\eqref{eq5} we get $\psi \left(  r\right)  =2\frac{A^{2}}{C}r^{2}-\frac{D}{C}$.
As we are only interested in increasing potentials, $C$ is positive and
$\psi \left(  r\right)  =$ $\psi_{\mathrm{ha}}\left(  r\right)  =\frac{1}{2}\omega
^{2}r^{2}$ with $\omega \neq0$ --- up to an additive constant. This potential is
defined on $\left[  0,+\infty \right) $ so $\mathcal{R}=+\infty$ and $\psi_{\infty
}=+\infty$.

\item Let us assume conversely that $\psi_{\infty}$ is infinite. As the
potential is increasing, there exists an $r_0$ in the neighborhood of $\mathcal{R}$ such
that for all $r>r_{0}$, $\psi(r)>0$. By dividing equation (\ref{eq5}) by $\psi$ for
$r>r_{0}$, we get%
\[
\frac{1}{\psi}(A+B\psi)^{2}=\frac{1}{2r^{2}}\left(  C+\frac{D}{\psi}\right).
\]
The right hand side of this equality tends to the finite limit $\frac{C}{2\mathcal{R}^{2}}$ when
$r\rightarrow \mathcal{R}$ (that is to zero if $\mathcal{R}=+\infty$). If $B\neq0$, since
$\psi_{\infty}=+\infty$, the left hand side tends to $+\infty$ when
$r\rightarrow \mathcal{R}$. Therefore, $\psi_{\infty}$ infinite implies that $B=0$ and
$\psi$ is harmonic by the first part of this analysis.
\qed

\end{itemize}
\end{proof}

The quantity $\psi_{\infty}$ indicates the asymptotic direction of the
parabola. When $\psi_{\infty}=+\infty$, then
the symmetry axis of the parabola is parallel to $(Oy)$. We do not consider
the case $\psi_{\infty}=-\infty$ because it corresponds to bottom-oriented
parabolas which are always associated with decreasing harmonic potentials
$\psi_{\mathrm{ha}}^{-}\left(  r\right)  =-\frac{1}{2}\omega^{2}r^{2}$. In this case the
effective potential never has global nor local minima and no orbit 
could ever be periodic.

Before exhibiting the isochrone potentials we can say a little more about the tangent to the parabola at the origin.

\begin{lemma}
\label{lem3} For a potential given by \eqref{eq5}, two cases may happen
concerning the tangent at the origin of the isochrone parabola:

\begin{enumerate}
\item It is vertical and the reduced potential is Keplerian up to an additive
constant. This corresponds to $C=0$ in \eqref{eq5}.

\item It is not vertical and modulo a transvection we can manage to get a
horizontal tangent corresponding to $D=0$ in the transvected version of \eqref{eq5}.
\end{enumerate}
\end{lemma}

\begin{proof}
With a gauge transformation we may write the isochrone parabola equation as
$(Ax+BY)^{2}=CY+Dx$. Let us apply to it a transvection with a parameter
$\epsilon$. The new equation is then
\begin{equation}
\left(  A^{\prime}x+B^{\prime}Y\right)  ^{2}=C^{\prime}Y+D^{\prime}x\text{
\ with\  \ }\left \{
\begin{array}
[c]{l}%
A^{\prime}=A+B\epsilon \text{, }C^{\prime}=C\\
B^{\prime}=B\text{ and }D^{\prime}=D+C\epsilon.
\end{array}
\right.  \label{eq:paraPrime}%
\end{equation}
By considering the gradient of the function $f(x,y)=\left(  A^{\prime}%
x+B^{\prime}y\right)  ^{2}-C^{\prime}y-D^{\prime}x$ at the origin, we get the
equation of the tangent to the parabola at the origin, $D^{\prime}%
x+C^{\prime}y=0$. Depending on its direction, two cases may be distinguished:

\begin{enumerate}
\item When $C=0$, the parabola passes through the origin with a vertical
tangent. One may further simplify
the parabolic equation choosing $\epsilon$ to cancel $A^{\prime}$ since $B$
should be non-zero according to lemma~\ref{lem1}. We eventually obtain
$(B^{\prime})^{2}Y^{2}=D^{\prime}x$. This equation implies that $D^{\prime}%
>0$. Making explicit H{\'e}non variables with \eqref{defxY}, we get $\psi \left(
r\right)  =\psi_\mathrm{ke} \left(r\right)  =-\frac{\mu}{r}$ where $\mu=\sqrt{\frac{D^{\prime}}{2B^{\prime2}}}$
is a positive constant since $r\mapsto \psi \left(  r\right)  $ is increasing.

\item When $C\neq0$, it is possible to choose the parameter $\epsilon$ of the
transvection to cancel $D^{\prime}$ so that the parabola passes through the
origin with a horizontal tangent. In other words, we choose the arbitrary
constant of the potential to simplify the study of its corresponding parabola,
which may be described by $\left(  A^{\prime}x+B^{\prime}Y\right)
^{2}=C^{\prime}Y$ with $A^{\prime}\neq0$. $A^{\prime}$ cannot vanish unless
$\epsilon=-\frac{A}{B}=-\frac{D}{C}$ which is forbidden by lemma~\ref{lem1}.
\qed

\end{enumerate}
\end{proof}

Let us summarize the situation at this point (see figure~\ref{fig:redpara}).

Any parabola in the plane $\left(  x,y\right)  $ is associated with an isochrone
potential. Combining lemmas \ref{lem1}, \ref{lem2} and \ref{lem3} we can only
study the family of parabolas passing through the origin and belonging to one
of the two following classes:

\begin{itemize}
\item Straight parabolas, which possess a vertical symmetry axis and thus never
admit any vertical tangent. As we have explained before we are only interested
by straight up-oriented parabolas. Using affine transformations, straight
parabolas could be moved in such a manner that their apices are the origin of
the $\left(  x,y\right)$-plane. They correspond to harmonic potentials.

\item Tilted parabolas, whose symmetry axes are inclined from the vertical ones
and possess a horizontal or vertical tangent at the origin. This tilted parabola family is composed of three categories:

\begin{itemize}
\item Laid parabolas, with a vertical tangent at the origin corresponding to Kepler potentials;\label{laidpara}

\item Right-oriented parabolas, with a horizontal tangent at the origin;

\item Left-oriented parabolas, with a horizontal tangent at the origin.

\end{itemize}
\end{itemize}

In figure \ref{fig:redpara}, we have plotted in the $(xOy)$ plane the four \emph{reduced} classes of parabolas. A precise definition of the corresponding potentials is given in definition~\ref{def:redphysgau} p.\pageref{def:redphysgau}.

\begin{figure}[ptb]
\centering
\resizebox{0.9\textwidth}{!}{\includegraphics{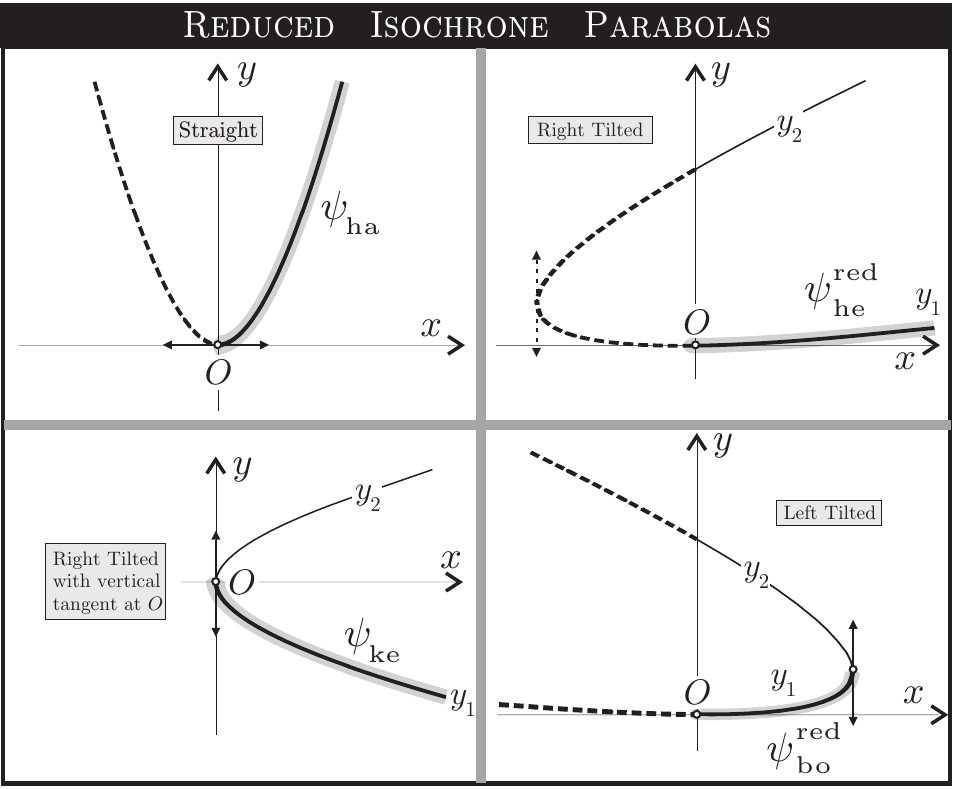}}
\centering
\caption{The four classes of reduced parabolas corresponding to the reduced
isochrone potentials. The part of the parabola associated with the increasing
potential is highlighted ($y_{1}$). The dashed part of the parabola corresponds
to potentials with an imaginary distance argument ($x<0$). The unhighlighted
solid line part of the parabola ($y_{2}$) in the $x>0$ half-plane corresponds
to decreasing potentials.}%
\label{fig:redpara}%
\end{figure}

The reduced isochrone potential contained in each reduced parabola is
emphasized in this figure and it corresponds to a limited part of the parabola.
As a matter of fact, the variable of the potential is the radial distance,
a positive real number. Each isochrone potential is then included in
the $x$--positive right plane. This remark excludes left-oriented laid
parabolas. For any non-straight parabolas there are two functions $x\mapsto
y_{1}\left(  x\right)  $ and $x\mapsto y_{2}\left(  x\right)$ into which the
$x $--positive part of the graph of the parabola could be decomposed (see
figure \ref{fig:redpara}). The slope of the chord between the origin and a point
$M$ of abscissa $x>0$ on the graph of $y_{1}$ or $y_{2}$ is given by the ratio
$\frac{y_{1}\left(  x\right)  }{x}$ or $\frac{y_{2}\left(  x\right)  }{x}$
which is precisely the definition of the potential $\psi$.
This remark shows that $\psi$ is an increasing (resp. decreasing) function if
the graph of $y$ is convex (resp. concave), i.e. the chord between two points
is above (resp. below) the function. As we look for increasing potentials in
order to have \textsc{pro}'s, we have to consider the convex part of the
parabola graph. This part is named $y_{1}$ in figure \ref{fig:redpara}.

Tilted parabolas have a symmetry axis with a finite slope. Any $\epsilon-$
transvection adds $\epsilon$ to this slope, modifying the orientation of these
parabolas. Nevertheless, we cannot jump from a left-oriented parabola to a
right-oriented one using an affine transformation. However, according to lemma
\ref{lem3} and conserving its orientation, we can morph any tilted parabola
with a horizontal tangent or a vertical tangent at the origin. In the latter
case, the symmetry axis is parallel to $\left(  Ox\right)  $. The morphing
from the reduced parabolas to the whole set of isochrone ones is detailed in
figure \ref{allfour} following our analysis of the concerned potentials in the
next section.

Our reduction to four families of parabolas and their corresponding potentials
enables us to obtain the whole set of isochrone potentials. In his historical
study, Michel H\'{e}non did not remark on the crucial role of these affine
transformations. He dismissed out-of-origin parabolas and forgot
left-oriented tilted ones.

Let us now determine explicitly the isochrone potentials of the reduced families.

\subsection{Classification of isochrone potentials
\label{subsec:isochroneclassification}}

From the previous analysis we will now state and prove the following
classification result.

\begin{theorem}
\label{theo1}The isochrone potentials are classified by these two properties:
\begin{enumerate}
	\item There are essentially four types of reduced isochrone potentials:
	\begin{itemize}
		\item The Keplerian potential $\psi_{\mathrm{ke}}$ for which the reduced parabola
					has a horizontal symmetry axis and a vertical tangent at the origin.
		\item The harmonic potential $\psi_{\mathrm{ha}}$ for which the reduced parabola is
					straight with a horizontal tangent at the origin.
		\item Two other potentials $\psi_{\mathrm{he}}^{\mathrm{red}}$ and $\psi_{\mathrm{bo}}^{\mathrm{red}}$
					for which the reduced parabolas have horizontal tangents at the origin and are
					respectively right and left oriented. They are given by the formulae
					\[
					\psi_{\mathrm{he}}^{\mathrm{red}}:= 
					\frac{\mu}{2b} -\frac{\mu}{b+\sqrt{b^{2}+r^{2}}},
					\hspace{1em}\psi_{\mathrm{bo}}^{\mathrm{red}}:=
					-\frac{\mu}{2b} +\frac{\mu}{b+\sqrt{b^{2}-r^{2}}},
					\]
					where $\mu$ and $b$ are positive constants.
		\end{itemize}
	\item Any isochrone potential $\psi$ is equivalent under an affine
	transformation to one of the previous types. That is to say there exist
	two constants $\epsilon$ and $\lambda$ and some reduced potential $\psi^{\mathrm{red}}
	\in \{ \psi_{\mathrm{ke}},\psi_{\mathrm{ha}},\psi_{\mathrm{he}}^{\mathrm{red}}%
	,\psi_{\mathrm{bo}}^{\mathrm{red}}\}$ such that $\psi \left(  r\right)  =%
	\psi^{\mathrm{red}}\left(  r\right)  +\epsilon+\frac{\lambda}{2r^{2}}$.
	\end{enumerate}
\end{theorem}

The potential $\psi_{\mathrm{he}}$ is the original potential discovered by
Michel H\'{e}non. From our knowledge, the potential $\psi_{\mathrm{bo}}$
is a new one. We call it the \emph{bounded} potential for reasons appearing in sec. \ref{thenewsection24}.
\begin{proof}
Let $\mathcal{P}$ be the parabola associated with an isochrone potential $\psi$
which is neither Keplerian nor harmonic. According to lemmas \ref{lem1},
\ref{lem2} and \ref{lem3} we are left to consider the case where $\mathcal{P}$
passes through the origin, has a horizontal tangent and has a symmetry axis which
is not vertical. According to \eqref{eq5} and the previous lemmas, this
corresponds to having an equation of the form
\[
2r^{2}=\frac{n\psi}{(\psi-m)^{2}},
\]
for some constant $m$ and $n$ both non zero. As a consequence, we see here
that the potential $\psi \ $will depend on two constants. Normalizing
the potential by setting $\psi=mV$, we are led to the functional equation
\[
q=q(V)=\frac{V}{(V-1)^{2}},\hspace{1em}\text{{with}}\hspace{1em}q=\kappa
x=2\kappa r^{2},
\]
where $\kappa=m/n$ is another non zero constant. The inversion of the
function $q$ gives two solutions $V(q)$ of the quadratic equation
\begin{equation}
qV^{2}-(2q+1)V+q=0 . \label{eq13ps}%
\end{equation}
They are of the form
\begin{equation}
\left \{
\begin{array}
[c]{ll}%
V^{+}(q) & :=\frac{2q+1-\sqrt{4q+1}}{2q}=1-\frac{2}{1+\sqrt{4q+1}},\\
V^{-}(q) & :=\frac{2q+1+\sqrt{4q+1}}{2q}=1-\frac{2}{1-\sqrt{4q+1}}.
\end{array}
\right.  \label{VplusVmoins}%
\end{equation}

\begin{figure}[h]
\centering
\resizebox{0.75\textwidth}{!}{\includegraphics{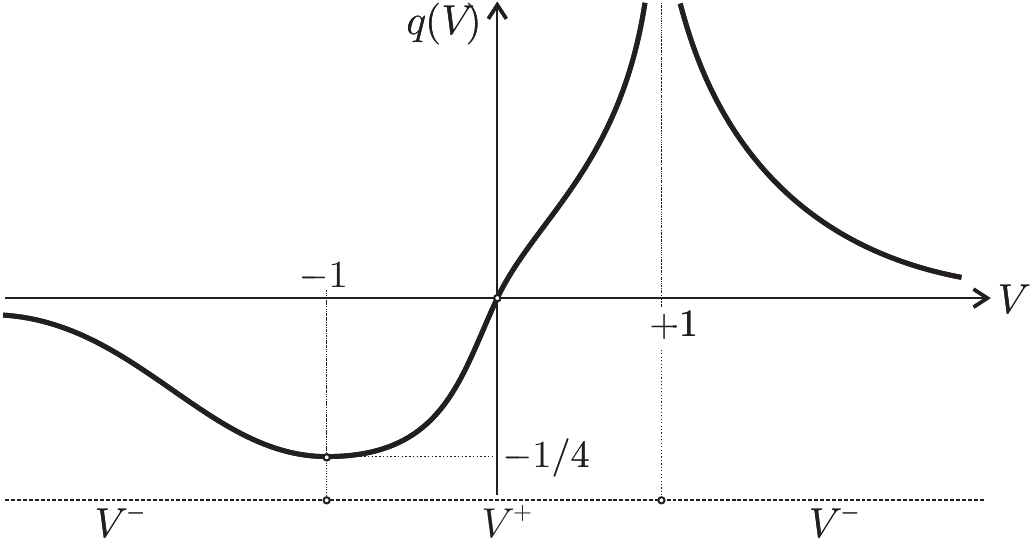}}
\centering
\caption{The potential $V^+$ and $V^-$ as functions of $q$.}%
\label{vdeq}%
\end{figure}

These two functions $q\mapsto V(q)$ are defined on the real
interval $q\geqslant-1/4$. As shown in figure \ref{vdeq}, the $\pm$ signs of
$V$ are chosen in such a way that $q\mapsto V^{+}(q)$ is increasing on
$[-1/4,+\infty )$ and $q\mapsto V^{-}(q)$ is decreasing on both $[-1/4,0)$ and
$]0,+\infty )$. From the quadratic equation \eqref{eq13ps} we have 
\begin{equation}
V^{+}(q)+V^{-}(q)=2+\frac{1}{q}\  \text{\ and}\hspace{1em}V^{+}(q)V^{-}(q)=1.
\label{eq:somme}%
\end{equation}
Now we compute the potential. From the expression $q=\kappa x=2\kappa r^{2}$, we will classify the potentials by the sign of the constant $\kappa$.

\begin{enumerate}
\item {\bf When $\mathbf{\kappa<0}$}, then $q$ is necessarily negative. Therefore
 $-1/4\leqslant
q\leqslant0$ which implies
\[
 r^{2}\leqslant \frac{1}%
{8|\kappa|}.%
\]
Setting a new constant $b:=\frac{1}{\sqrt{8|\kappa|}}$ in such a way
that the previous inequality becomes $r\leqslant b$, we have $q=\frac{-r^{2}%
}{4b^{2}}$. Therefore,
\[
\psi \left(  r\right)  =mV(q)=mV\left(  \tfrac{-r^{2}}{4b^{2}}\right)  .
\]
This gives us two possible potentials $\psi^{\pm}$. But to have a \textsc{pro},
the function $r\mapsto \psi(r)$ must be ultimately increasing. That is,
$$\frac{-mr}{2b^{2}}\frac{dV}{dq}\left(  \frac{-r^{2}}{4b^{2}}\right)  >0.$$
Since $q\mapsto V^{+}(q)$ is increasing we must choose $m=-\frac{\mu}{2b}$ for
some positive constant $\mu$. The factor $\frac{1}{2b}$ is just here for simplicity of the
final result. Similarly in the formula for $\psi^{-}$ we must
choose $m=\frac{\mu}{2b} >0$. This leads to the two potentials
\[
\left \{
\begin{array}
[c]{lll}%
\psi_{\mathrm{bo}}^{+}(r) & :=-\frac{\mu}{2b} V^{+}\left(  \frac{-r^{2}%
}{4b^{2}}\right)  & =-\frac{\mu}{2b}+\frac{\mu}{b+\sqrt{b^{2}-r^{2}}},\\
\psi_{\mathrm{bo}}^{-}(r) & :=\frac{\mu}{2b} V^{-}\left(  \frac{-r^{2}%
}{4b^{2}}\right)  & =\frac{\mu}{2b}-\frac{\mu}{b-\sqrt{b^{2}-r^{2}}}.
\end{array}
\right.
\]
From \eqref{eq:somme}, we get that
\begin{equation}
\psi_{\mathrm{bo}}^{-}(r)-\psi_{\mathrm{bo}}^{+}(r)=\frac{\mu}%
{2b}\left[  V^{+}\left(  \tfrac{-r^{2}}{4b^{2}}\right)  +V^{-}\left(
\tfrac{-r^{2}}{4b^{2}}\right)  \right]  =\frac{\mu}{b} -\frac{2b\mu}{r^{2}}.
\label{psibo+-}%
\end{equation}
As a consequence, the left-oriented parabolas associated with $\psi_{\mathrm{bo}}^{+}$ and
$\psi_{\mathrm{bo}}^{-}$ are exchanged by an affine transformation. This is
the meaning of the word \emph{essentially} in the statement of the theorem
since the group orbits of $\psi_{\mathrm{bo}}^\mathrm{red}:=\psi_{\mathrm{bo}}^{+}$ and of
$\psi_{\mathrm{bo}}^{-}$ under the action of the affine group are the same.

\item {\bf When $\mathbf{\kappa>0}$}, setting $b:=1/\sqrt{8\kappa}$ again, we similarly get
$\psi=mV(q)=mV\left(  \frac{r^{2}}{4b^{2}}\right)  $. And by setting again
$\frac{\mu}{2b} :=|m|$ we get two new isochrone potentials
\[
\left \{
\begin{array}
[c]{lll}%
\psi_{\mathrm{he}}^{+}(r) & :=\frac{\mu}{2b}V^{+}\left(  \frac{r^{2}%
}{4b^{2}}\right)  & =\frac{\mu}{2b}-\frac{\mu}{b+\sqrt{b^{2}+r^{2}}},\\
\psi_{\mathrm{he}}^{-}(r) & :=-\frac{\mu}{2b}V^{-}\left(  \frac{r^{2}%
}{4b^{2}}\right)  & =-\frac{\mu}{2b} +\frac{\mu}{b-\sqrt{b^{2}+r^{2}}}.
\end{array}
\right.
\]
Again from \eqref{eq:somme}, we have that
\begin{equation}
\psi_{\mathrm{he}}^{+}\left(  r\right)  -\psi_{\mathrm{he}}^{-}\left(
r\right)  =\frac{\mu}{2b} \left[  V^{+}\left(  \tfrac{r^{2}}{4b^{2}}\right)
+V^{-}\left(  \tfrac{r^{2}}{4b^{2}}\right)  \right]  =\frac{\mu}{b}
+\frac{2b\mu}{r^{2}}. \label{psihe+-}%
\end{equation}
Therefore, $\psi_{\mathrm{he}}^\mathrm{red}:=\psi_{\mathrm{he}}^{+}$ and
$\psi_{\mathrm{he}}^{-}$ are also exchanged by the affine group and their
respective group orbits under this group action will coincide. These potentials are
defined for all values of $r\in[0,+\infty)$ so that their parabolas are then right-oriented.
\end{enumerate}

To conclude the proof of the theorem we only have to observe that according to
lemmas \ref{lem1}, \ref{lem2} and \ref{lem3}, any isochrone is in the orbit of
a reduced one under the affine subgroup generated by the $J_{\epsilon,\lambda}$ tranformations.\qed
\end{proof}

These reduced potentials can be visualized in figure~\ref{fig:redpara}.

Using natural notations taken from the proof of theorem \ref{theo1},
from \eqref{psibo+-} and \eqref{psihe+-} we can write%
\begin{equation}
\  \  \text{\ }\left \{
\begin{array}
[c]{l}%
\psi_{\mathrm{bo}}^{-}=J_{+\epsilon,\lambda}\left(  \psi_{\mathrm{bo}}%
^{+}\right) \\
\\
\psi_{\mathrm{he}}^{-}=J_{-\epsilon,\lambda}\left(  \psi_{\mathrm{he}}%
^{+}\right)
\end{array}
\right.  \text{\  \ with\ }\epsilon=\frac{\mu}{b} \text{ and }\lambda
=-4b\mu . \label{property+-}%
\end{equation}

The tilted parabolas presented in figure~\ref{fig:redpara} are the ones associated
with $\psi_{\mathrm{he}}^{+}$ for the right (called $\mathcal{P}%
_{\mathrm{he}}^{+}$) parabola and with $\psi_{\mathrm{bo}}^{+}$ for the
left one (called $\mathcal{P}_{\mathrm{bo}}^{+}$). These two parabolas 
both open to the top, i.e. in the direction where $y$ increases. Using 
property \eqref{property+-} or by direct representation, one can verify that,
using natural notations, $\mathcal{P}_{\mathrm{he}}^{-}$(resp.
$\mathcal{P}_{\mathrm{bo}}^{-}$) is the image of $\mathcal{P}%
_{\mathrm{he}}^{+}$(resp. $\mathcal{P}_{\mathrm{bo}}^{+}$) by the
symmetry under the $\left(  O,x\right)$-axis. Thus, these two ``negative"
parabolas both open to the bottom.

\subsection{Some physical meaning of this classification\label{thenewsection24}}

The potential $\psi_{\mathrm{he}}^{\mathrm{red}}$ defined by $$\psi_{\mathrm{he}}^{\mathrm{red}}(r):=\frac{\mu}{2b}-\frac{\mu}{b+\sqrt{b^2+r^2}}$$
%$$
%\psi_{\mathrm{he}}\left(  r\right):=-\frac{\mu}{b+\sqrt{b^2+r^2}}=\psi_{\mathrm{he}}^{\mathrm{red}}(r)-\frac{\mu}{2b}
%$$
is the original isochrone potential discovered by Michel H\'{e}non.
%Similarly setting
%$$
%\psi_{\mathrm{bo}}\left(  r\right):=\frac{\mu}{b+\sqrt{b^2-r^2}}=\psi_{\mathrm{bo}}^{\mathrm{red}}(r)+\frac{\mu}{2b}
%$$
%we get 
Similarly, the potential $$\psi_{\mathrm{bo}}^{\mathrm{red}}:=-\frac{\mu}{2b}+\frac{\mu}{b+\sqrt{b^2-r^2}}$$ defines another type of isochrone potential.
The index \emph{bo} means \emph{bounded potential}. Indeed, from the above formula the mappings $r\mapsto \psi_{\mathrm{bo}}^{\mathrm{red}}(r)+\epsilon$ %+\frac{\lambda}{2r^2}$ 
are only defined for bounded values of 
\begin{equation}
r\in\mathcal{D}_{\psi_{\mathrm{bo}}}=[0,b]. \label{newdefrbo}
\end{equation}

The fact that such potentials are associated with a \emph{bounded} system 
give them special features which are very different from the three other types of isochrone potentials. Up to our knowledge, such bounded potentials do not seem to have appeared in the literature before.

The potential of Michel H\'{e}non is equivalent to a Keplerian one when $r\to\infty$.
Using relation \eqref{VplusVmoins}, we can readily see that $V^{+}(q)\sim
q$ when $q\rightarrow0$. The roots product in \eqref{eq:somme} then
implies $V^{-}(q)\sim q^{-1}$ in the same limit. Then both $\psi_{\mathrm{bo}}^{\mathrm{red}}$ and
$\psi_{\mathrm{he}}^{\mathrm{red}}$ come from $V^{+}$. Hence we can assert that they
are harmonic near their center: $\psi_{\mathrm{bo}}^{\mathrm{red}}\sim \psi_{\mathrm{he}}^{\mathrm{red}}
\sim \psi_{\mathrm{ha}}$ when $r\rightarrow0$. 
\\

From a physical point of view, $\epsilon$--transvections $J_{\epsilon,0}:\psi\to\psi+\epsilon$ add a constant to 
the potential, hence they do not change anything for the dynamics in the considered potential, changing only
the value of the energy attributed to the trajectories.

When the $\lambda$--gauge  $J_{0,\lambda}:\psi\to\psi+\frac{\lambda}{2r^2}$ is applied to a reduced potential, it makes it divergent as
$r^{-2}$ when $r\rightarrow0$. As we said at the beginning of sec.
\ref{subsec:GalPropIsoParabolas}, such transformations correspond to a change
of the value of the angular momentum in the corresponding isochrone orbit.

Geometrically, when the physical convex part of the parabola starts from the origin, then,
when $r\rightarrow0$, the corresponding potential is finite (if it is
$\psi_{\mathrm{bo}}$, $\psi_{\mathrm{he}}$ or $\psi_{\mathrm{ha}} $)
or diverges like $r^{-1}$ (if it is Keplerian). This behavior is not perturbed
by $\epsilon-$transvections. In all other cases isochrone potentials diverge
like $r^{-2}$ when $r\rightarrow0$; but, using a $\lambda-$translation, we can
manage this physical problem. 

These remarks enable us to define three classes of
isochrone potentials. They are classes of equivalence under the action
of $J_{\epsilon,\lambda}$ affine transformations as detailed in sec.~\ref{subsec:affinegroupaction}. 
Definition \ref{def:redphysgau} sets their names in addition to
the name of the four isochrone potential types.

\begin{definition} 

\label{def:redphysgau}
\begin{enumerate}
\item We call the four isochrone potentials
\[%
\begin{array}
[c]{ll}%
\psi_{\mathrm{ke}}\left(  r\right)  =-\dfrac{\mu}{r}, & \phantom{and\ } \psi_{\mathrm{ha}%
}\left(  r\right)  =\frac{1}{2}\omega^{2}r^{2},\\
& \\
\psi_{\mathrm{he}}\left(  r\right)  =-\dfrac{\mu}%
{b+\sqrt{b^{2}+r^{2}}}, \ \ & and \ \psi_{\mathrm{bo}}(r)=\dfrac{\mu
}{b+\sqrt{b^{2}-r^{2}}},%
\end{array}
\text{ }%
\]
the Kepler, the harmonic, the H{\'e}non and the bounded potential, respectively. 
\item We call \emph{reduced isochrone} potentials
$\psi_{\mathrm{iso}}^{\mathrm{red}}$ one of the four potentials
$$\psi_{\mathrm{ke}}, \  
\psi_{\mathrm{ha}}, \  
\psi_{\mathrm{he}}^\mathrm{red}=J_{\frac{\mu}{2b},0}\left( \psi_{\mathrm{he}} \right)=\frac{\mu}{2b}+\psi_{\mathrm{he}}
\text{ or }  
\psi_{\mathrm{bo}}^\mathrm{red}=J_{-\frac{\mu}{2b},0}\left( \psi_{\mathrm{bo}} \right)=-\frac{\mu}{2b}+\psi_{\mathrm{bo}}.$$
\item We call \emph{physical isochrone} potentials $\psi_{\mathrm{iso}%
}^{\mathrm{phy}}$ the result of a transvection applied to a reduced isochrone: $\psi_{\mathrm{iso}}^{\mathrm{phy}}=J_{\epsilon,0}\left(  \psi_{\mathrm{iso}%
}^{\mathrm{red}}\right) = \psi_{\mathrm{iso}%
}^{\mathrm{red}} + \epsilon $. 
\item We call \emph{gauged isochrone}
potentials $\psi_{\mathrm{iso}}^{\mathrm{gau}}$ the result of a vertical
translation applied to a physical isochrone: $\psi_{\mathrm{iso}%
}^{\mathrm{gau}}=J_{0,\lambda}\left(  \psi_{\mathrm{iso}}^{\mathrm{phys}}\right)
=\psi_{\mathrm{iso}}^{\mathrm{phys}}+\frac{\lambda}{2r^2}
$.
\end{enumerate} 
\end{definition}

Physical isochrones possess interesting physical properties which are presented and studied in section \ref{sec4}. They all confine a finite mass in a finite radius $r<\mathcal{R}$ (see equation \eqref{definitionR}, page \pageref{definitionR}). Reduced isochrones are special cases of physical ones: their parabolas pass through the origin with a horizontal or a vertical tangent. 

Due to their $r^{-2}$ divergence in the potential, when $\lambda\ne0$, gauged isochrone potentials have an infinite mass at their center and thus possess poor physical meaning. However, they are essential to the completeness of the isochrone set. 

\subsection{The affine group action on the Isochrone set}
\label{subsec:affinegroupaction}

Let us denote respectively $\mathbb{I}_{\mathrm{pot}}$ and $\mathbb{I}_{\mathrm{par}}$ the set of isochrone potentials and parabolas. These two sets are in bijection and theorem \ref{theo1} states that, from a mathematical point of view, they are four-dimensional manifolds. As a matter of fact, each isochrone potential is uniquely determined by four real parameters $(\mu, b,\epsilon,\lambda)$ with $\mu>0$, $b\geq0$ and $(\epsilon,\lambda)\in\mathbb{R}^2$ --- n.b. for $\psi_\mathrm{bo}$,  $b>0$.

We have also seen that the two-dimensional affine group $\mathbb{A}\simeq(\mathbb{R}^2,+)$, generated by the affine
transformations $J_{\epsilon,\lambda}$ with $(\epsilon,\lambda)\in\mathbb{R}^2$, acts on both sets, either on potentials or on the cor\-res\-pon\-ding parabolas. Since the dimension of $\mathbb{A}$ is less than the dimension of $\mathbb{I}_{\mathrm{pot}}$ and $\mathbb{I}_{\mathrm{par}}$ ($2<4$), the action is not transitive and each group orbit $\mathbb{A}\cdot\psi$ or $\mathbb{A}\cdot\mathcal{P}$ for corresponding potential $\psi$ or parabola $\mathcal{P}$ is a two-dimensional sub-manifold of $\mathbb{I}_{\mathrm{pot}}$ or $\mathbb{I}_{\mathrm{par}}$. 
This translates the second part of 
theorem~\ref{theo1}: we have four types of group orbits under the action of $\mathbb{A}$, one for each type of isochrone potential.
\\

Let us now see more precisely this action of the affine group and its corresponding group orbits.

\begin{figure}[h]
\centering \resizebox{0.97\textwidth}{!}{\includegraphics{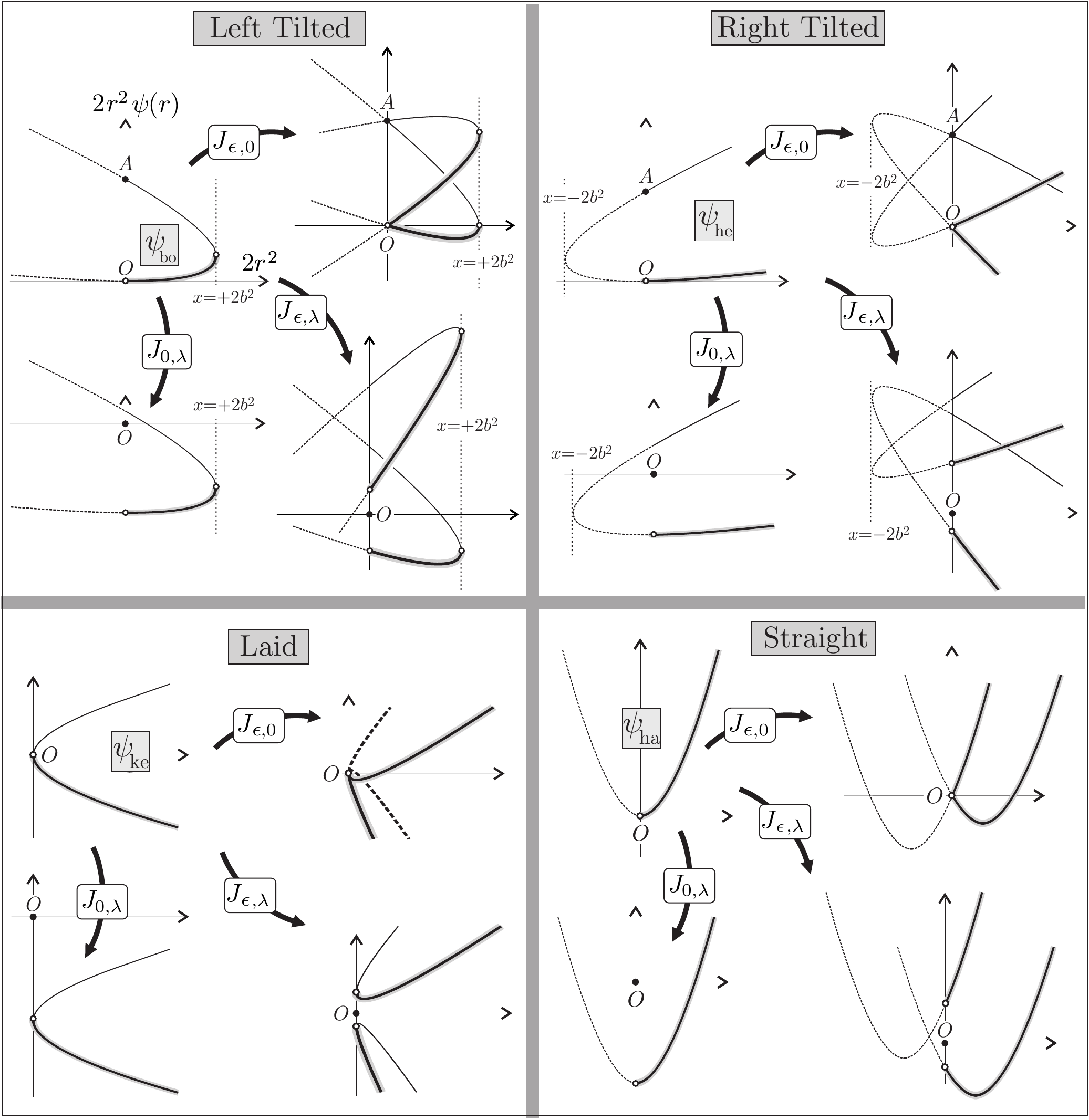}}
\centering
\caption{The action of affine transformations on reduced parabolas and their
corresponding potentials.}%
\label{allfour}%
\end{figure}

Each parabola is associated with an isochrone potential and vice-versa. Each
isochrone parabola belongs to one of the four classes of reduced parabolas we
have presented in figure~\ref{fig:redpara} and is associated with one of the four reduced isochrone
potentials made explicit in theorem~\ref{theo1}. In order to geometrically understand the morphing of parabolas
associated with the action of affine transformations, we propose the general
picture of figure \ref{allfour}.
 
We do not represent in this figure either the bottom-oriented straight
parabolas or the left-oriented laid ones because they respectively correspond to
decreasing and non-physical potentials. We specify that it is always possible to
have $\left(  Oy\right)$-axis crossing the parabola: this corresponds to a
horizontal translation of the parabola associated with a good choice of the
origin of the physical referential.

Transvections correspond to $J_{\epsilon,0}$. They are associated with a
swivel combined with a deformation of the parabola: the points of the parabola
lying on the $\left(  Oy\right)$-axis are invariant as is the abscissa of the
vertical tangent. 

General affine transformations $J_{\epsilon,\lambda}$ swivel, deform and translate
a reduced parabola. They affect both the energy and the angular momentum of
the considered isochrone orbit. Any parabola obtained from the action of
$J_{\epsilon,\lambda}$ on a reduced one corresponds to an isochrone potential
in the same group orbit of the reduced potential under the action of the
Affine Group. In this sense we can claim that there are only four different
isochrone potentials up to an affinity on its parabola.  
\\

We note that relations exist between the isochrone potentials. As a
matter of fact, $\psi_{\mathrm{ke}}$ and $\psi_{\mathrm{ha}}$ come from
$\psi_{\mathrm{he}}$ when $b\rightarrow0$ and $b\rightarrow+\infty$,
respectively. \label{remark:isochronelink}
Furthermore, known relations exist between $\psi_{\mathrm{ke}}$ 
and $\psi_{\mathrm{ha}}$, such as the Bohlin
transformation~(\cite{Bohlin}, \cite{ArnoldBol}, \cite{LB2}, see also the footnote~\ref{footnote:bohlin} p.\pageref{footnote:bohlin}) which maps the
harmonic orbits onto Keplerian ones and vice versa. All these relations are not
in the scope of the affine group action and do not affect the parameters
$(\mu,b)$ or $\omega$ of the concerned potentials.

Nevertheless, making use of rotations $R_{\theta}$ of an angle $\theta$ in the
$\left(  x,y\right)$-frame and starting, for instance, from the laid Kepler
parabola, we can obtain a new parabola with an arbitrarily oriented axis of
symmetry. Then, acting with $J_{\epsilon,\lambda}$, we can recover the corresponding
reduced parabola in one of the four families. This operation is graphically
illustrated in figure~\ref{kepler2henon} in the case of the morphing from the
Kepler isochrone to the H\'{e}non isochrone.

\begin{figure}[h]
\centering
\resizebox{0.9\textwidth}{!}{\includegraphics{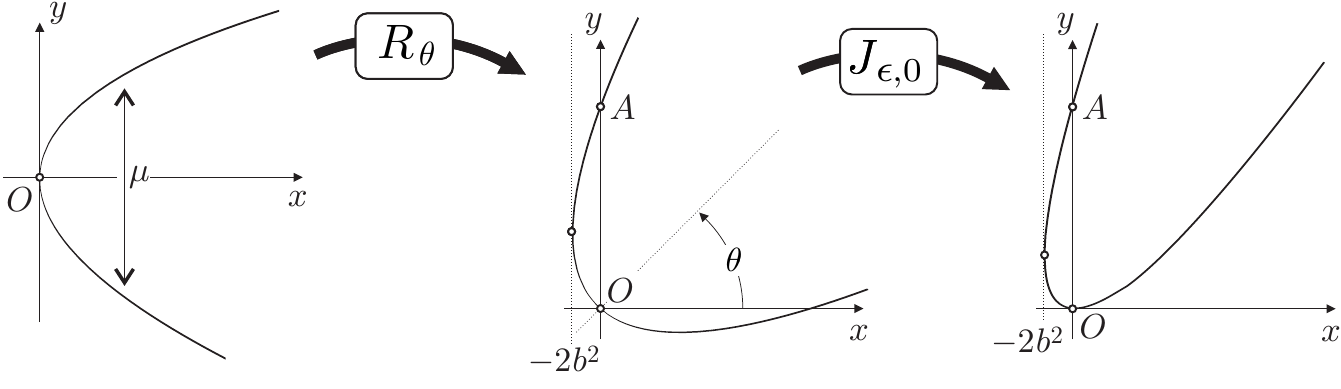}}
\centering
\caption{Rotation and transvection of the Kepler parabola to the H{\'e}non one.}%
\label{kepler2henon}%
\end{figure}

Varying the unique parameter $\mu$ of the Kepler potential, written as
\begin{equation}
y_{\mathrm{ke}}\left(  x\right)  =-\mu \sqrt{2x} \label{y_x_kepler}
\end{equation}
 in H\'{e}non's variables, the aperture of its laid parabola varies and produces a
variation of the $b$ parameter of the H\'{e}non potential corresponding to the
negative part of the rotated parabola. Using this process one can easily
understand that when $\theta \in \left(  -\frac{\pi}{2},+\frac{\pi}{2}\right)$, we recover a H{\'e}non potential $\psi_{\mathrm{he}}$; when $\theta
\in \left(  \frac{\pi}{2},+\frac{3\pi}{2}\right)$, we recover the bounded
potential $\psi_{\mathrm{bo}}$; and for $\theta=+\frac{\pi}{2}$, we obtain the
harmonic $\psi_{\mathrm{ha}}$ from the Kepler potential $\psi_{\mathrm{ke}}$.
More generally, any isochrone potential is contained
in the group orbit of a Kepler potential under the action of the group $\textsc{SO}(2)\ltimes \mathbb{A}$.

As we have completely classified $\mathbb{I}_{\mathrm{pot}}$ and $\mathbb{I}_{\mathrm{par}}$, we can now
return to the study of isochrone \textsc{pro}'s. We will see that relevant isochrone rotations are not Euclidian but hyperbolic.

\section{Isochrone orbits and isochrone transformations\label{sec3}\qquad}

\subsection{Introduction and motivation\label{subsec_LBidea}}

From the geometrical classification of the isochrone potentials
established through the action of the Affine Group in section~\ref{sec2}, we
propose now to investigate isochrone orbits\footnote{When not specified, orbit refers to the trajectory of the considered 
test particle in the considered potential and no more to the group orbit of a potential or parabola under a group action.}.

For this purpose, we generalize a transformation that originates in the work of Newton, Karl Bohlin~\cite{Bohlin} and Donald Lynden-Bell~\cite{LyndenBell} who recently passed away and to whom we dedicate this work.
He explored a remarkable property of Michel H{\'e}non's isochrone, namely $\psi_\mathrm{he}$ is equivalent to a harmonic potential 
at small distances and to a Keplerian potential at larger ones, see section~\ref{sec4}. In those two potentials, the orbits are closed ellipses.
Newton showed, in the later edition of the Principia, how to map a Keplerian elliptical orbit onto a harmonic one and vice versa. His methods relied on a
total exchange of energy and potential between a Kepler and a Hooke system. Pointing out a freedom that involves 
partial exchange of energy and potential, Donald Lynden-Bell derived H{\'e}non's isochrone as a convex interpolation of Kepler and Hooke potentials. 
Let us detail now their mathematical analysis and generalize it to isochrone orbits transformations.

\subsubsection{Isochrone orbits transformations}

A \textsc{p}eriodic \textsc{r}adial \textsc{o}rbit (\textsc{pro}) $r_{0}%
(t_{0})$ in a radial potential $\psi_{0}$ is governed by the ordinary
differential equation
\[
\frac{1}{2}\left(  \frac{dr_{0}}{dt_{0}}\right)  ^{2}+\frac{\Lambda_{0}^{2}%
}{2r_{0}^{2}}=\xi_{0}-\psi_{0}\left(  r_{0}\right)  \text{.}%
\]
In H\'{e}non variables $x_{0}=2r_0^2$ and $y_{0}=x_0\psi_0(x_0)$, it can be written as
\begin{equation}
\displaystyle { \frac{1}{16}\left(  \frac{dx_{0}}{dt_{0}}\right)  ^{2}+\Lambda_{0}^{2}%
=x_{0}\xi_{0}-y_{0}\left(  x_{0}\right)  \text{.} } \label{fundamentalode}%
\end{equation}
Since the force derived from a radial potential is radial, the motion of a test particle takes place in a fixed plane and this particle is described by its polar coordinates $(r_0,\varphi_0)$ in this plane. 
 
When the potential is isochrone, $y_{0}$ is a parabola. This property is
preserved by linear transformations of parabolas (see lemma
\ref{lemme_lineaire} in appendix~\ref{appendixlemma}) and consequently for the
orbits they contain. Placing the origin of the $(x,y)$-plane at the center of
the system described by $\psi_{0}$,
linear transformations relate isochrone orbits together. There exists then a
change of variables $(r_{0},t_{0})\mapsto(r_{1},t_{1})$ mapping an isochrone
orbit onto another one that satisfies an orbital equation in the new potential
$y_{1}$, i.e.
\begin{equation}
\frac{1}{16}\left(  \frac{dx_{1}}{dt_{1}}\right)  ^{2}+\Lambda_{1}^{2}%
=x_{1}\xi_{1}-y_{1}\left(  x_{1}\right)  \text{.} \label{fundamentalodeimg}%
\end{equation}
As Donald Lynden-Bell explained (\cite{LB2} sect. 3 or \cite{LyndenBell} sect. 2), it
is convenient to study orbits of identical angular momentum%
\begin{equation}
\Lambda_{1}=\Lambda_{0}=\Lambda \text{.} \label{lambdaconservation}%
\end{equation}
This hypothesis allows one to get the same Kepler's area law for both orbits, in their
respective radial potentials.

At this point, no constraints specify how each of the three remaining terms  
$\left(  \frac{dx_{0}%
}{dt_{0}}\right)  ^{2}$, $x_{0}\xi_{0}$, and
$y_{0}$ in~\eqref{fundamentalode}
is transformed in the
mapping. For instance, the Bohlin transformation (see \cite{Bohlin} for the
original reference or \cite{ArnoldBol} for a modern presentation) consists of
a full exchange between energy and potential terms. As underlined by
Lynden-Bell, the exchange can also be partial: only part of the potential term
$y_{0}$ is then mapped onto the energy $\xi_{1}$ and vice versa. We thus propose
to conserve
\begin{equation}
x_{1}\xi_{1}-y_{1}(x_{1})=x_{0}\xi_{0}-y_{0}(x_{0}). \label{eq:energycone}%
\end{equation}
The two conditions \eqref{lambdaconservation} and \eqref{eq:energycone} imply
\begin{equation}
\frac{dx_{1}}{dt_{1}}=\frac{dx_{0}}{dt_{0}}\  \text{ and }\ 2\Lambda=x_{0}%
\frac{d\varphi_{0}}{dt_{0}}=x_{1}\frac{d\varphi_{1}}{dt_{1}}
\label{eq:condchgvar}%
\end{equation}
for the radial and angular velocities of the orbits in the mapping.

The more general linear transformation of $\vec{w}=\left(  \xi x,y\right)
^{\top}$ satisfying the constraint \eqref{eq:energycone} is given by
\begin{equation}
\label{eq:Babexpr}\vec{w}_{1}=B_{\alpha,\beta}\left(  \vec{w}_{0}\right)
\text{ with\ }B_{\alpha,\beta}=\left[
\begin{array}
[c]{cc}%
\alpha & \beta \\
\alpha-1 & \beta+1
\end{array}
\right]  ,\  \left(  \alpha,\beta \right)  \in \mathbb{R}^{2}.
\end{equation}
Lynden-Bell transformation only depends on one parameter with $\beta=1-\alpha$.
From now on, we will assume $\text{det}\left(B_{\alpha,\beta}\right) =
\alpha+\beta \neq0$ because the corresponding singular transformation leads to constant
potentials or not well-defined image orbits. As a consequence, $B_{\alpha
,\beta}$ will be invertible and can be used to change the reference frame. In
this case we call $B_{\alpha,\beta}$ a \emph{bolst} in the general case or an \emph{ibolst} when it is symmetric. Reasons for these names
will become clear later.

\subsubsection{The bolst as the generalized Bohlin transformation}

A bolst $B_{\alpha,\beta}$ maps two orbits in two isochrone potentials. It induces a change of time which can be made explicit: using~\eqref{eq:condchgvar} and~\eqref{eq:Babexpr} we get
\begin{equation}
\frac{dt_{1}}{dt_{0}}=\frac{dx_{1}}{dx_{0}}=\frac{\alpha \xi_{0}}{\xi_{1}%
}+\frac{\beta}{\xi_{1}}\frac{dy_{0}}{dx_{0}} \label{causalite}.%
\end{equation}
We assume $\xi
_{1}\neq0$ since associated orbits are not well-defined in the coordinates of
$\vec{w}$. To deal with $\xi_{1}=0$ one may apply first a transvection
$J_{\epsilon,0}$ to $\vec{w}$, then study the orbit with $\xi
_{1}+\epsilon \neq0$. 

In order to ensure a bijective time transformation
$t_{0}\rightarrow t_{1}(t_{0})$, we need to impose a fixed sign on
$\frac{dt_{1}}{dt_{0}}$. For instance, we assume it to be positive. Combining
its expression~\eqref{causalite} with the second condition 
of~\eqref{eq:condchgvar}, the time evolution can be expressed in terms of 
the polar angles of the two orbits in their respective planes of motion. 
They are linked through
\begin{equation}
\left[  \frac{\alpha \xi_{0}}{\xi_{1}}+\frac{\beta}{\xi_{1}}\frac{y_{0}}{x_{0}%
}\right]  \frac{d\varphi_{1}}{d\varphi_{0}}=\frac{\alpha \xi_{0}}{\xi_{1}%
}+\frac{\beta}{\xi_{1}}\frac{dy_{0}}{dx_{0}} > 0. \label{odeboosttheta}%
\end{equation}
As we will see below, this \textsc{ode} gives $\varphi_{1}$ as a function of $\varphi_{0}$, i.e. $\varphi_{1}(\varphi_{0})$, when
$y_{0}\left(  x_{0}\right)$ is specified. When it is solved, the orbit can
be plotted in polar coordinates $\left(  x_{1},\varphi_{1}\right)$. In the
next proposition we solve this equation when a bolst is applied to a Keplerian orbit. 
In theorem~\ref{prop:Bohlin}, we call system a potential - orbit couple.

\begin{theorem}
\label{prop:Bohlin} Only the harmonic and Keplerian potentials can exchange their radial 
orbits with a linear change of polar angle. The transformation of a Kepler system into a scaled
Kepler system is given by $B_{\alpha,0}$. On the other hand, $B_{0,\beta}$ maps a Kepler system
onto a harmonic one by fully exchanging the energy and potential. This is the classical Bohlin transformation\footnote{\label{footnote:bohlin}This transformation
is also known as the transformation of Levi-Civita~\cite{LeviCivita} and was already introduced by C. MacLaurin in~\cite{ML} and then E. Goursat in~\cite{Goursat}
as excellently remarked by Alain Albouy and Niccol{\`o} Guicciardini.}.

Otherwise, when $\alpha\beta\neq0$, the image of a Keplerian $\textsc{pro}$ by
$B_{\alpha,\beta}$ is an isochrone orbit. Its azimuthal angle is given by
\begin{equation}
\varphi_{1}\left(  \varphi_{0}\right)  =\frac{\varphi_{0}}{2}+\tfrac{\chi
}{\sqrt{\left(  1+\chi \right)  ^{2}-e^{2}}}\arctan \left[  \sqrt{\tfrac
{1+\chi-e}{1+\chi+e}}\tan \left(  \frac{\varphi_{0}}{2}\right)  \right]
\  \  \  \text{with }\chi=\frac{p\alpha \left \vert \xi_{0}\right \vert }{\mu \beta},
\label{formulaphi}%
\end{equation}
where $p$ and $e$ are respectively the semilatus rectum and excentricity of
the primary Keplerian orbit. The expression holds when $\alpha\rightarrow 0$
and for the neutral bolst $B_{\alpha,0}$ when
$\beta \rightarrow0$. The precession $\Delta \varphi_1$ of the transformed polar angle during the transfer from the periastron to the apoastron and back is given by
$$
\Delta \varphi_1=\pi
\left(
1+\frac{\chi}{\sqrt{(1+\chi)^2-e^2}} 
\right).
$$
\end{theorem}

%Since $B_{0,\beta}$ is the Bohlin transformation that maps Keplerian orbits onto harmonic ones, $B_{\alpha,\beta}$ generalizes it. 
%This is the reason why we call it a generalized \textbf{\ul{bo}}h\textbf{\ul{l}}in tran\textbf{\ul{s}}forma\textbf{\ul{t}}ion.

\begin{proof}
Assume potential $\psi_{0}$ to be $\psi_{\mathrm{ke}}$. If the primary orbit is a
\textsc{pro}, then the radial distance is known by
\[
\frac{1}{r_{0}}=\frac{1+e\cos \varphi_{0}}{p},%
\]
where $p$ and $e$ are respectively the semilatus rectum and the excentricity
of the Keplerian elliptic orbit of energy $\xi_{0}<0$ that we consider.
Moreover, from equation (\ref{y_x_kepler}), we have $y_0(x_0)=-\mu\sqrt{2x_0}$. Hence,
\[
\frac{y_{0}}{x_{0}}=\psi_{0}=-\frac{\mu}{r_{0}}\text{ and }\frac{dy_{0}%
}{dx_{0}}=-\frac{\mu}{\sqrt{2x_{0}}}=-\frac{\mu}{2r_{0}}.
\]
In this case, the \textsc{ode} \eqref{odeboosttheta} becomes%
\[
\left[  \alpha \xi_{0}-\frac{\mu \beta}{p}\left(  1+e\cos \varphi_{0}\right)
\right]  \frac{d\varphi_{1}}{d\varphi_{0}}=\alpha \xi_{0}-\frac{\mu \beta}%
{2p}\left(  1+e\cos \varphi_{0}\right)
\]
for $\xi_{1}\neq0$. Two cases appear to be trivial:

\begin{enumerate}
\item When $\alpha=0$ and $\beta \neq0$, then
\[
\frac{d\varphi_{1}}{d\varphi_{0}}=\frac{1}{2}.%
\]
The system \eqref{eq:Babexpr} can be directly inverted and gives%
\[
\left \{
\begin{array}
[c]{l}%
r_{1}^{2}=-\frac{\beta \mu}{\xi_{1}}r_{0},\\
\psi_{1}\left(  r_{1}\right)  =\beta_{1}+\frac{1}{2}\omega_{1}r_{1}%
^{2}, \quad \text{\ where\ }\omega_{1}^{2}=\frac{2\left \vert \xi_{0}\right \vert
\xi_{1}^{2}}{\mu^{2}\beta^{2}}\text{ and }\beta_{1}=\frac{
\beta+1 }{\beta}\xi_{1}.
\end{array}
\right.
\]
This duality between the harmonic and the Keplerian potentials is the same as
that described by a Bohlin transformation~\cite{Grandati:2010}. In order to
get a real $r_{1}$, the quantity $\frac{\beta}{\xi_{1}}$ must be negative. The angle  $\varphi_{0}$ of the Keplerian orbit is twice that
of the corresponding  $\varphi_{1}$ of the harmonic one, as
represented in figure~\ref{twospecialcases}. The focus $F$ of the
Keplerian ellipse is the center of the harmonic one. \newline

\item When $\alpha \neq0$ and $\beta=0$, then%
\[
\frac{d\varphi_{1}}{d\varphi_{0}}=1.
\]
The system can still be inverted as
\[
\left \{
\begin{array}
[c]{l}%
r_{1}^{2}=\frac{\alpha \xi_{0}}{\xi_{1}}r_{0}^{2},\\
\psi_{1}\left(  r_{1}\right)  =\alpha_{1}-\frac{\mu_{1}}{r_{1}}, \quad \text{ where
}\alpha_{1}=\frac{\left(  \alpha-1\right)  }{\alpha}\xi_{1}\text{ and }\mu
_{1}=\mu \sqrt{\frac{\xi_{1}}{\alpha \xi_{0}}}.
\end{array}
\right.
\]
The quantity $\frac{\xi_{1}}{\alpha \xi_{0}}$ must be positive when
$x_{0,1}=2r_{0,1}^{2}>0$. This transformation maps the primary Keplerian
ellipse onto a scaled confocal one. The two moving points are always aligned
with the common focus of the two ellipses. As $\xi_{1}$ needs to be negative
to ensure bounded bolsted orbit, this imposes $\alpha>0$.
\end{enumerate}

These two special cases are represented in figure \ref{twospecialcases}.
\begin{figure}[h]
\centering
\resizebox{0.95\textwidth}{!}{
\includegraphics{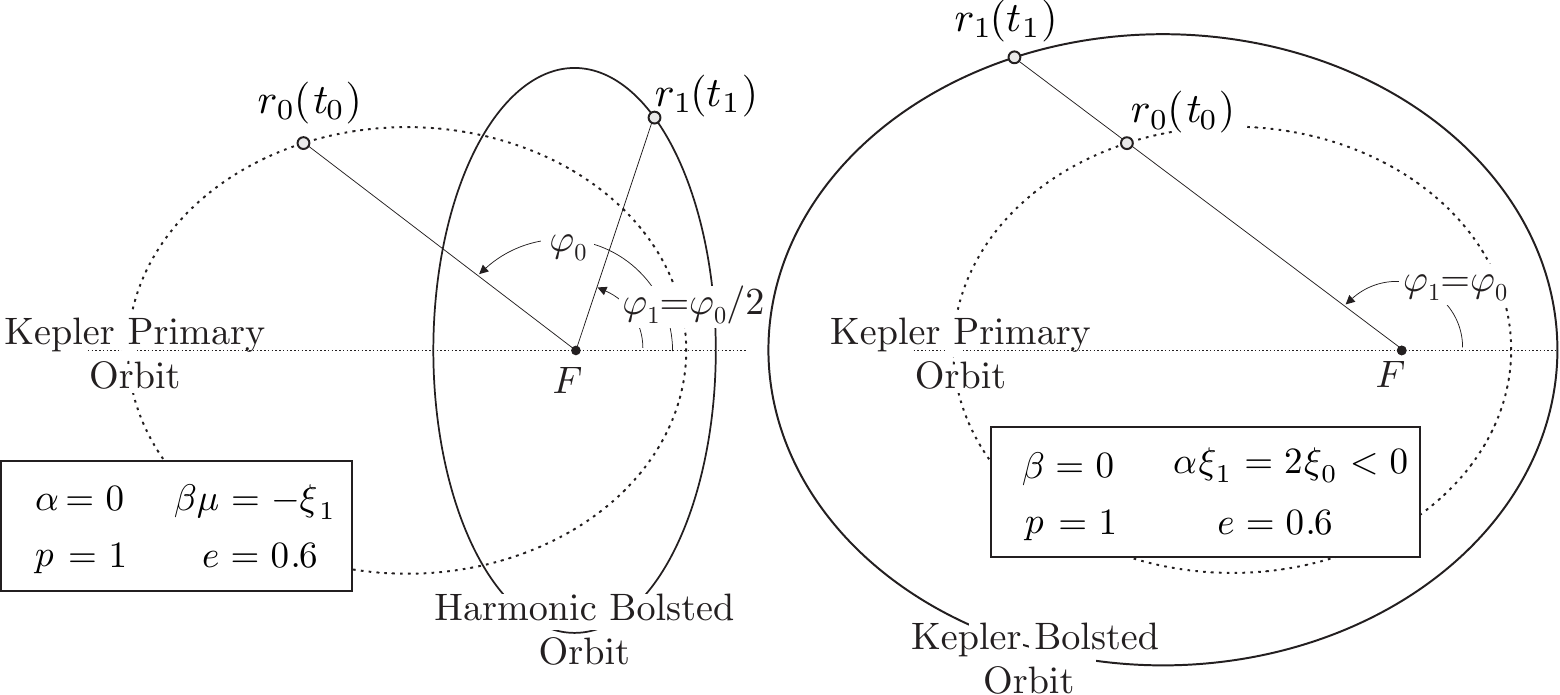} } \centering
\caption{The transformation $B_{0,\beta}$ of a Keplerian \textsc{pro} gives a
harmonic \textsc{pro} when $\beta \xi_{1}<0$, as represented on the left
panel. The transformation $B_{\alpha,0}$ of a Keplerian \textsc{pro} gives a
Keplerian \textsc{pro} when $\alpha>0$ and $\xi_{1}<0$, as represented on the
right panel. }%
\label{twospecialcases}%
\end{figure}

To show that only the harmonic and Keplerian potentials can
exchange their radial orbit with a linear change of polar angles, we assume that
\begin{equation}
\label{eq:linearchgpolarangle}\varphi_{1}(t_{1}) = m \varphi_{0}(t_{0})
\  \text{ with~} \ m = \text{cst}.
\end{equation}
Combining \eqref{eq:condchgvar} with the derivative
of~\eqref{eq:linearchgpolarangle} one can verify that $y_{0}$ satisfies the
\textsc{ode}
\begin{equation}
\frac{dy}{dx_{0}}-\frac{m}{x_{0}}y=\xi_{0}\frac{\alpha}{\beta}\left(
m-1\right),  \label{eq:y0linearchgpolarangle}%
\end{equation}
which holds for $\beta \neq0$. The solution of~\eqref{eq:y0linearchgpolarangle}
is given by
\[
y_{0}(x_{0})=kx_{0}^{m}-\xi_{0}\frac{\alpha}{\beta}x_{0}.
\]
But $y_{0}$ must describe a parabola, so either

\begin{itemize}
\item $m=1$ or $m=0$. Then the potential is constant or constant with a gauge,
and no \textsc{pro} exists.

\item $m=\frac{1}{2}$. Then $y_{0}$ represents a Keplerian potential up to a
constant. Inserting the solution $y_{0}$ in $(\xi_{1} x_{1}, y_{1})^{\top}$,
$y_{1}$ is a harmonic potential with a constant.

\item $m=2$. Then $y_{0}$ represents a harmonic potential up to a constant and
$y_{1}$ a transvected Keplerian potential.
\end{itemize}

%preuve cas le plus general
Let us examine now the more general case when $\alpha \beta \neq0$. The \textsc{ode} for
phases is written as%
\begin{equation}
\frac{d\varphi_{1}}{d\varphi_{0}}=\frac{N\left(  \varphi_{0}\right)
}{D\left(  \varphi_{0}\right)  }\  \text{\ where }\left \{
\begin{array}
[c]{lcccl}%
N\left(  \varphi_0 \right)  &=&\dfrac{dx_{1}}{dx_{0}}&=&\dfrac{1}{\xi_{1}}\left[
\alpha \xi_{0}-\frac{\mu \beta}{2p}\left(  1+e\cos \varphi_{0}\right)  \right], \\
D\left(  \varphi_0 \right) & =& \dfrac{x_{1}}{x_{0}}&=&\dfrac{1}{\xi_{1}}\left[
\alpha \xi_{0}-\frac{\mu \beta}{p}\left(  1+e\cos \varphi_{0}\right)  \right].
\end{array}
\right.  \label{Nphi/dphi}%
\end{equation}
We first remark that the denominator function $\varphi \rightarrow D\left(
\varphi \right)  $ is strictly positive as both $x_{0}=2r_{0}^{2}$ and
$x_{1}=2r_{1}^{2}$ are positive functions. In \eqref{causalite} we have seen
that the sign of $N\left(  \varphi \right)$ cannot change; as a consequence
the function $\varphi_{0}\rightarrow \varphi_{1}\left(  \varphi_{0}\right)  $
is monotone. In our hypothesis where $N\left(  \varphi \right)  \geq0$,
$\varphi_{1}$ is an increasing function of $\varphi_{0}$. After a little
rearrangement, from \eqref{Nphi/dphi} we obtain
\[
\varphi_{1}={\displaystyle \int_{0}^{\varphi_{0}}}\frac{\eta+\cos \varphi
}{\delta+2\cos \varphi}d\varphi \quad  \text{where }\eta=\frac{\mu \beta
-2p\alpha \xi_{0}}{\mu \beta e}\geq1\text{ and }\delta=\frac{2\mu \beta
-2p\alpha \xi_{0}}{\mu \beta e}>2.
\]
We notice that the particular case when the primary Keplerian orbit is
circular, i.e. $e=0$, linearly links $\varphi_{0}$ and $\varphi_{1}$. The
integral for $\varphi_{1}$ can be made explicit: introducing $u=\tan \left(
\varphi/2\right)  $ we get$\  \cos \varphi=\frac{1-u^{2}}{1+u^{2}}$,
$d\varphi=\frac{2du}{1+u^{2}}$ and thus\label{Bioche}
\[
\  \text{ }\varphi_{1}=2{\displaystyle \int_{0}^{u_{0}}}\frac{\ell+2+\ell
u^{2}}{\left(  m+4+mu^{2}\right)  \left(  1+u^{2}\right)  }du\  \text{\ where}%
\left \{
\begin{array}
[c]{c}%
\ell=\eta-1\geq0,\\
m=\delta-2>0.
\end{array}
\right.
\]
A partial fraction decomposition gives
\[
\varphi_{1}={\displaystyle \int_{0}^{u_{0}}}\frac{1}{1+u^{2}}du+\frac{2\ell
-m}{\sqrt{m\left(  m+4\right)  }}{\displaystyle \int_{0}^{v_{0}}}\frac
{dv}{1+v^{2}},%
\]
where $v=\sqrt{\frac{m}{m+4}}u$. The integration leads to
\[
\varphi_{1}=\frac{\varphi_{0}}{2}+\frac{2\ell-m}{\sqrt{m\left(  m+4\right)  }%
}\arctan \left[  \sqrt{\frac{m}{m+4}}\tan \left(  \frac{\varphi_{0}}{2}\right)
\right],
\]
and so%
\[
\varphi_{1}=\frac{\varphi_{0}}{2}+\tfrac{\chi}{\sqrt{\left(  1+\chi \right)
^{2}-e^{2}}}\arctan \left[  \sqrt{\tfrac{1+\chi-e}{1+\chi+e}}\tan \left(
\frac{\varphi_{0}}{2}\right)  \right]  \  \  \  \text{with }\chi=\frac
{p\alpha \left \vert \xi_{0}\right \vert }{\mu \beta}.
\]
If $\alpha=0$ we would recover the relation $\varphi_{1}=\frac{\varphi_{0}}%
{2}$ previously mentioned. In the same way, when $\beta \rightarrow0$, then
$\varphi_{1}\rightarrow \varphi_{0}$. 
When the bolsted orbit is a \textsc{pro}, we can easily compute the increment of the azimuthal angle $\Delta \varphi$
during the transfer from $r_{a}$ to $r_{p}$ and back. In the Keplerian case, from figure \ref{fig:generalboostane0betane0}, we see that the transfer 
for $r_0:r_{0,p}\to r_{0,a}$ corresponds to $\varphi_0:0\to \pi$. Hence, using (\ref{formulaphi}) one gets
\[
\begin{array}[l]{ll}
\varphi_1:0 \to \dfrac{1}{2}\Delta \varphi_1 &= \dfrac{\pi}{2}+\dfrac{\chi}{\sqrt{(1+\chi)^2-e^2}}\mathrm{arctan}(\infty)\\
  & =\dfrac{\pi}{2}\left( 1+\frac{\chi}{\sqrt{(1+\chi)^2-e^2}} \right).
\end{array}
\]
Since
\[
p=\frac{\Lambda^{2}}{\mu}\text{ and }e=\sqrt{1+\frac{2\Lambda^{2}\xi_{0}}%
{\mu^{2}}},%
\]
we see that $\Delta \varphi_1$ depends on $\Lambda^{2}$ but not on $\xi_{1}$.
This is a characterization of isochrone orbits (see theorem~\ref{thm:characisopot} 
in appendix~\ref{appendix:isocharac}). Given a point $\left(
\varphi_{0},r_{0}\right)  $ on the primary Keplerian ellipse, its image on the
bolsted orbit has a polar angle $\varphi_{1}$ given by the formula
\eqref{formulaphi} and a distance $r_{1}$ given by the relation
\eqref{eq:Babexpr}, i.e.
\[
x_{1}=2r_{1}^{2}=2\alpha \frac{\xi_{0}}{\xi_{1}}r_{0}^{2}-2\mu \frac{\beta}%
{\xi_{1}}r_{0}\implies r_{1}^{2}=\frac{\alpha \xi_{0}r_{0}^{2}-\mu \beta r_{0}%
}{\xi_{1}}.
\]
When $\alpha$, $\beta$ and $\xi_{1}$ are such that $r_{1}^{2}>0$ for all
$r_{0}$ on the Keplerian orbit, this corresponds to an isochrone \textsc{pro}.
\qed

\end{proof}

Theorem \ref{prop:Bohlin} shows that any Keplerian \textsc{pro} can be
transformed into a particular isochrone one by a suitable bolst $B_{\alpha
,\beta}$. When $\alpha=0$, the bolst coincides with a Bohlin transformation.
In the other cases, it generalizes it; we have plotted an example of such a bolst in figure~\ref{fig:generalboostane0betane0}.

\begin{figure}[h]
\centering
\resizebox{0.9\textwidth}{!}{
\includegraphics{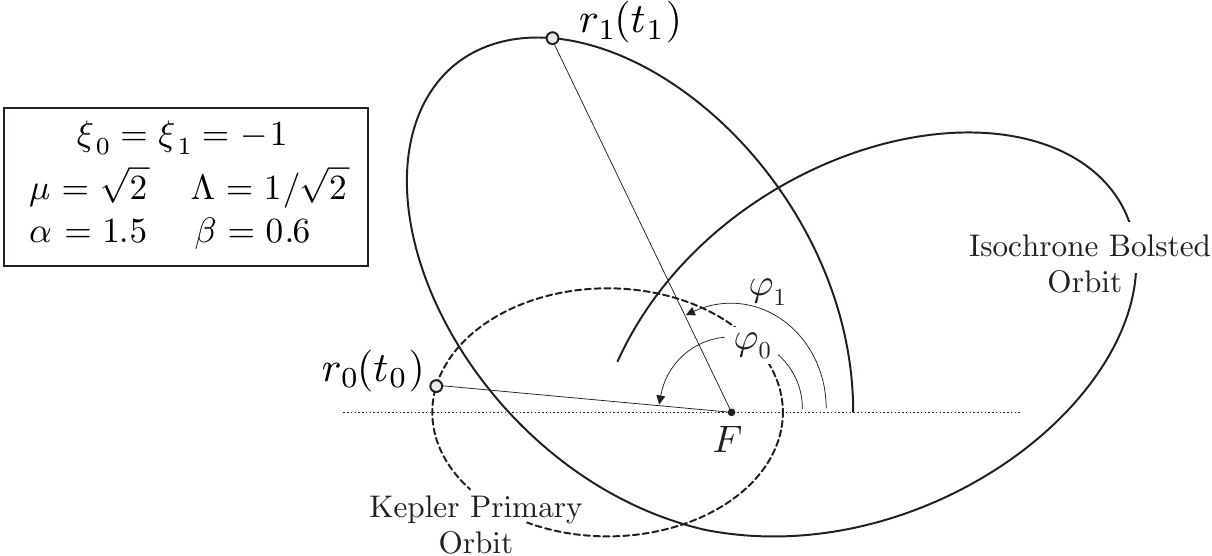} } \centering
\caption{The $\alpha=1.5$ and $\beta=0.6$ bolst of a Keplerian ellipse
($e\simeq0,7$; $p\simeq0.35$) which gives an isochrone orbit with the same
energy $\left(  \xi_{0}=\xi_{1}=-1\right)  $. }%
\label{fig:generalboostane0betane0}
\end{figure}

Reciprocally, we will see in sec. \ref{isoconstruct} that any
isochrone \textsc{pro} could be connected to a Keplerian ellipse.

\subsubsection{The bolst, a key to isochrony}

%The boosts $B_{\alpha,\beta}$ relates isochrone orbits together. 

Geometrically, a Keplerian parabola in a frame $\mathcal{R}_{O}=\left(
O,\vec{\mathit{i}},\vec{\mathit{j}}\right)  $ is laid (see sec.
\ref{subsec:GalPropIsoParabolas}, p. \pageref{laidpara}), i.e. its tangent at the origin is
$\mathbb{R}\vec{\mathit{j}}$ and its axis of symmetry is $\mathbb{R}%
\vec{\mathit{i}}$. According to lemma  \ref{lemme_des_droites} in
appendix~\ref{appendixlemma}, its image by a bolst $B_{\alpha,\beta}$ remains laid, and
so Keplerian, in the image frame $B_{\alpha,\beta}\left(  \mathcal{R}%
_{O}\right)  $. But, in $\mathcal{R}_{O}$, the image parabola has two distinct
intersections with $\mathbb{R}\vec{\mathit{j}}$ and thus appears to be a non-Keplerian isochrone. Therefore, it appears that to be or not to be Keplerian
depends
on the choice of the reference frame. This is an aspect of the isochrone
relativity that we will discuss in what follows.

For this discussion, we will not consider the general case of any bolst
$B_{\alpha,\beta}$. With only technical restrictions, we will consider the case where the bolst is 
symmetric: according to (\ref{eq:Babexpr}), $B_{\alpha,\beta}$ is a symmetric matrix if and only if $\alpha-1=\beta$. Introducing the parameter 
$\gamma=\alpha+\beta$ which is the variable eigenvalue of $B_{\alpha,\beta}$, with the other $1$, the
general bolst $B_{\alpha,\beta}$ then becomes the symmetric $B_{\gamma}$ that
we call an \emph{ibolst} for which the isochrone relativity appears to be clear.

We have seen that bolsts generalize the Bohlin transformation, and we will see now that ibolsts are the boosts of the isochrone relativity. 
Names appear to be clarified: \emph{bolst} stands for \textbf{\ul{bo}}h\textbf{\ul{l}}in boo\textbf{\ul{st}} and \emph{ibolst} for symmetr\textbf{\ul{i}}c \textbf{\ul{bo}}h\textbf{\ul{l}}in boo\textbf{\ul{st}}.

\subsection{Isochrone relativity}

\label{subsec:isochronerelativity}

The special theory of relativity has two pillars:

\begin{enumerate}
\item The Einstein principle of special relativity imposes that the laws
of physics can be written in the same way in all Galilean frames;

\item The length of any space-time interval is conserved through
changes of Galilean frames, aka Lorentz frames.
\end{enumerate}

These principles make time and length relative to a given Galilean
frame. These two physical quantities are linearly exchanged during changes of
Galilean frames.

In the same way, the linear exchange between $\xi x$ and $y$ proposed in the
previous section conserves the ``\emph{isochrone interval}" $\xi x-y$ in equation
\eqref{eq:energycone}. This conservation is imposed by that of the
fundamental orbital law \eqref{fundamentalode} which renders the conservation
of the energy along the orbit. The linearity of the
transformation is associated with the isochrony preservation. The conservations
of the ``isochrone interval" and isochrone law are the two pillars of the
isochrone relativity.

For the sake of simplicity, we restrict our attention to symmetric exchanges
between $\xi x$ and $y$: the bolst $B_{\alpha,\beta}$ is then reduced to the
ibolst $B_{\gamma=\alpha+\beta}$, choosing $\alpha-1=\beta$,
\[
B_{\gamma}=\frac{1}{2}\left[
\begin{array}
[c]{cc}%
\gamma+1 & \gamma-1\\
\gamma-1 & \gamma+1
\end{array}
\right]  .
\]

\subsubsection{The ibolst Algebra}

Let $\mathcal{R}=\left(  \vec{\mathit{i}},\vec{\mathit{j}}\right)$ be the
canonical basis of $\mathbb{R}^{2}$. Any vector $\vec{z}\in \mathbb{R}^{2}$ has
affine coordinates $\left(  z_{1},z_{2}\right)  $ in the frame $\mathcal{R}%
_{O}=\left(  O,\vec{\mathit{i}},\vec{\mathit{j}}\right)$, i.e. there exists
a unique point $Z$ in the $Oz_{1}z_{2}$ plane such that $\vec{z}=\overrightarrow{OZ}=z_{1}%
\vec{\mathit{i}}+z_{2}\vec{\mathit{j}}$. We do not use the usual upper index
for contravariant components because, as we are in $\mathbb{R}^{2}$, we do not
use Einstein notation for sums and we prefer to conserve the upper index for
powers. The orthonormality is defined in the Euclidian sense, i.e. with
natural notations
\[
\left \Vert \vec{z}\right \Vert ^{2}=\left(  \vec{z}|\vec{z}\right)  =z_{1}%
^{2}+z_{2}^{2}\  \  \text{then\ }\left(  \vec{\mathit{i}}|\vec{\mathit{i}%
}\right)  :=\left \Vert \vec{\mathit{i}}\right \Vert ^{2}=\left \Vert
\vec{\mathit{j}}\right \Vert ^{2}=:\left(  \vec{\mathit{j}}|\vec{\mathit{j}%
}\right)  =1\text{ and }\left(  \vec{\mathit{i}}|\vec{\mathit{j}}\right)
=0\text{.}%
\]
The $\mathcal{R}$ basis is then orthonormal for the Euclidian scalar product. 
We will also use the Minkowski scalar product for which%
\[
\left \Vert \vec{z}\right \Vert _{m}^{2}=\left \langle \vec{z}|\vec
{z}\right \rangle =z_{1}^{2}-z_{2}^{2}\  \  \text{then\ }\left \langle
\vec{\mathit{i}}|\vec{\mathit{i}}\right \rangle :=\left \Vert \vec{\mathit{i}%
}\right \Vert _{m}^{2}=1\text{,\  \ }\left \langle \vec{\mathit{j}}%
|\vec{\mathit{j}}\right \rangle :=\left \Vert \vec{\mathit{j}}\right \Vert
_{m}^{2}=-1\text{ and }\left \langle \vec{\mathit{i}}|\vec{\mathit{j}%
}\right \rangle =0\text{.}%
\]

Consider the two eigenvectors $\vec{\mathit{k}}=\frac{1}{\sqrt{2}}\left(
\vec{\mathit{i}}-\vec{\mathit{j}}\right)  $ and $\vec{\mathit{l}}=\frac
{1}{\sqrt{2}}\left(  \vec{\mathit{i}}+\vec{\mathit{j}}\right)  $ of the ibolst
$B_{\gamma}$ such that 
\begin{equation}
B_{\gamma}\left(  \vec{\mathit{k}}\right)  =\vec{\mathit{k}}\text{ and
}B_{\gamma}\left(  \vec{\mathit{l}}\right)  =\gamma \vec{\mathit{l}}.
\label{B_gamma_eigen}
\end{equation}
The basis $\mathcal{\tilde{R}}=\left(  \vec{\mathit{k}},\vec{\mathit{l}%
}\right)  $ is just $\mathcal{R}$ rotated by an angle of $-\frac{\pi}{4}$.
It is thus orthonormal for the Euclidian scalar product.
Moreover, we see that for the Minkowski scalar product, we have
\begin{equation}
\label{eq:klMink}
\left \langle \vec{\mathit{k}}|\vec{\mathit{k}}\right \rangle =\left \langle
\vec{\mathit{l}}|\vec{\mathit{l}}\right \rangle =0\text{ and }\left \langle
\vec{\mathit{k}}|\vec{\mathit{l}}\right \rangle =\left \langle \vec{\mathit{l}%
}|\vec{\mathit{k}}\right \rangle =1.
\end{equation}
From (\ref{B_gamma_eigen}), let us remark that the set $\mathbb{B}=\left \{  B_{\gamma},\gamma \in
\mathbb{R}^{\ast}\right \}  $ forms a commutative linear group since 
\begin{equation*}
\forall \left(  \gamma,\gamma^{\prime}\right)  \in \mathbb{R}^{\ast}%
\times \mathbb{R}^{\ast},\  \ B_{\gamma}\circ B_{\gamma^{\prime}}=B_{\gamma
^{\prime}}\circ B_{\gamma}=B_{\gamma \gamma^{\prime}}\in \mathbb{B}.
\end{equation*}
For this law, $B_{1}$ is an identity element. The inverse of a
transformation $B_{\gamma}$ for $\gamma \in \mathbb{R}^{\ast}$ is $B_{\frac
{1}{\gamma}}$.

As expected, any ibolst is symmetric, i.e. for the Euclidian scalar product
and for any vectors $\vec{w}$ and $\vec{z}$, we have
\begin{equation}
\left(  B_{\gamma}\left(  \vec{w}\right)  |\vec{z}\right)  =\left(  \vec
{w}|B_{\gamma}\left(  \vec{z}\right)  \right)  . \label{iboost-symmetry}%
\end{equation}
As a matter of fact, since the matrix $B_{\gamma}$ is symmetric, considering the expansion of
these vectors in the basis $\mathcal{\tilde{R}}$ noted with a tilde, we get directly from (\ref{B_gamma_eigen}) that
\begin{align*}
\left(  B_{\gamma}\left(  \vec{w}\right)  |\vec{z}\right)   &  =\left(
\tilde{w}_{1}B_{\gamma}\left(  \vec{\mathit{k}}\right)  +\tilde{w}%
_{2}B_{\gamma}\left(  \vec{\mathit{l}}\right)  |\tilde{z}_{1}\vec{\mathit{k}%
}+\tilde{z}_{2}\vec{\mathit{l}}\right)  =\tilde{w}_{1}\tilde{z}_{1}%
+\gamma \tilde{w}_{2}\tilde{z}_{2}\\
&  =\tilde{z}_{1}\tilde{w}_{1}+\gamma \tilde{z}_{2}\tilde{w}_{2}=\left(
B_{\gamma}\left(  \vec{z}\right)  |\vec{w}\right)  =\left(  \vec{w}|B_{\gamma
}\left(  \vec{z}\right)  \right)  .
\end{align*}
%This property ensures that $B_{\gamma}$ can be represented in the basis
%$\left(  \vec{\mathit{k}},\vec{\mathit{l}}\right)  $, and then in the basis
%$\left(  \vec{\mathit{i}},\vec{\mathit{j}}\right)  $, by a symmetric matrix.
%The same symmetry property also ensures that the mapping between $\gamma
%\in \mathbb{R}^{\ast}$ and any symmetric $2\times2\ $matrix represented by
%$B_{\gamma}$ is a group homomorphism. 
However this symmetry property does not generally hold for the Minkowski scalar product.

\subsubsection{Lengths and spaces}

Let us consider $\xi$ and $\Lambda$ as two fixed parameters. We can define in
$\mathcal{R}_{O}\ $the affine coordinates system $\left(  w_{1}=\xi
x,w_{2}=y\right)  $. Using these coordinates we set
\[
\vec{w}^{\prime}=B_{\gamma}\left(  \vec{w}\right)  \text{.}%
\]
The symmetry \eqref{iboost-symmetry} of the ibolst for the Euclidian scalar
product gives
\[
\forall \alpha \in \mathbb{R}\text{,\ }\left(  \vec{w}^{\prime}|\alpha
\vec{\mathit{l}}\right)  =\left(  B_{\gamma}\left(  \vec{w}\right)
|\alpha \vec{\mathit{l}}\right)  =\left(  \vec{w}|B_{\gamma}\left(  \alpha
\vec{\mathit{l}}\right)  \right)  =\gamma \left(  \vec{w}|\alpha \vec
{\mathit{l}}\right)  .
\]
With $\alpha=\sqrt{2}$, this relation corresponds to the equality
\begin{equation}
\xi^{\prime}x^{\prime}+y^{\prime}=\gamma \left(  \xi x+y\right)  \text{.}
\label{proopoone}%
\end{equation}
This same symmetry, but in the direction given by $\vec{\mathit{k}}$, gives the conservation of the isochrone interval
\begin{equation}
\forall \alpha \in \mathbb{R}\text{,\ }\left(  \vec{w}^{\prime}|\alpha
\vec{\mathit{k}}\right)  =\left(  \vec{w}|\alpha \vec{\mathit{k}}\right)
\  \Rightarrow \  \xi^{\prime}x^{\prime}-y^{\prime}=\xi x-y. \label{proptwo}%
\end{equation}
By multiplication of these two relations we get directly\footnote{The relation
\eqref{proptwo} holds for any bolst $B_{\alpha,\beta}$. This is not the case
for \eqref{proopoone} which requires the $B_{\gamma}$--symmetry. As a
consequence, the relation \eqref{propthree} is simple only in the symmetric
case.}%
\begin{equation}
\left(  \xi^{\prime}x^{\prime}\right)  ^{2}-y^{\prime2}=\gamma \left[  \left(
\xi x\right)  ^{2}-y^{2}\right]  . \label{propthree}%
\end{equation}

This relation corresponds to the fact that an ibolst is not an isometry using
the Minkowskian norm%
\begin{equation}
\left \langle \vec{w}^{\prime}|\vec{w}^{\prime}\right \rangle =\gamma
\left \langle \vec{w}|\vec{w}\right \rangle .
\end{equation}

As a consequence, the radial cone
\[
\mathscr{C}=\left \{  \vec{z}\in \mathbb{R}^{2},\  \left \langle \vec{z}|\vec
{z}\right \rangle =0\right \}
\]
is preserved by the ibolst as $\mathscr{C}=\mathbb{R}\vec{\mathit{k}}%
\cup \mathbb{R}\vec{\mathit{l}}$. Its name comes from the fact that the line
$y=\xi x$ defines a radial orbit $\left(  \Lambda=0\right)  $ of energy $\xi$
in the potential $\psi \left(  r\right)  $. In a Kepler potential
$\psi_{\mathrm{ke}}\left(  r\right)  =-\frac{\mu}{r}$ a test particle of
energy $\xi<0$ with a radial orbit moves on a segment from $r_{a}=\frac{\mu
}{\left \vert \xi \right \vert }$ at $t=0$ to $r\rightarrow0$ when $t\rightarrow
+\infty$. As its period should be infinite, a radial orbit is not a
\textsc{pro} but we can say that it is a maximal time-bounded orbit.

In this relativistic formulation of the problem we can then define
periodic-like vectors lying in the periodic space%
\[
\mathscr{P}=\left \{  \vec{z}\in \mathbb{R}^{2},\  \left \langle \vec{z}|\vec
{z}\right \rangle <0\right \}
\]
and aperiodic-like vectors lying in the aperiodic space%
\[
\mathscr{A}=\left \{  \vec{z}\in \mathbb{R}^{2},\  \left \langle \vec{z}|\vec
{z}\right \rangle >0\right \}  .
\]
As the convex $x-$positive part of parabolas containing $\textsc{pro}$ in the
coordinates system $(\xi x,y)$ is delimited by the radial cone and exactly
contained in $\mathscr P$, the names $\mathscr P$ and $\mathscr{A}$ are natural.

\subsubsection{Orbits relativity}

Let us define the ibolsted frame $\mathcal{R}_{O}^{\prime}=\left(  O,\vec
{u},\vec{v}\right)  $ such that%
\begin{equation}
\label{eq:uvbasis}
\left \{
\begin{array}
[c]{l}%
\vec{u}=B_{\gamma}\left(  \vec{\mathit{i}}\right)  
       =B_{\gamma}\left(  \dfrac{\vec{\mathit{l}}+\vec{\mathit{k}}}{\sqrt{2}}  \right)
       =\dfrac{\gamma \vec{\mathit{l}}+\vec{\mathit{k}}}{\sqrt{2}}\\
\text{and}\\
\vec{v}=B_{\gamma}\left(  \vec{\mathit{j}}\right)  
       =B_{\gamma}\left(  \dfrac{\vec{\mathit{l}}-\vec{\mathit{k}}}{\sqrt{2}}  \right)
       =\dfrac{\gamma \vec{\mathit{l}}-\vec{\mathit{k}}}{\sqrt{2}}
\end{array}
\right.  \implies \vec{\mathit{k}}=\frac{\vec{u}-\vec{v}}{\sqrt{2}}.
\end{equation}

\begin{definition}
\label{def:refframe}
The reference frame of a given parabola $\mathcal{P}$ is the frame $\left(
O,\vec{t},\vec{n}\right)$ where the tangent to the parabola at the origin is
$\mathcal{T}_{O}\left(  \mathcal{P}\right)  =\mathbb{R}\vec{t}$ and the
symmetry axis is $\mathcal{S}\left(  \mathcal{P}\right)  =$ $\mathbb{R}\vec
{n}$.
\end{definition}

A reference frame geometrically defines a parabola up to a scale factor. For
instance,< $\mathcal{R}_{O}$ is the reference frame of the Keplerian parabola
containing $\psi_{\mathrm{ke}}$ up to the scale factor $\mu$. According to
lemma~\ref{lemme_des_droites} in appendix~\ref{appendixlemma}, the line $\mathbb{R}%
\vec{v}$ is tangent to the bolsted parabola and $\mathbb{R}\vec{u}$ is its
symmetry axis. Thus, $\mathcal{R}_{O}^{\prime}$ is the reference frame of the
bolsted parabola and characterizes it up to a scale factor.
\\

\begin{figure}[ptb]
\centering \resizebox{0.95\textwidth}{!}{\includegraphics{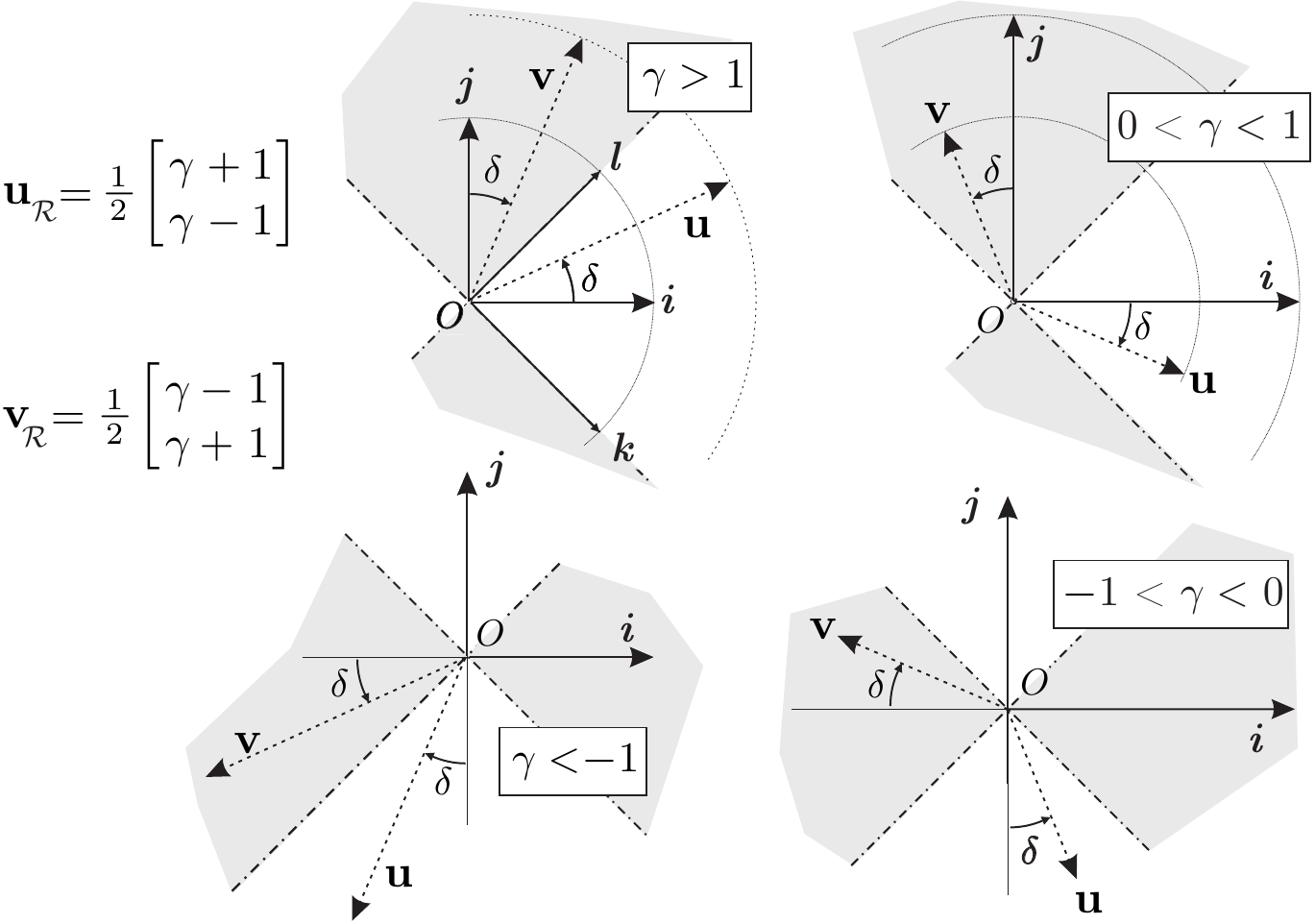}}
\centering
\caption{The bolsted frame $\mathcal{R}_{O}^{\prime}=\left(  O,\vec{u},\vec
{v}\right)  $, when $\xi \xi^{\prime}>0$. Its periodic space $\mathscr{P}^{\prime}$ is represented in grey, while its aperiodic one $\mathscr{A}^{\prime}$ is in white.}%
\label{vecteursuv}%
\end{figure}

All possibilities are represented in figure~\ref{vecteursuv}, when the primary
energy $\xi$ and the image energy $\xi^{\prime}$ share the same sign. When
$\xi \xi^{\prime}<0$, the direction of $\vec{u}$ has to be inverted.

Depending on the value of $\gamma \neq1$, we can define the angle $\delta$ given by
\[
\tan \delta=\left \vert \frac{\gamma+1}{\gamma-1}\right \vert
\]
which is useful to construct $\mathcal{R}_{O}^{\prime}$ from $\mathcal{R}_{O}$
by simple composition of a homothety and a hyperbolic rotation (see~\cite{Borel} p.28 for a nice description in French or \cite{Gourg} for general properties of rotation in special relativity).

As with $\vec{\mathit{i}}$ and $\vec{\mathit{j}}$, the two ibolsted basis vectors
$\vec{u}$ and $\vec{v}$ have the same Euclidian norm
\[
\left \Vert \vec{u}\right \Vert ^{2}=\left \Vert \vec{v}\right \Vert ^{2}%
=\frac{\gamma^{2}+1}{2}%
\]
and opposed Minkowskian lengths%
\[
\left \Vert \vec{u}\right \Vert_m =\gamma=-\left \Vert\vec{v}\right \Vert_m.
\]
Moreover, from~\eqref{eq:klMink} and~\eqref{eq:uvbasis}, the primary and the ibolsted basis are orthogonal in the
Minkowskian scalar product: $\left \langle \vec{\mathit{i}}|\vec{\mathit{j}%
}\right \rangle =\left \langle \vec{u}|\vec{v}\right \rangle =0$. Depending on $\gamma$ 
and on the frame $\mathcal{R}_{O}$ or $\mathcal{R}_{O}^{\prime}$ used to define the scalar product, one vector
is aperiodic-like and the other is periodic-like, see figure~\ref{vecteursuv}. \newline

%Let us consider an orbit in this Keplerian potential. This orbit is
%completely characterized when an energy $\xi $ and an angular momentum $%
%\Lambda $ are given. When $\Lambda \neq 0$, \textsc{pro} corresponds to $\xi
%\in \left[ \xi _{c},0\right[ $ where $\xi _{c}=-\mu ^{2}/\Lambda ^{2}$ is
%the circular orbit energy for the Keplerian potential. The case $\Lambda =0$
%corresponds to radial orbits.
%The orbits are then completely characterized by the energy and angular momentum.

In the canonical frame $\mathcal{R}_{O}$, using the isochrone relativity
formalism and introducing the proper time $d\tau=\xi dt$ of an orbit of energy
$\xi$ and angular momentum $\Lambda$, the orbital differential
equation~\eqref{fundamentalode} in the affine coordinates $(\xi x,y)$ can be written as
\begin{equation}
\frac{1}{16}\left[  \frac{d}{d\tau}\left(  \vec{w}|\vec{\mathit{i}}\right)
\right]  ^{2}=\left(  \vec{w}|\vec{\mathit{i}}-\vec{\mathit{j}}\right)
+\left(  \vec{w}_{\Lambda}|\vec{\mathit{j}}\right),
\label{orbital-differential-equation}
\end{equation}
where $\vec{w}_{\Lambda}
=-\Lambda^{2}\vec{\mathit{j}}$. The vector $\vec{w}$ describes the potential
of parabola $\mathcal{P}$ and the orbit which corresponds to an arc of
$\mathcal{P}$. When this orbit is a \textsc{pro}, this arc is finite. When
$\vec{w}$ describes a Keplerian orbit, its ibolsted image $\vec{w}^{\prime}%
$\ is characterized by theorem~\ref{thm:orbitboost}.
%, the two extremities $A$ and $P$ of this
%parabolic arc correspond to the apoastre and the periastre of the orbit.
%These two points are associated to the two solutions $\vec{w}_{a}$ and $\vec{%
%w}_{p}$ of the equation $\frac{d}{d\tau }\left( \vec{w},\vec{\mathit{i}}%
%\right) =0$. Considering the orbital differential equation %
%\eqref{orbital-differential-equation} in the coordinates $(\xi x,y)$ these
%two extremal points of the orbit are on the extremal line
%\begin{equation*}
%\Delta=\left \{ \vec{w}\in \mathbb{R}^{2},\sqrt{2}\left( \vec{w}|\vec{%
%\mathit{k}}\right) =\Lambda ^{2}\right \}.
%\end{equation*}
%We note $\vec{w_\Lambda}=\overrightarrow{OK}$ such that $K$ %=\Delta \cap(O,\vec j)$
%is the $y-$intercept of $\Delta$. Similarly, we set
%$\vec {w^\prime_\Lambda}=%
%\overrightarrow{OK^\prime}$ with $\overrightarrow{OK^\prime}=B_\gamma(\overrightarrow{OK})$.
%$K^\prime$ corresponds to the $y^\prime-$intersect of $\Delta^\prime=B_\gamma(\Delta)$ with
%$(O,\vec v)$. %, $K^\prime=\Delta^\prime \cap (O,\vec v)$.
%%we set $\vec {%
%%w^\prime_{\Lambda^\prime}}=\overrightarrow{OH}$ the $y-$intercept of $%
%%\Delta^\prime=B_\gamma \left(\Delta \right)$ and

\begin{theorem}
\label{thm:orbitboost} A vector $\vec w^{\prime}$ describes an isochrone orbit
$(\xi^{\prime},\Lambda^{\prime})$ on its arc of parabola if and only if it
is the image, by an ibolst $B_{\gamma}\in \mathbb{B}$, of a vector $\vec w$
which describes a Keplerian orbit $(\xi,\Lambda)$ on a Keplerian parabola. In
the Keplerian frame $\mathcal{R}_{O}=(O,\vec{\mathit{ i}},\vec{\mathit{j}})$,
the orbit $(\xi^{\prime},\Lambda^{\prime})$ is isochrone but generally not
Keplerian. In its natural bolsted frame $\mathcal{R}^{\prime}_{O}=(O,\vec
u,\vec v)$ it is a Keplerian orbit with angular momentum $\Lambda$. If $\xi
\xi^{\prime}> 0$ then\footnote{When $\gamma<0$, $\Lambda$ is imaginary and
does not correspond to a \textsc{pro}.}
%$\vec {w^\prime_{\Lambda^\prime}}= - \gamma
%\Lambda^2 \vec j $ and%
\[
\Lambda^{\prime}= \sqrt{\gamma} \Lambda
\]
else
%$\vec {w^\prime_{\Lambda^\prime}}= - \Lambda^2 \vec j $ and%
\[
\Lambda^{\prime}= \Lambda.
\]

\end{theorem}

\begin{proof}
In the affine coordinate system $\left(  w_{1}=|\xi|x,w_{2}=y\right)$, an
orbit of energy $\xi<0$ and angular momentum $\Lambda$ corresponds to an arc
of the parabola $\mathcal{P}$. When this orbit is a \textsc{pro}, the two
extremities $A$ and $P$ of this parabolic arc are associated with the two
solutions apoastron $\vec{w}_{A}$ and periastron $\vec{w}_{P}$ of the equation $\frac{d}{d\tau
}\left(  \vec{w},\vec{\mathit{i}}\right)  =0$, with $d\tau=|\xi|dt$.
Considering the orbital differential equation~\eqref{orbital-differential-equation} 
in the coordinates $(|\xi|x,y)$, these
two extremal points of the orbit are on the extremal line
\[
\Delta=\left \{  \vec{w}\in \mathbb{R}^{2},\sqrt{2}\left(  \vec{w}%
|\vec{\mathit{k}}\right)  =\Lambda^{2}\right \}  .
\]
Trivially we then note that $\Delta$ is parallel to $\mathbb{R}\vec
{\mathit{k}}$. As the vectors $\vec{w}$ defining this \textsc{pro} satisfy
$\left[  \frac{d}{d\tau}\left(  \vec{w},\vec{\mathit{i}}\right)  \right]
^{2}\geq0$, they are periodic-like vectors. Defining
\[
K=\mathcal{T}_{O}\left(  \mathcal{P}\right)  \cap \Delta,
\]
the point $K$ is the $\vec{\mathit{k}}$ parallel projection of $A$ and $P$ on
$\mathbb{R}\vec{\mathit{j}}$ and trivially,
\begin{equation}
\overrightarrow{OK}=-\Lambda^{2}\vec{\mathit{j}}.\label{OK}%
\end{equation}
\begin{figure}[ptb]
\centering
\resizebox{\textwidth}{!}{\includegraphics{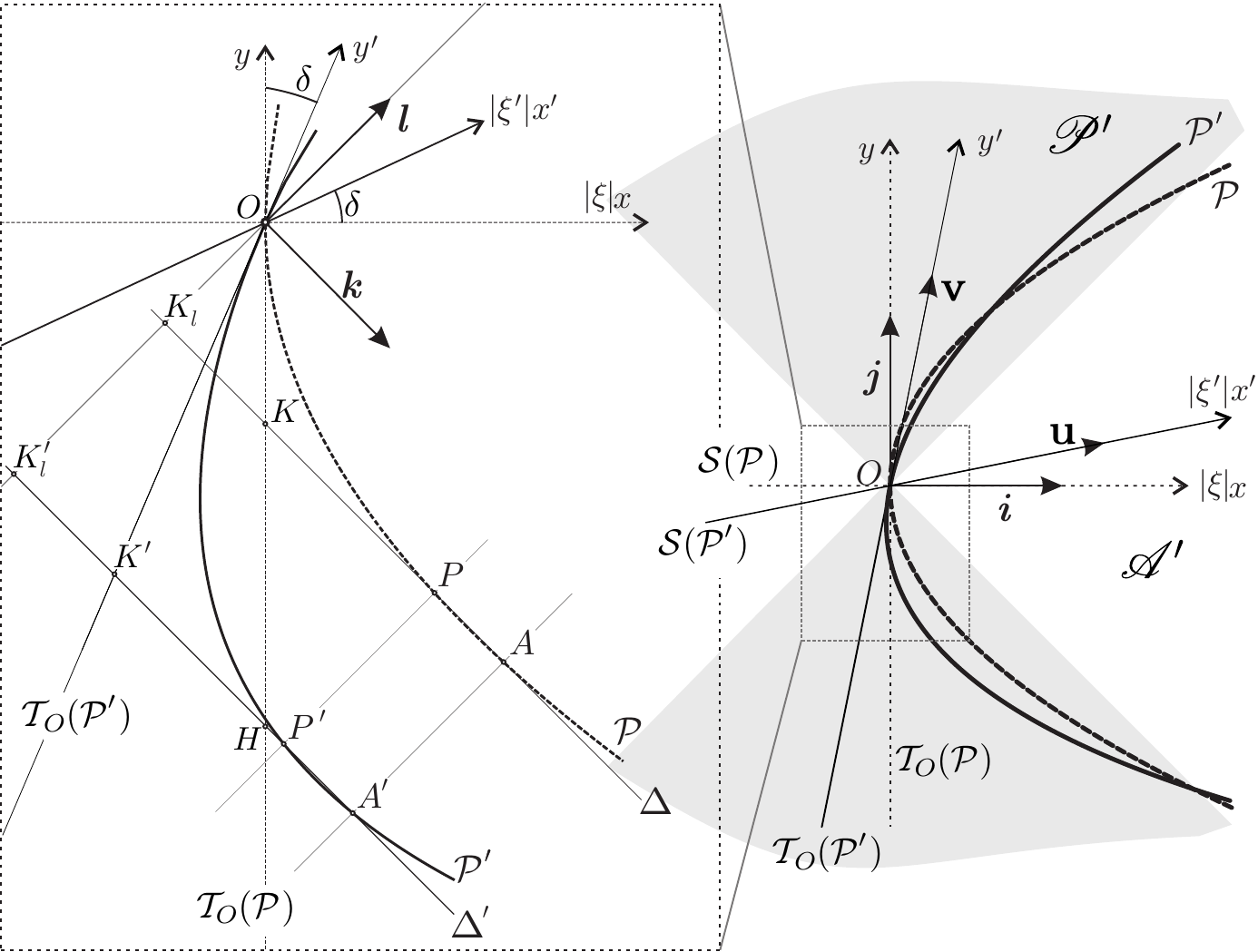}}
\centering
\caption{The $\gamma>1$ ibolst of the Kepler parabola when $\xi\xi^{\prime}%
>0$.}%
\label{iboostkepler}%
\end{figure}When $\gamma>1$, the ibolst of the Kepler parabola is represented
in figure \ref{iboostkepler}; the other values of $\gamma$ can be deduced
directly from figure \ref{vecteursuv}\ and the analysis we will give below. With
natural notations, we set $\mathcal{P}^{\prime}=B_{\gamma}\left(
\mathcal{P}\right)  $ and $\Delta^{\prime}=B_{\gamma}\left(  \Delta \right)  $.
As $\Delta$ is parallel to $\mathbb{R}\vec{\mathit{k}}$, which is an invariant
direction of the ibolst, $\Delta^{\prime}$ is also parallel to $\mathbb{R}%
\vec{\mathit{k}}$. Let us consider $K^{\prime}=B_{\gamma}\left(  K\right)$.
According to lemma \ref{lemme_des_droites} we have%
\begin{align*}
K^{\prime} &  =B_{\gamma}\left(  \mathcal{T}_{O}\left(  \mathcal{P}\right)
\cap \Delta \right)  \\
&  =B_{\gamma}\left(  \mathcal{T}_{O}\left(  \mathcal{P}\right)  \right)  \cap
B_{\gamma}\left(  \Delta \right)  \\
&  =\mathcal{T}_{O}\left(  \mathcal{P}^{\prime}\right)  \cap \Delta^{\prime}%
\end{align*}
and quantitatively, as $\overrightarrow{OK}=-\Lambda^{2}\vec{\mathit{j}}$,
after an ibolst, we get
\begin{equation}
\overrightarrow{OK}^{\prime}=-\Lambda^{2}\vec{v}.\label{OKprime}%
\end{equation}
This relation clearly indicates that $\Lambda$ is the same angular momentum
for the Keplerian orbit and for the ibolsted orbit when it is considered in
the reference frame of its ibolsted parabola, where it is also a Keplerian one.
In addition,
\begin{equation}
\label{eq:Deltalin}
\Delta^{\prime}=K^{\prime}+\mathbb{R}\vec{\mathit{k}}=B_{\gamma}(K)+\mathbb{R}%
B_{\gamma}(\vec{\mathit{k}}).
\end{equation}
Therefore, each point of the isochrone orbit on its arc of parabola is directly
linked by the ibolst to its $\vec{\mathit{l}}-$parallel projection on the Keplerian parabola .
We can determine the angular momentum $\Lambda^{\prime}$ of the
isochrone orbit in the Keplerian coordinates. Combining~\eqref{eq:uvbasis}, \eqref{OK} and
\eqref{OKprime} we get%
\[
\overrightarrow{KK}^{\prime}=-\Lambda^{2}\left(  \vec{v}-\vec{\mathit{j}%
}\right)  =-\frac{\Lambda^{2}}{\sqrt{2}}\left(  \gamma-1\right)
\vec{\mathit{l}}.
\]
If we introduce now the two orthogonal projections $K_{l}$ and $K_{l}^{\prime
}$ of $K$ and $K^{\prime}$ on $\mathbb{R}\vec{\mathit{l}}$, we have
$\overrightarrow{KK}^{\prime}=\overrightarrow{K_{l}K_{l}}^{\prime}$ and 
$\overrightarrow{OK_{l}}^{\prime} = B_\gamma\left( \overrightarrow{OK_{l}} \right)$ by~\eqref{eq:Deltalin}. Since 
$K_l=\Delta\cap\mathbb{R}\vec{\mathit{l}}$, $K_l^\prime=\Delta^\prime\cap\mathbb{R}\vec{\mathit{l}}$
and $B_\gamma$ sends $\Delta$ to $\Delta^\prime$ and $\vec{\mathit{l}}$ to $\gamma\vec{\mathit{l}}$, we 
get $$ \overrightarrow{OK_{l}}^{\prime} = B_\gamma\left( \overrightarrow{OK_{l}} \right)=\gamma \, \overrightarrow{OK_{l}} .$$
%but doing
%a little algebra we have%
%\begin{align*}
%\overrightarrow{OK_{l}}^{\prime} &  =\left(  \overrightarrow{OK^{\prime}}%
%|\vec{\mathit{l}}\right)  \vec{\mathit{l}}=-\Lambda^{2}\left(  B_{\gamma
%}\left(  \vec{\mathit{j}}\right)  |\vec{\mathit{l}}\right)  \vec{\mathit{l}}\\
%&  =-\Lambda^{2}\left(  \vec{\mathit{j}}|B_{\gamma}\left(  \vec{\mathit{l}%
%}\right)  \right)  \vec{\mathit{l}}\\
%&  =-\frac{\gamma \Lambda}{\sqrt{2}}^{2}\vec{\mathit{l}}%
%\end{align*}
%Hence
%\[
%\overrightarrow{OK_{l}}=\overrightarrow{OK_{l}^{\prime}}-\overrightarrow
%{K_{l}K_{l}}^{\prime}=-\frac{\Lambda^{2}}{\sqrt{2}}\vec{\mathit{l}}%
%\]
And finally by Thales theorem,%
\begin{equation}
\frac{OK_{l}^{\prime}}{OK_{l}}=\frac{OH}{OK}=\gamma \quad \text{where
}H=\mathcal{T}_{O}\left(  \mathcal{P}\right)  \cap \Delta^{\prime
}.\label{thales}%
\end{equation}
The length $OH$ is the squared
angular momentum $\Lambda^{\prime 2}$ of the ibolsted orbit considered in the reference frame of
the Keplerian parabola. As $OK$ is the squared angular momentum of the
Keplerian orbit in its natural frame, we have
\[
\Lambda^{\prime2}=\gamma \Lambda^{2}.
\]
When $\xi \xi^{\prime}<0$, the orientation of $\vec{u}$ is inverted. The line
$\Delta^{\prime}$ is $\Delta^{\prime}=K^{\prime}+\mathbb{R}\vec{\mathit{l}}$, and since
$\overrightarrow{OK_{l}}%=\frac{\Lambda^{2}}{\sqrt{2}}\vec{\mathit{k}}$, 
$ is directed by $\vec{\mathit{k}}$, then
$\Lambda^{\prime}=\Lambda$. \qed

\end{proof}

%Using the relativity formalism and introducing the proper time of the orbit $d\tau
%=\xi dt$, in the coordinates $(\xi x,y)$ in the canonical frame $\mathcal{R}_{O}$ the orbital differential
%equation~\eqref{fundamentalode} writes
%\begin{equation}
%\frac{1}{16}\left[ \frac{d}{d\tau }\left( \vec{w}|\vec{\mathit{i}}\right) %
%\right] ^{2}=\left( \vec{w}|\vec{\mathit{i}}\right) -\left( \vec{w}|\vec{%
%\mathit{j}}\right) -\Lambda ^{2}=\text{ }\sqrt{2}\left( \vec{w}|\vec{\mathit{%
%k}}\right) + \left(\vec w_\Lambda | \vec j \right)  \label{orbital-differential-equation}
%\end{equation}
%where $\vec w_\Lambda = (0,-\Lambda^2)^\top$.

The bolsted orbital differential equations follow from
theorem~\ref{thm:orbitboost}. As we can see from~\eqref{eq:orbitrelativity},
in isochrone relativity, orbital laws are the same in all 
reference frames.

\begin{corollary}
In the canonical frame $\mathcal{R}_{O}$, the bolsted orbital differential
equation is
\begin{equation}
\frac{1}{16}\left[  \frac{d}{d\tau}\left(  \vec{w^{\prime}}|\vec{\mathit{i}%
}\right)  \right]  ^{2}=\text{ }\left(  \vec{w^{\prime}}|\vec{\mathit{i}}%
-\vec{\mathit{j}}\right)  +\left(  \vec{w_{\Lambda^{\prime}}}|\vec{\mathit{j}%
}\right)  \label{orbital-differential-equation-img}%
\end{equation}
with $\vec{w_{\Lambda^{\prime}}}=-\Lambda^{\prime2}\vec{\mathit{j}}$.\newline
In the bolsted frame $\mathcal{R}_{O}^{\prime}$ with affine coordinates
$(\xi^{\prime}x^{\prime},y^{\prime})$ and proper time $d\tau^{\prime}%
=\xi^{\prime}dt^{\prime}$, the bolsted orbital differential equation is
\begin{equation}
\frac{1}{16}\left[  \frac{d}{d\tau^{\prime}}\left(  \vec{w^{\prime}}|\vec
{u}\right)  \right]  ^{2}=\text{ }\left(  \vec{w^{\prime}}|\vec{u}-\vec
{v}\right)  +\left(  \vec{w_{\Lambda}^{\prime}}|\vec{v}\right)
\label{eq:orbitrelativity}%
\end{equation}
with $\vec{w_{\Lambda}^{\prime}}=-\Lambda^{2}\vec{v}$.
\end{corollary}

%As reported in our classification only $\psi_{\mathrm{he}}^{+}$ and
%$\psi_{\mathrm{bo}}^{+}$ are physical potentials. However $\psi
%_{\mathrm{he}}^{-}$ and $\psi_{\mathrm{bo}}^{-}$ could host
%\textsc{pro.} However not every image does
%contain \textsc{pro} .

Isochrone \textsc{pro} are contained in the periodic-space of their
parabola reference frame. But in the Keplerian primary frame this periodic-space 
appears vertical when $\gamma>0$ and horizontal otherwise. In some cases, the
\textsc{pro} is then associated with an arc of parabola which is concave or
located in the negative part of the Keplerian frame. Those image orbits are
not physical. %\textit{a supprimer si pas dans la figure: and explain the dissymmetry in figure~\ref{fig:Iboost-diagram}.}
\\

\subsubsection{Potentials relativity}

The Keplerian nature of an isochrone potential is revealed in the reference frame 
$\mathcal{R}_{O}^\prime$ of its parabola, cf theorem~\ref{thm:orbitboost}. An ibolst can also bolst a harmonic potential
and then exactly provide the appropriate primary frame which characterizes the radial oscillation of a \textsc{pro} in the image isochrone potential.
In such a frame, all periods of \textsc{pro} have indeed the same value.

We give hereafter an explicit formulation of the parameters of all the image
potentials. 
They can be obtained by direct resolution of quadratic equations. 

When the primary potential is Keplerian $\psi_{\mathrm{ke}}\left(
r\right)  =-\frac{\mu}{r}$, the primary orbits are such that $\xi<0$ in order
to be bounded. If $\gamma>0$, the ibolsted potential $\psi^{\prime}\left(
r^{\prime}\right)  $ is always a transvection of a H{\'e}non or a bounded
isochrone potential introduced in sec.~\ref{subsec:isochroneclassification}. Using the notations of reduced potentials, coming from the proof of theorem~\ref{theo1} and equation~\eqref{property+-}, one can verify
that
\begin{equation}%
\begin{tabular}
[c]{|c|c|c|}\hline
& $%
\begin{array}
[c]{c}%
\  \\
\
\end{array}
\mathrm{sign}\left(  \xi^{\prime}\right)  =-\mathrm{sign}\left(  \xi \right)
>0$ & $%
\begin{array}
[c]{c}%
\  \\
\
\end{array}
\mathrm{sign}\left(  \xi^{\prime}\right)  =\mathrm{sign}\left(  \xi \right)
<0$\\ \hline
$\gamma>1$ & $%
\begin{array}
[c]{c}%
\  \\
\
\end{array}
J_{\epsilon,0}\left(  \psi_{\mathrm{bo}}^{+}\right) = \psi_{\mathrm{bo}}^{+}+\epsilon$ & $%
\begin{array}
[c]{c}%
\  \\
\
\end{array}
J_{\epsilon,0}\left(  \psi_{\mathrm{he}}^{-}\right)=\psi_{\mathrm{he}}^{-}+\epsilon  $\\ \hline
$0<\gamma<1$ & $%
\begin{array}
[c]{c}%
\  \\
\
\end{array}
J_{-\epsilon,0}\left(  \psi_{\mathrm{bo}}^{-}\right)= \psi_{\mathrm{bo}}^{-}-\epsilon $ & $%
\begin{array}
[c]{c}%
\  \\
\
\end{array}
J_{-\epsilon,0}\left(  \psi_{\mathrm{he}}^{+}\right)= \psi_{\mathrm{he}}^{+} -\epsilon$\\ \hline
\end{tabular}
\end{equation}
where
\[
\epsilon=\frac{\mu^{\prime}\left(  \gamma+1\right)  ^{2}}{8\gamma b}%
>0\text{,}\  \  \mu^{\prime}=\left \vert \frac{8\mu \xi^{\prime}\gamma}{\left(
\gamma+1\right)  \sqrt{\left \vert 8\xi \xi^{\prime}\left(  \gamma+1\right)
\right \vert }}\right \vert \  \  \text{and\  \ }b=\left \vert \frac{\mu \left(
\gamma-1\right)  }{\sqrt{\left \vert 8\xi \xi^{\prime}\left(  \gamma+1\right)
\right \vert }}\right \vert .
\]
Then, when the primary potential $\psi \left(  r\right)  $ is the harmonic
$\psi_{\mathrm{ha}}\left(  r\right)  =+\frac{1}{2}\omega^{2}r^{2}$, the
primary energy is positive $\xi>0$ in order to get bounded orbits. When $\gamma
>0$, the ibolst leads to the four increasing potentials $\psi^{\prime
}\left(  r^{\prime}\right)$,
\begin{equation}%
\begin{tabular}
[c]{|c|c|c|}\hline
& $%
\begin{array}[c]{c}%
\  \\
\
\end{array}
\mathrm{sign}\left(  \xi^{\prime}\right)  =\mathrm{sign}\left(  \xi \right)
>0$ & $%
\begin{array}
[c]{c}%
\  \\
\
\end{array}
\mathrm{sign}\left(  \xi^{\prime}\right)  =-\mathrm{sign}\left(  \xi \right)
<0$\\ \hline
$\gamma>1$ & $%
\begin{array}
[c]{c}%
\  \\
\
\end{array}
\ J_{\epsilon,0}\left(  \psi_{\mathrm{he}}^{+}\right)=\psi_{\mathrm{he}}^{+}+\epsilon  $ & $%
\begin{array}
[c]{c}%
\  \\
\
\end{array}
J_{\epsilon,0}\left(  \psi_{\mathrm{bo}}^{-}\right)=\psi_{\mathrm{he}}^{+}+\epsilon  $\\ \hline   % bo - ?
$0<\gamma<1$ & $%
\begin{array}
[c]{c}%
\  \\
\
\end{array}
J_{-\epsilon,0}\left(  \psi_{\mathrm{bo}}^{+}\right) =\psi_{\mathrm{bo}}^{+}-\epsilon $ & $%
\begin{array}
[c]{c}%
\  \\
\
\end{array}
J_{-\epsilon,0}\left(  \psi_{\mathrm{bo}}^{+}\right) =\psi_{\mathrm{bo}}^{+}-\epsilon $\\ \hline   % he - ?
\end{tabular}
\  \  \label{gammaneg}%
\end{equation}
where
\[
\epsilon=\frac{\mu^{\prime}\left(  \gamma-1\right)  ^{2}}{8b\gamma}%
,\  \  \mu^{\prime}=\left \vert \frac{4\xi^{\prime}\xi \gamma}{\omega \left(
\gamma-1\right)  \sqrt{\left \vert \xi^{\prime}\left(  \gamma-1\right)
\right \vert }}\right \vert \  \  \text{and\  \ }b=\frac{\left(  \gamma+1\right)
 \left \vert \xi \right \vert }{2\omega \sqrt{\left \vert \xi^{\prime}\left(
\gamma-1\right)  \right \vert }}.%
\]
The classical Bohlin transformation $B_{-1}$ exchanges the two potentials $\psi_\mathrm{ke}$
and $\psi_\mathrm{ha}$, cf. theorem~\ref{prop:Bohlin} p.\pageref{prop:Bohlin}. The commutative structure and 
associative property of the group $\mathbb{B}$ then provide the image of any isochrone potential
by $B_\gamma$ when $\gamma<0$.

\begin{figure}[h]
\centering \resizebox{1.0\textwidth}{!}{\includegraphics{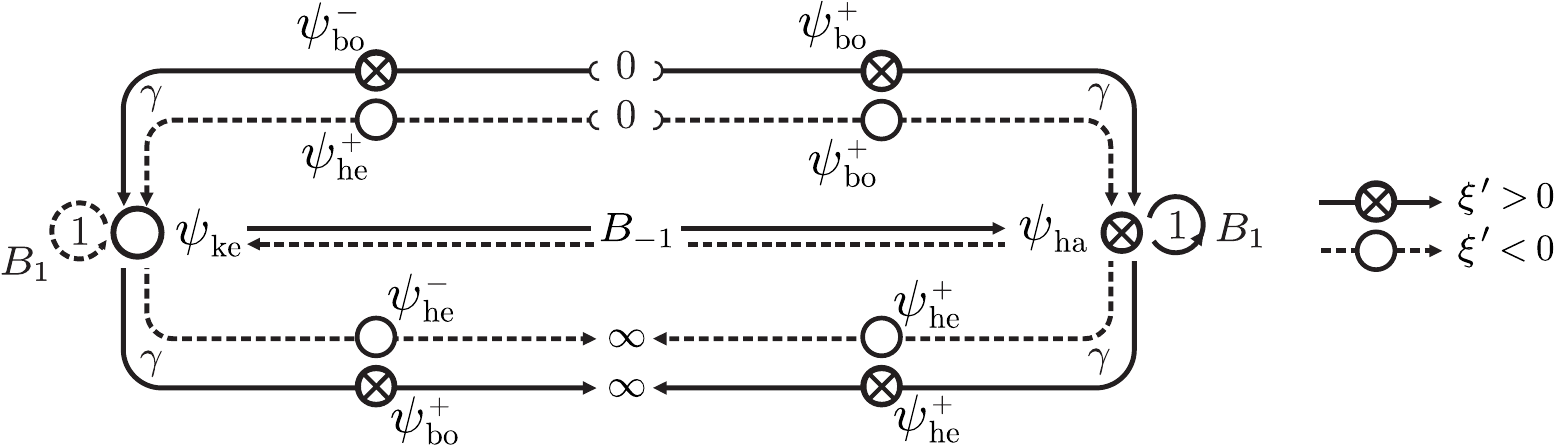}}\caption{Diagram
of the set of all possible ibolsted potentials up to an additive constant. The
nomenclature used is the one defined in the isochrone classification of 
sec. \ref{subsec:isochroneclassification}.}%
\label{fig:Iboost-diagram}%
\end{figure}

 A transvection $J_{\epsilon,0}$ swivels a parabola when it only adds
the constant $\epsilon$ to the corresponding potential. This constant has no
particular role and we can neglect it in a potential diagram summarizing the
effect of the ibolst on isochrone potentials. This is the purpose of figure~\ref{fig:Iboost-diagram}.

Using this diagram and the group property of the ibolst, we can recover all
ibolsted potentials only from the Keplerian one. Isochrone potentials form
the group orbit of Kepler potentials under the action of $\mathbb{B}$.

\subsubsection{Isochrone orbits construction\label{isoconstruct}}

Isochrony is a Keplerian property seen from an appropriate %change of 
reference frame.
Theorem~\ref{thm:orbitboost} gives a method to find the relative isochrone reference frame 
from a Keplerian potential. From any isochrone potential
one may reciprocally construct its isochrone orbits and find their related Keplerian description 
graphically using parabolas.

In order to be concrete, we build now the complete \emph{back to the Kepler} process
 when the needed ibolst has 
$\gamma > 1$ for $\xi\xi^\prime>0$
in figure~\ref{fig:soupe_corbeau}. 
\begin{figure}[h]
\centering \resizebox{0.95\textwidth}{!}{\includegraphics{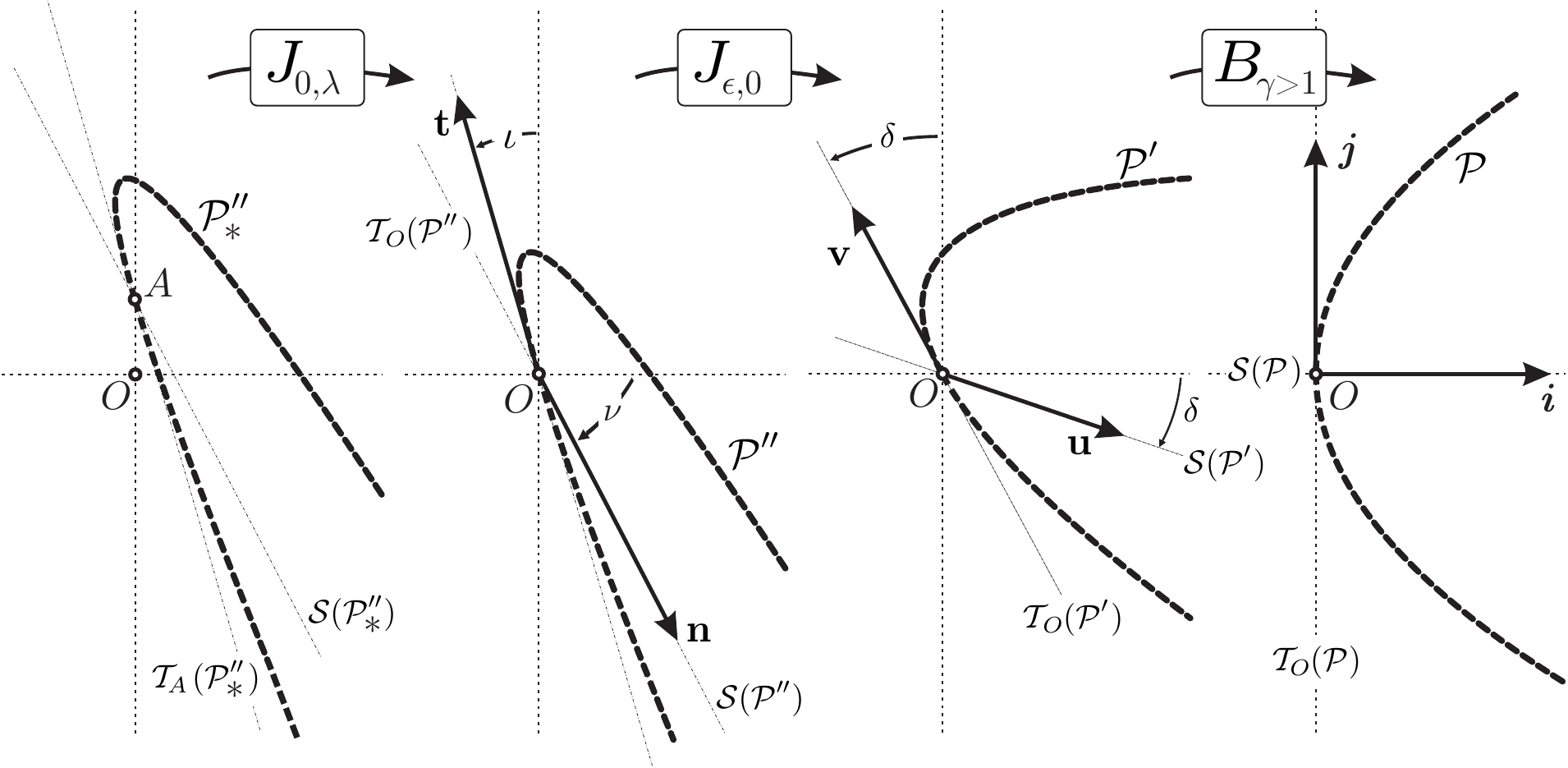}}\caption{Graphical
determination of the Keplerian reference frame of an isochrone parabola. Here isochrone orbits have negative energy.}%
\label{fig:soupe_corbeau}%
\end{figure}

Consider a parabola $\mathcal{P^{\prime\prime}_\ast}$. From definition~\ref{def:redphysgau} p.\pageref{def:redphysgau}, we retrieve a \emph{physical} parabola 
from a vertical translation $J_{0,\lambda}
(\mathcal{P}^{\prime\prime}_\ast) = \mathcal{P}^{\prime\prime}$.
Then, by definition~\ref{def:refframe} p.\pageref{def:refframe}, we find the natural frame $\left( O,\vec t, \vec n \right)$ attached to $\mathcal{P}^{\prime\prime}$.
While the angles $\iota$ and $\nu$ are not equal, we adjust the parabola with a transvection $J_{\epsilon,0}$ to prepare it for a bolst. We then debolst
the parabola with the ibolst $B_\gamma$ given by the angle $\delta=\iota=\nu$, with $\delta<\frac{\pi}{2}$. 
Given $\mathcal{P}$ and $\mathcal{P}^{\prime}$, the isochrone orbit can be related to its Keplerian description as in figure~\ref{iboostkepler}.

This geometrical construction gives the radial distances 
$$
r_0=\frac{1}{2\sqrt{|\xi_0|}} \sqrt{\left( \vec w | \vec{\mathit{i}}\right)}
\ \text{ and } \
r_1=\frac{1}{2\sqrt{|\xi_0|}} \sqrt{\left( \vec w^\prime | \vec{\mathit{i}}\right)}
$$ 
of the Keplerian and isochrone orbits in the $\left( |\xi_0| x, y \right)$-coordinates of the Keplerian frame.
The angles $\varphi_0$ and $\varphi_1$ are provided
by theorem~\ref{prop:Bohlin} p.\pageref{prop:Bohlin} and given by 
$$
a=\frac{1}{2}\left(r_{0,p}+r_{0,a}\right), \
\frac{|\xi_0|}{\mu} = \frac{1}{2a}, \
e=\frac{r_{0,a}-r_{0,p}}{2a}, \
p=\left(1-e^2\right) a, \ \text{ and } \
\frac{\alpha}{\beta} = \frac{\gamma+1}{\gamma-1}.
$$
They can also be geometrically determined. In fact, the precession of the
isochrone apocenters or pericenters $n_\varphi$ depends on $\Lambda$ and the ordinate of the 
intersection of the convex part of the parabola and $\mathbb{R}\vec{\mathit{j}}$, 
see proposition~\ref{prop:Tniso}. This intersection is given by the vertical
translation parameter $\lambda$ and the aperture of the parabola; more precisely, by the distance
$4b\mu$ between the two intersections of the parabola and the axis $\mathbb{R}\vec{\mathit{j}}$, just as one can deduce from~\eqref{property+-} and its following properties on page~\pageref{property+-}. 

This construction does not explicitly depend on the
hypothesis $\gamma>1$, and can be generalized to other values of $\gamma$ as
long as the considered initial orbit is a \textsc{pro}, i.e. $\vec{w}^{\prime}$
remains a periodic-like vector on the convex part of a parabola. It is also possible to construct positive
energy ibolsted orbits from negative energy Keplerian orbits.
\\

This procedure can also be generalized using a bolst $B_{\alpha,\beta}$, which is 
a transvection % of parameter $\alpha-\beta-1$ 
of an ibolst $B_{\gamma}$ when expressed in the  basis 
$\left(  \vec{\mathit{l}},\vec{\mathit{k}}\right)  $. 
In the same way, the first translation $J_{0,\lambda}$ is not compulsory.

\section{Applications\label{sec4}}

\subsection{Physical properties of isochrone potentials\label{subsec_physapp}%
}

Up to an affine transformation, there are four different increasing potentials
which are isochrone, i.e. in which the radial periods $\tau_{r} $ only depend
on the energy of the considered radially oscillating particles. Two of them
are very well known: the Kepler potential $\psi_{\mathrm{ke}}$ is associated with a Dirac
density distribution and the harmonic potential $\psi_{\mathrm{ha}}$ is sourced by a
constant density distribution of matter in the considered volume. In figure~\ref{fig:bo_and_he} we present the plot of the two other ones, i.e.
$\psi_{\mathrm{bo}}$ and $\psi_{\mathrm{he}}$. Notice their harmonic quadratic behavior
at small radial distances.

\begin{figure}[h]
\centering
\resizebox{0.97\textwidth}{!}{\includegraphics{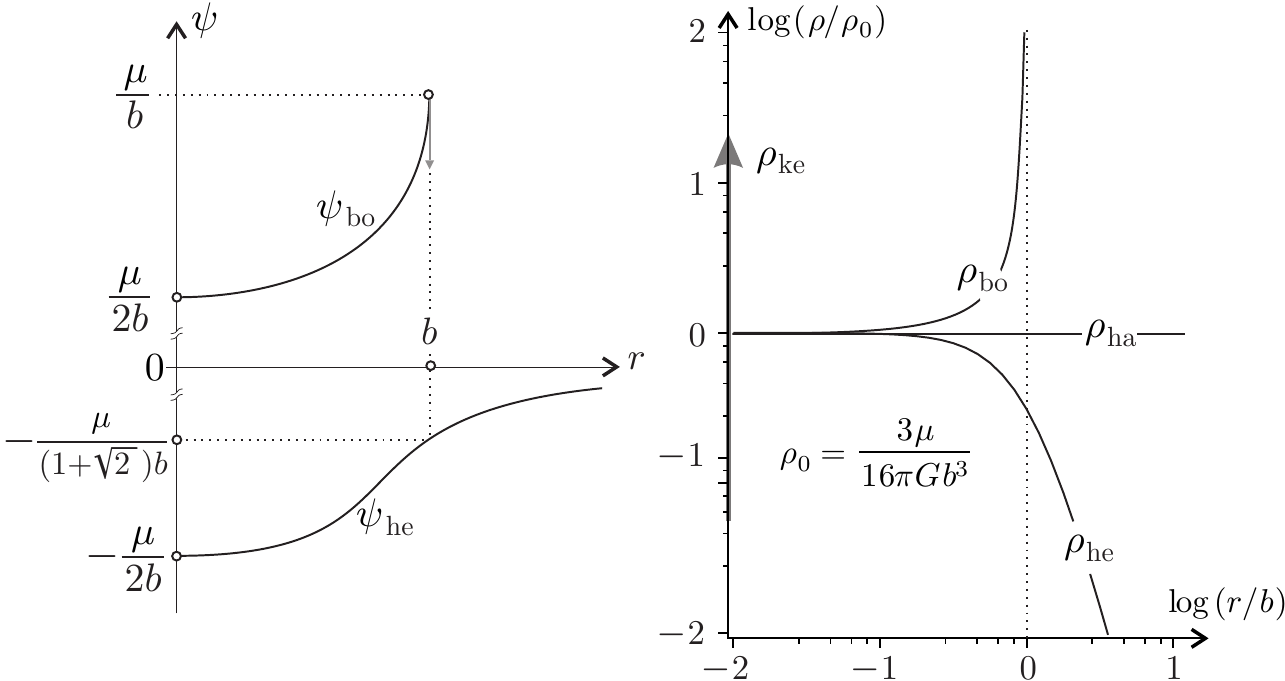}}\caption{The
bounded and the H\'{e}non isochrone potential (left). The mass density of isochrones (right).}%
\label{fig:bo_and_he}%
\end{figure}

The H\'{e}non potential $\psi_{\mathrm{he}}$ has important physical
properties in gravitational stellar dynamics: in a forthcoming paper in preparation by Simon-Petit et al., we will
show that it appears to be a fundamental equilibrium state where stellar
systems settle down after violent relaxation (e.g. \cite{VioRel} for the
original contribution and \cite{BT08} p. 380-382 for a modern review). The corresponding
density is a core-halo structure: the typical size of the core is the length
$b$ and the surrounding halo falls like a $r^{-4}$ power law. This property ensures that the mass $M_{\mathrm{he}}(r)$ contained in any ball of radius $r$ in a H\'{e}non potential is finite. As a matter of fact, by Gauss' theorem, we have $GM_{\mathrm{he}}(r)=r^2\frac{d\psi_{\mathrm{he}}}{dr}$ and $\lim_{r\to\infty}GM_{\mathrm{he}}(r)=\mu$. Recalling definition \ref{def:redphysgau}, this finite mass property is trivially conserved for the reduced version of the H\'{e}non potential 
$\psi_{\mathrm{he}}^{\mathrm{red}}=\psi_{\mathrm{he}}^{+}=\psi_{\mathrm{he}}+\frac{\mu}{2b}$ and for all physical H\'{e}non's $\psi_{\mathrm{he}}^{\mathrm{phy}}=\psi_{\mathrm{he}}+\epsilon$ for any real $\epsilon$. However, the gauged H\'{e}non $\psi_{\mathrm{he}}^{\mathrm{gau}}=\psi_{\mathrm{he}}^{\mathrm{phy}}+\frac{\lambda}{2r^2}$ contains an infinite mass in its center and has poor physical meaning. Nevertheless, this latter potential is still isochrone. As we said in the classification of the sec. \ref{subsec:isochroneclassification}, gauged potentials are essential for the completeness of the isochrone set.

The properties of systems associated with the $\psi_{\mathrm{bo}}$ potential
are more unusual. When it is considered on its whole domain $\mathcal{D}_{\psi_{\mathrm{bo}}}=[0,b]$, the systems have an infinite total mass. As a matter of fact, 
$GM_{\mathrm{bo}}(r)\sim\mu\sqrt{\frac{b}{2(b-r)}}$ when $r\to b$. This property holds for any physical bounded potential. 
In fact these systems are self-confined because there exists an infinite repulsive force at their
boundaries in $r=b$. Perhaps $\psi_{\mathrm{bo}}$ potentials might be
used as classical models for structurally confined systems like, for example,
quarks in the nucleon. Indeed, such fundamental particles are confined in the
nucleon (here of size $b$) and are characterized by asymptotic freedom, i.e.
they do not feel any force at the center of the nucleon.
Gauged bounded potentials are even more unusual with their infinite central mass!

The repartition of mass in physical isochrones is progressive: the mass is concentrated into a point in the center of a Kepler system, while in a H\'{e}non one, the mass is 
equally distributed up to a characteristic length settled by the parameter $b$, and in a less concentrated decreasing repartition after the characteristic radius. 
When $b$ increases, the first dense harmonic part grows and the H\'{e}non potential 
eventually behaves like a harmonic potential since
\begin{equation}
\psi_{\mathrm{he}}^{\mathrm{red}}\underset{b\rightarrow\infty}{\sim} \frac{\mu}{8b^3} r^2,
\label{hebgrand}
\end{equation}
i.e. the physical H\'{e}non isochrone is changed into the physical harmonic when $b\to +\infty$. 
This property can be easily seen on the mass density distribution in the right panel of the figure \ref{fig:bo_and_he}. Subsequently, since 
\begin{equation}
\psi_{\mathrm{bo}}^{\mathrm{red}}\underset{b\rightarrow\infty}{\sim} \frac{\mu}{8b^3} r^2
\label{bobgrand}
\end{equation}
we can say in a converse manner that when the infinite mass of the unbounded harmonic is concentrated into a finite domain of size $b$. We can recover the bounded isochrone by controlling $b$. 

Let us revisit the properties of orbits. 

\subsection{Period and precession of periastron for isochrones}

Proposition~\ref{prop:Tniso} gathers the properties $\tau_r$ and $n_\varphi$ of isochrone orbits and
reveals the interesting similarities of isochrone radial periods. 
Their form in $\psi_\mathrm{he}$ and $\psi_\mathrm{bo}$ is the same as in the Keplerian potential.
We will use this remark to generalize Kepler's third law in the next subsection. In a
harmonic potential, $\tau_r$ is the same regardless of the energy of the massive particles.
Moreover, in $\psi_\mathrm{ke}$ and $\psi_\mathrm{ha}$, $n_\varphi$ is rational
and all orbits are closed. 

\begin{proposition}
\label{prop:Tniso}
Given a \textsc{pro} $(\xi,\Lambda)$ in an isochrone potential, its radial and azimuthal
periods are 
\begin{equation}%
\begin{tabular}[c]{l|c|c|c|c|}\cline{2-5}
$\begin{array}[c]{c} \  \\ \ \end{array}$& ~~$ \psi_{\mathrm{ke}}$~~ & ~~$\psi_{\mathrm{ha}}$~~ & ~~$\psi_{\mathrm{he}}$ ~~ & ~~$\psi_{\mathrm{bo}}$ ~~\\ \hline
\multicolumn{1}{|l|}{$\begin{array}[c]{c} \  \\ \ \end{array} \tau_{r}$} & $2\pi \mu \left \vert 2\xi\right \vert
^{-3/2}$ & $\pi \omega^{-1}$ & $2\pi \mu \left \vert 2\xi\right \vert ^{-3/2}$
& $2\pi \mu \left \vert 2\xi\right \vert ^{-3/2} $\\ \hline
\multicolumn{1}{|l|}{$\begin{array}[c]{c} \  \\ \ \end{array} n_{\varphi}$} & $1$ & $\frac{1}{2}$ & $
\frac{1 }{2}+\frac{\Lambda}{2\sqrt
{4b\mu+\Lambda^{2}}}$ & $\frac{1 }{2}-\frac{\Lambda}{2\sqrt
{4b\mu+\Lambda^{2}}}$\\ \hline
\end{tabular}
\label{table_tau_n}%
\end{equation}
\end{proposition}

\begin{proof}
Using isochrone potential
expressions, the radial period~\eqref{radialperiod} and increment $n_\varphi$ of the azimuthal angle~\eqref{djdl}
come from the computation of the radial action
\[
\mathcal{A}_{r}=\frac{1}{\pi}{\int_{r_{p}}^{r_{a}}}\sqrt{2\left[  \xi-\psi \left(
r\right)  \right]  -\dfrac{\Lambda^{2}}{r^{2}}}dr.
\]
For a Keplerian orbit of energy $\xi_k<0$ in  $\psi_{\mathrm{ke}}\left(  r\right)  =-\frac{\mu}{r}$ 
and a harmonic orbit of energy $\xi_h>0$ in $\psi_{\mathrm{ha}}\left(  r\right)  =\frac{1}{2}\omega^{2}r^{2}$, we have
\begin{equation}
\label{eq:RAke}
\mathcal{A}_{r}^{\mathrm{ke}}=\frac{\sqrt{2\left \vert \xi_{k}\right \vert }}{\pi}{\int_{r_{p}%
}^{r_{a}}}\dfrac{\sqrt{\left(  r-r_{p}\right)  \left(  r_{a}-r\right)
}}{r}dr \text{ with\ }\left \{
\begin{array}
[c]{l}%
r_{p}+r_{a}=\frac{\mu}{\left \vert \xi_{k}\right \vert }\\
r_{p}r_{a}=\frac{\Lambda^{2}}{2\left \vert \xi_{k}\right \vert }%
\end{array}
\right.
\end{equation}
and%
\begin{equation} 
\label{eq:RAha}
\mathcal{A}_{r}^{\mathrm{ha}}=\frac{\sqrt{\mu}}{2\pi}{\int_{r_{p}^{2}}^{r_{a}^{2}}}%
\dfrac{\sqrt{\left(  u-r_{p}^{2}\right)  \left(  r_{a}^{2}-u\right)  }}%
{u}du \text{ with \ }\left \{
\begin{array}
[c]{l}%
r_{p}^{2}+r_{a}^{2}=\frac{2\xi_{h}}{\omega^{2}}\\
\left(  r_{p}r_{a}\right)  ^{2}=\frac{\Lambda^{2}}{\omega^{2}}.
\end{array}
\right.
\end{equation}
The computation of these radial actions can be done by meticulous integration to recover $\tau_r$ and $n_\varphi$ in $\psi_{\mathrm{ke}}$ and $\psi_{\mathrm{ha}}$.
Conversely, knowing the radial and azimuthal periods, one recovers the expression of $\mathcal{A}_{r}^{\mathrm{ke}}$ and $\mathcal{A}_{r}^{\mathrm{ha}}$.
Indeed, 
for $\psi_{\mathrm{ke}}$, $\tau_r$ follows from the classical Kepler's third law, and $n_\varphi=1$ because the center of attraction of a Keplerian
ellipse is located at one of its foci (see figure~\ref{twospecialcases}). For the harmonic potential, $\tau_r=\frac{\pi}{\omega}$  and $n_\varphi=\frac{1}{2}$ because harmonic ellipses are centered 
at their centers of attraction, see figure~\ref{twospecialcases}. 
As it is shown in appendix~\ref{appendixb}, one gets%
\[
\mathcal{A}_{r}^{\mathrm{ke}}=\frac{\mu}{\sqrt{2\left \vert \xi_{k}\right \vert }}-
\Lambda   \text{ and }\mathcal{A}_{r}^{\mathrm{ha}}=\frac{\xi_{h}}{2\omega}%
-\frac{\Lambda }{2}.%
\]
For the two non classical isochrones $\psi_{\mathrm{he}}%
^{\mathrm{bo}}\left(  r\right)  =\pm \frac{\mu}{b}\left(  1+\sqrt
{1\mp \frac{r^{2}}{b^{2}}}\right)  ^{-1}$, generalizing \cite{BT08} p.152, we introduce
$s=1+\sqrt{1\mp \frac{r^{2}}{b^{2}}}$. For the H{\'e}non potential, $s>2$ and
the \textsc{pro} has $\xi_{-}<0$ according to its effective potential, 
see sec.~\ref{subsecHenonparabola} p.\pageref{subsecHenonparabola}.
In the same way, for the bounded potential, $2>s>0$ and its \textsc{pro} has positive
energy $\xi^+>0$.
Then, for $s_p<s_a$, the radial actions are
\[
\mathcal{A}_{r,\mathrm{he}}^{\mathrm{bo}}=\mp\frac{b\sqrt{2\left \vert
\xi_{-}^+\right \vert }}{\pi}{\int_{ s_{p}  }%
^{ s_{a}}}\frac{(s-1)}{s\left(  s-2\right)  }\sqrt{\left(  s-s_{p} \right)  \left(  s_{a}
-s\right)  }ds
\]
with
\begin{equation}   
\label{eq:newvarsRA} % new variable s for the radial action
\left \{
\begin{array}
[c]{l}%
s_{p}+s_{a}=2+\frac{\mu}{b\left \vert \xi_{-}^{+}\right \vert }\\
s_{a}s_{p}=\frac{4b\mu+\Lambda^{2}}{2b^{2}\left \vert \xi_{-}^{+}\right \vert }.
\end{array}
\right.
\end{equation}
Hence, using $\mathcal{I}_{2}$ from appendix~\ref{appendixb}, one gets%
\[
\mathcal{A}_{r,\mathrm{he}}^{\mathrm{bo}}=\mp \frac{\mu}{\sqrt{2\left \vert
\xi_{-}^{+}\right \vert }}- \frac{1}{2}\left(  \Lambda
\mp\sqrt{4b\mu+\Lambda^{2}}\right)  . 
\]
The results follow by derivation. \qed
\end{proof}% end proof prop:Tniso

The dynamics is unchanged when adding constants to potentials, i.e. $\psi\to\psi+\epsilon$.
%Under an affine transformation, the orbits preserve their radial period and azimuthal precession.
However, the expression of the periods are modified and can be deduced from propositions~\ref{prop:Tniso} and~\ref{prop:TnaffineTransf}
for the reduced, physical and gauged isochrones.
%, their expression is not as simple but can be
%deduced from proposition~\ref{prop:Tniso} and~\ref{prop:TnaffineTransf}.
%
%Considering the behavior of the radial action under the affine
%transformations, the following property relates the radial action, period and
%precession of apo- and pericenter of an orbit and its affine transform.

\begin{proposition}
\label{prop:TnaffineTransf}
Let $\psi$ and $\psi^{\ast}$ be two potentials related by an affine
transformation $\psi^{\ast}= J_{\epsilon,\lambda}(\psi)=\psi+\epsilon+\frac{\lambda}{2r^2}$. 

An orbit defined in $\psi$ and its affine transformation in $\psi^\ast$
share the same orbital properties $\tau_r$ and $n_\varphi$.

Provided that $\lambda+\Lambda^{2} > 0$, the radial action and its derivatives are transformed as follows:

\begin{enumerate}
\item $\mathcal{A}_{r}^{\ast} ( \xi; \Lambda) = \mathcal{A}_{r}\left(  \xi- \epsilon; \sqrt{\lambda+
\Lambda^{2}} \right)  ,$

\item $\tau_{r}^{\ast} ( \xi; \Lambda) = \tau_{r} \left(  \xi- \epsilon;
\sqrt{\lambda+ \Lambda^{2}} \right)  $,

\item $n_{\varphi}^{\ast} ( \xi; \Lambda) = n_{\varphi} \left(  \xi- \epsilon;
\sqrt{\lambda+ \Lambda^{2}} \right)  \frac{\Lambda}{\sqrt{\lambda+ \Lambda^{2}}}$.
\end{enumerate}
\end{proposition}

\begin{proof}
The radial action of an orbit of energy $\xi$ and angular momentum $\Lambda$
in $\psi^{\ast}$ is given by
\[%
\begin{array}
[c]{ccl}%
\mathcal{A}_{r}^{\ast}(\xi;\Lambda) & = & \displaystyle \frac{1}{\pi}\int_{r_{p}^{\ast}(\xi;\Lambda
)}^{r_{a}^{\ast}(\xi;\Lambda)}\sqrt{2(\xi-\psi^{\ast}(r))-\frac{\Lambda^{2}%
}{r^{2}}}dr\\
& = & \displaystyle \frac{1}{\pi}\int_{r_{p}^{\ast}(\xi;\Lambda)}^{r_{a}^{\ast}(\xi
;\Lambda)}\sqrt{2(\xi-\epsilon-\psi(r))-\frac{\lambda+\Lambda^{2}}{r^{2}}}dr\\
& = & \displaystyle \frac{1}{\pi}\int_{r_{p}\left(  \xi-\epsilon;\sqrt{\lambda+\Lambda^{2}%
}\right)  }^{r_{a}\left(  \xi-\epsilon;\sqrt{\lambda+\Lambda^{2}}\right)  }%
\sqrt{2(\xi-\epsilon-\psi(r))-\frac{\lambda+\Lambda^{2}}{r^{2}}}dr\\
\mathcal{A}_{r}^{\ast}(\xi;\Lambda) & = & \mathcal{A}_{r}\left(  \xi-\epsilon;\sqrt{\lambda+\Lambda^{2}%
}\right),  
\end{array}
\]
where $r$ is the radial distance in the reference frame associated with $\psi$, and
$r^\ast$ is the image in the same frame by the affine transformation.
For the second relation we use the definition $\frac{\tau_{r}^{\ast} ( \xi;
\Lambda)}{2 \pi} = \frac{\partial \mathcal{A}_{r}^{\ast}}{\partial \xi} ( \xi; \Lambda)$. For
the third one, we get
\[%
\begin{array}
[c]{ccl}%
n_{\varphi}^{\ast} ( \xi; \Lambda) & = & - \frac{\partial \mathcal{A}_{r}^{\ast}}%
{\partial \Lambda} ( \xi; \Lambda)\\
& = & - \frac{\partial}{\partial \Lambda} \left(  \mathcal{A}_{r} \left(  \xi- \epsilon;
\sqrt{\lambda+ \Lambda^{2}} \right)  \right) \\
& = & - \frac{\partial \mathcal{A}_{r}^{}}{\partial \Lambda} \left(  \xi- \epsilon; \sqrt
{\lambda+ \Lambda^{2}} \right)  \times \frac{\Lambda}{\sqrt{\lambda+ \Lambda^{2}}}
.
\end{array}
\]
And the third relation follows. 

Eventually, a transformation $J_{\epsilon,\lambda}$ maps an orbit $(\xi, \Lambda)$ onto another one of
parameters $(\xi+\epsilon, \sqrt{\Lambda^2-\lambda})$ when $\Lambda^2-\lambda>0$. Inserting them in the previous
relations, we recover the invariance of $\tau_r$ and $n_\varphi$ under $J_{\epsilon,\lambda}$:
the radial period of the image orbit $\tau^{\ast}(\xi+\epsilon,\sqrt{\Lambda^2-\lambda})$ is that of the 
primary orbit $\tau(\xi,\Lambda)$. In the same way,  
$n_{\varphi}^{\ast} ( \xi+\epsilon; \sqrt{\Lambda^2-\lambda})) =n_\varphi(\xi,\Lambda)$. \qed
\end{proof}

Eventually, the radially periodic orbits are rosettes,~\cite{BT08} sect. 3. %sect. 3, p.147
The number $n_\varphi$ of revolutions to reach a periastron from the preceding one can be greater
or lower than for a harmonic or Keplerian potential.
A gauge introduces % the power law $1/r^\alpha$, $\alpha\geq2$.
orbits that spiral into the origin~\cite{LB:2015}, % p.\pageref{1964}
as it happens for orbits of the extremal line defining an imaginary radial distance on its parabola at the pericenter.
The gauged harmonic presents a similarity with $\psi_\mathrm{he}$ and $\psi_\mathrm{bo}$,
as described in proposition~\ref{prop:rosettes}. The precession of orbits that emerge when adding a $\frac{1}{r^2}$-term to the
potential corresponds to the one described in Proposition XLIV of Newton's \emph{Principia}~\cite{Newton}
for the Kepler force.

\begin{proposition}
\label{prop:rosettes}
 %A nice parallel between harmonic, bound and H\'enon orbits :
 Bounded, H{\'e}non and gauged harmonic \textsc{pro}'s are rosettes with azimuthal precessions $n_\varphi$ such that:
 $$
 \begin{array}{|c|c|}
    \hline
    \begin{array}[c]{c} \  \\ \ \end{array} \psi_\mathrm{he}\; \text{and}\;  J_{0,\lambda} (\psi_\mathrm{ha} )\; \text{with}\;\lambda>0 \;& 
		\begin{array}[c]{c} \  \\ \ \end{array}  \psi_\mathrm{bo}\; \text{and}\;  J_{0,\lambda} (\psi_\mathrm{ha} )\; \text{with}\;\lambda<0 \;
    \\
    \hline
    \begin{array}[c]{c} \  \\ \ \end{array} 
		\includegraphics[scale=0.9]{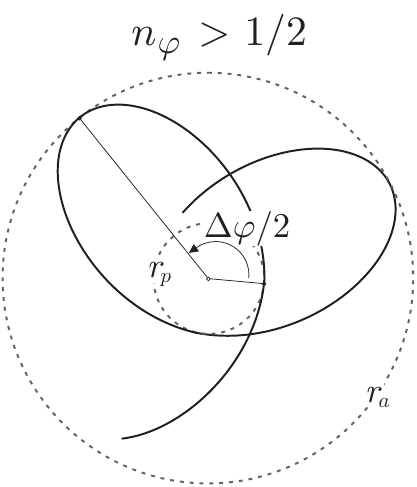}   &  
		\raisebox{-0.05\height}{\includegraphics[scale=0.9]{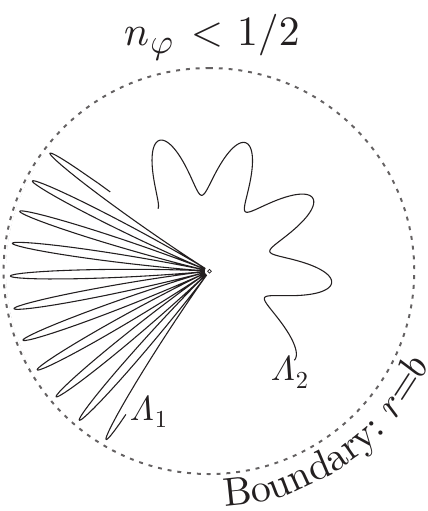}}
    \\
    \hline
 \end{array}
$$
\end{proposition}
\begin{proof}
Let us illustrate the case of a harmonic
oscillator and its gauge transform
$ \psi^{\ast} = J_{0,\lambda}( \psi_{\mathrm{ha}} ) = \psi_{\mathrm{ha}} + \frac{\lambda}{2 r^{2}} $.
From proposition~\ref{prop:TnaffineTransf}, we get that for the modified potential,
\[
\tau^{\ast}_{r} = \tau_{r} = \frac{\pi}{\omega}, \hspace{1em} n^{\ast
}_{\varphi} = \frac{\Lambda}{2 \sqrt{\lambda+ \Lambda^{2}}} .
\]
Thus, for harmonic potentials, adding a gauge modifies $n_{\varphi}$, whereas the period
never changes.
Moreover we get the dynamical consequence that $n^{\ast}_{\varphi}
< 1 / 2$ when $\lambda> 0$ and $n^{\ast}_{\varphi} > 1 / 2$ when $\lambda< 0$.

The parallel property exists for $\psi_\mathrm{he}$ and $\psi_\mathrm{bo}$. According to~\eqref{table_tau_n} in proposition~\ref{prop:Tniso},
$n_\varphi^\mathrm{bo} = \frac{1}{2}-\frac{\Lambda}{2\sqrt
{4b\mu+\Lambda^{2}}}$ where $\frac{\Lambda}{2\sqrt
{4b\mu+\Lambda^{2}}}>0$, and then $n_\varphi^\mathrm{bo}<\frac{1}{2}$. In the same way, $n_\varphi^\mathrm{he}>\frac{1}{2}$.

This $n_{\varphi}$ property shapes the corresponding orbits. On the one hand, when $n_{\varphi}>\frac{1}{2}$, the azimuthal precession $\Delta \varphi/2$ during the transfer from apoastron ($r=r_a$) to the periastron ($r=r_p$) is greater than $\pi$. Thus the orbit must turn around the center of the system as it is indicated on the left panel of the proposition. On the other hand, when $n_{\varphi}<\frac{1}{2}$, the transfer $r_p\to r_a\to r_p$ cannot turn around the center; such orbits oscillate between $r_a$ and $r_p$, precessing around the center, as is plotted on the right panel of the proposition. The smaller the value of the angular momentum, the tighter the oscillation is. On this right panel we have $0\simeq\Lambda_1<\Lambda_2$. 
 \qed
\end{proof}

Let us conclude this section remarking that the extension ($r_a$ and $r_p$) of an isochrone orbit is managed by its energy  (see the expression of $\xi$ in \eqref{eq:RAke}, \eqref{eq:RAha} and \eqref{eq:newvarsRA}) when the thickness of its oscillation is governed by its angular momentum. More precisely, radial (thin) orbits are obtained when $\Lambda\to0$ and circular (fat) orbits when $\Lambda=\Lambda_c$, the largest value possible of the angular momentum for the considered energy.

\subsection{Generalization of Kepler's Third Law\label{subsec_kep_third_law}
}

The Kepler potential $\psi_{\mathrm{ke}}\left(r\right)  =-\frac{\mu}{r}$
is sourced by a point of mass $M$ such that $\mu=GM$ where $G$ is
the Newton constant. Radially periodic orbits close after one
radial period $\tau_{r}$ and form ellipses with semi major axes $a=-\frac{\mu
}{2\xi}$. In his last major book \emph{Harmonices Mundi}~\cite{Kepler}, Johannes Kepler
proposed in 1619 his third law claiming that $\tau^{2}_r\times a^{-3}$ is
constant for all ellipses. Isaac Newton, half a century later,
proved this empirical observation using his laws of dynamics and his
gravitational force. This law appears to become a cornerstone of celestial
mechanics because the Kepler constant appears to be $\tau_{r}^{2}a^{-3}
=\frac{4\pi^{2}}{\mu}$ and thus gives the mass of the attracting body.

In this paper we have shown that Kepler potential generates the isochrone
group and we remark that Kepler's third law could be generalized. As a matter
of fact, considering the specific energy $\xi$ associated with a given
\textsc{pro} in an isochrone potential $\psi\in\left\{\psi_\mathrm{ke},\psi_\mathrm{he},\psi_\mathrm{bo}\right\}$, 
we see that according to proposition~\ref{prop:Tniso}, %table \eqref{table_tau_n}
except for the harmonic potential, all isochrone orbits are such
that
\begin{equation}
\tau_{r}^{2}\left \vert \xi \right \vert ^{3}=\frac{\pi^{2}\mu^{2}}{2}
=\text{cst}. \label{eq:genkeplaw}
\end{equation}
Nevertheless, the law~\eqref{eq:genkeplaw} expressed in terms of the
specific energy is not stable under transvections of the potential, $\psi\mapsto\psi^\ast=\psi+\epsilon$,   and has to be
slightly modified for physical potentials when adding a constant. 
As mentioned
in proposition~\ref{prop:TnaffineTransf}, a \textsc{pro} $\left(\xi,\Lambda\right)$ in $\psi^\ast$ will
satisfy
\begin{equation}
\tau_{r}^{2}\left \vert \xi - \epsilon \right \vert ^{3}=\frac{\pi^{2}\mu^{2}}{2}
=\text{cst}. \label{eq:genkeplawcst}
\end{equation}
In these relations, $\xi$ is the specific energy of the test particle moving on a
\textsc{pro} with period $\tau_{r}$. The parameter $\mu$ is directly related
to the total mass of the system which sources the potential when it is finite
i.e. $\psi_{\mathrm{ke}}$ and $\psi_{\mathrm{he}}$. For the other non classical
isochrone $\psi_{\text{bo}}$, the total mass is infinite but equation
\eqref{eq:genkeplaw} always holds with a less physically comprehensive $\mu$
constant. The modification of the law~\eqref{eq:genkeplaw} into~\eqref{eq:genkeplawcst} somehow
hides the symmetry of the considered system. We
thus propose a geometric formulation of Kepler's Third Law for isochrones.
\\

The formulation of Kepler $\tau^{2}_r\times a^{-3}$ in terms of the geometric
parameter $a$ is more appropriate for conveying the symmetry of the potential. In fact, 
the Lagrangian $L = T - U$, with $T$ the specific kinetic energy of a particle and $U=\psi_\mathrm{ke}$,
is invariant, under a time $t\to\tilde t = \zeta t$ and space $\vec r \to \tilde{\vec r} = \varpi \vec r$
rescaling, if $$\zeta^2 \propto \varpi^3$$
because $ \psi_\mathrm{ke}$ is a homogeneous function of degree $-1$, i.e.
$ \psi_\mathrm{ke}\left(\tilde{\vec r}\right) = \varpi^{-1}  \psi_\mathrm{ke}\left(\vec r\right) $.
In order to geometrically express Kepler's Third law, we introduce in definition~\ref{def:isosma} ``semi major axes", relevant to all isochrone potentials, and directly
related to their Keplerian relative description. These characteristic lengths, generally related to specific
energies by~\eqref{eq:RAke}, \eqref{eq:RAha} and~\eqref{eq:newvarsRA}, provide a method to determine the mass of an isochrone system as mentioned at the end of this
section.

\begin{definition}
\label{def:isosma}
Let $r_p$ and $r_a$ be the peri- and apoastron radial distance of a given isochrone periodic orbit. We call the \emph{isochrone semi-major axis} of this orbit by the following lengths:

\begin{enumerate}
\item in a Kepler potential,
\begin{equation*}
a=\frac{1}{2}\left( r_{a}+r_{p}\right) ,
\end{equation*}

\item in a homogeneous box of radius $R$,
\begin{equation*}
a=\left( \frac{1}{2}\right) ^{2/3}R ,
\end{equation*}

\item in a H\'{e}non potential,
\begin{equation*}
a=\frac{1}{2}\left( \sqrt{b^{2}+r_{a}^{2}}+\sqrt{b^{2}+r_{p}^{2}}
\right) ,
\end{equation*}

\item in a bounded potential,
\begin{equation*}
a=\frac{1}{2}\left( \sqrt{b^{2}-r_{a}^{2}}+\sqrt{b^{2}-r_{p}^{2}}
\right) .
\end{equation*}
\end{enumerate}
\end{definition}

In definition~\ref{def:isosma}, we have considered a homogeneous box to include
the description of its elliptic trajectories with the Third Law. In fact, 
the situation of the harmonic potential needs more attention since 
$\psi_{\mathrm{ha}}$ is degenerate.
%: all particles share the same period whatever their energy is. 
In
such a potential all test particles share the same period but different specific
energies, hence relation \eqref{eq:genkeplaw} cannot hold for each specific
energy.  

The harmonic potential is not exactly representative of a real system because of its
constant density and infinite spatial extension, which imply an infinite mass. 
Instead, the potential associated with a finite homogenous repartition of masses in a ball
of radius $R$ with constant density (while the outside region is empty) does represent a real system and can be written as
%If it derives from a finite set of massive particles, let consider the particle of
%highest energy. Its orbit is bounded and we may define $R$ as its radial distance at its apoastron.
%The finite harmonic potential is then confined in a box of radius $R$ and constant mass density.
%It writes
\begin{equation*}
\psi_\mathrm{ha}^R \left( r\right) =\left\{ 
\begin{array}{ll}
\frac{1}{2}\omega ^{2}r^{2}-\frac{3}{2}\omega ^{2}R^{2} & \text{if }r<R \\ 
-\frac{GM}{r} & \text{if }r>R.%
\end{array}%
\right. 
\end{equation*}
We call it a finite harmonic potential. 
Additionally, either Gauss' theorem or the continuity of the force at the boundary of the ball leads to the following relation:
\begin{equation}
\mu = GM = \omega^2 R^3 .
\label{eq:finitehacty}
\end{equation}
As mentioned on page~\pageref{remark:isochronelink}, the harmonic potential
corresponds to the limit of an isochrone potential $\psi_\mathrm{he}$ or $\psi_\mathrm{bo}$ when
$b\to\infty$. This result holds for the finite harmonic potential $\psi_\mathrm{ha}^R$. 
In figure~\ref{fig:bo_and_he}, we see the confluence of these potentials and their densities when the parameter $b$ is large, 
as written in proposition~\ref{prop:finiteharmonic}.
%Proposition~\ref{prop:finiteharmonic} derives the characteristic length in this context. 
%It naturally appears when considering the relation between isochrone potentials. 
As it will be proven in theorem~\ref{thm:KepLawTa}, the characteristic length for the finite harmonic also naturally appears in the expression of 
Kepler's Third Law.
\begin{proposition}
\label{prop:finiteharmonic}
The finite harmonic potential satisfies
$$
%\lim \limits_{b\rightarrow\infty} \psi_\mathrm{he}   =  \lim \limits_{b\rightarrow\infty}  \psi_\mathrm{bo} \sim  \psi_\mathrm{ha}^R
%\ \text{ with } \ R= 2^{2/3} b.
 \psi_\mathrm{he} \left( r \right)  \underset{b\rightarrow\infty}{\sim} \psi_\mathrm{ha}^R \left( r \right)
 \ \text{ and } \
  \psi_\mathrm{bo} \left( r \right) \underset{b\rightarrow\infty}{\sim}   \psi_\mathrm{ha}^R \left( r \right)
\ \text{ with } \ R= 2^{2/3} b \ \text{ for any fixed } r.
$$
\end{proposition}
\begin{proof}
As already mentioned, the potential $\psi_\mathrm{ha}^R$ is continuous in $r=R$ if and only if $\mu=GM=\omega^2 R^3$.  

We assume the potentials vanish at $r=0$ without loss of generality. We consider then the reduced potentials and their 
equivalents from~\eqref{hebgrand} and~\eqref{bobgrand} as
$\psi_{\mathrm{he}}^{\mathrm{red}}\underset{b\rightarrow\infty}{\sim} \frac{\mu}{8b^3} r^2$
and
$\psi_{\mathrm{bo}}^{\mathrm{red}}\underset{b\rightarrow\infty}{\sim} \frac{\mu}{8b^3} r^2$.
%$$\psi_{\mathrm{he}}^{\mathrm{red}}=J_{\epsilon_\mathrm{he},0}(\psi_\mathrm{he})\underset{b\rightarrow\infty}{\sim} \frac{\mu}{8b^3} r^2$$ 
%and 
%$$
% \psi_{\mathrm{bo}}^{\mathrm{red}}=J_{\epsilon_\mathrm{bo},0}(\psi_\mathrm{bo})  (r) \underset{b\rightarrow\infty}{\sim} \frac{\mu}{8b^3} r^2.
%$$

In this limit, the H{\'e}non and bounded potentials behave as homogeneous spheres inside a radius 
$R=2^{2/3} b$. \qed
\end{proof}

Now, Kepler's third law can be generalized to all isochrone potentials in 
theorem~\ref{thm:KepLawTa}.

\begin{theorem}
\label{thm:KepLawTa}
For any radially periodic orbit in an isochrone potential, the 
square of the radial period is proportional to the cube of the isochrone semi-major axis by
\begin{equation}
\label{eq:KeplerThirdLawSMA}
\tau_r^2 = \frac{4\pi^2}{\mu} a^3, 
\end{equation}
where $\mu$ is the mass parameter of $\psi_\mathrm{ke}$, $\psi_\mathrm{he}$, $\psi_\mathrm{bo}$ 
and $\mu=\omega^2 R^3$ for $\psi_\mathrm{ha}^R$.
\end{theorem}

\begin{proof}
In $\psi_\mathrm{ke}$, it is Kepler's third law. In $\psi_\mathrm{he}$, for a \textsc{pro} of energy $\xi<0$,
the radial variable $s$ introduced in the proof of proposition~\ref{prop:Tniso} satisfies~\eqref{eq:newvarsRA} as
\begin{equation*}
s_{a}+s_{p}=2-\frac{\mu }{\xi b}=2+\sqrt{\left( \frac{r_{a}}{b}\right) ^{2}+1%
}+\sqrt{\left( \frac{r_{p}}{b}\right) ^{2}+1}
\end{equation*}%
and
\begin{equation*}
\xi =-\frac{\mu }{2a}\ \text{\ with }a=\frac{\sqrt{r_{a}^{2}+b^{2}}+\sqrt{%
r_{p}^{2}+b^{2}}}{2}.
\end{equation*}
Inserting this expression of $\xi$ in~\eqref{eq:genkeplaw} gives~\eqref{eq:KeplerThirdLawSMA}.

Similarly, in $\psi_\mathrm{bo}$ the variable $s$ satisfies 
\begin{equation*}
s_{\min}+s_{\max}=2+\frac{\mu}{\xi b}=2+\sqrt{1-\left( \frac{r_{p}}{b}%
\right) ^{2}}+\sqrt{1-\left( \frac{r_{a}}{b}\right) ^{2}} 
\end{equation*}
and
\begin{equation*}
\xi=\frac{\mu}{2a} > 0 \ \text{ with } \ a = \frac{1}{2} = \sqrt{b^{2}-r_{p}^{2}}+\sqrt{b^{2}-r_{a}^{2}}.
\end{equation*}
By inserting this expression in~\eqref{eq:genkeplaw}, we recover the law~\eqref{eq:KeplerThirdLawSMA}.

In $\psi_{\mathrm{ha}}$, all orbits have the same radial period $\tau
_{r}=\frac{\pi}{\omega}$. When a harmonic system is compacted into a ball
of radius $R$ of constant density, then $\mu=\omega^2 R^3$ according to~\eqref{eq:finitehacty}.
%$\rho_{0}$, its potential is $\psi
%_{\mathrm{ha}}^{R}$, and its total mass is finite and trivially given by
%$M=\frac{4}{3}\pi R^{3} \rho_{0}$. It is also given by Gauss' Theorem
%\[
%\frac{GM}{R^{2}}=\left.  \frac{d\psi_{\mathrm{ha}}^{R}}{dr}\right \vert
%_{r=R}=\omega^{2}R~~\implies \mu=GM=\omega^{2}R^{3}.
%\]
Hence, the period could be related to the radius of the ball through the
relation $\tau_{r}=\frac{\pi}{\sqrt{\mu}}R^{3/2}$. Introducing the length
$a=\left(  \frac{1}{2}\right)  ^{2/3}R$, one has
\[
\mu \tau_{r}^{2}=4\pi^{2}a^{3}.  \ \ \ \ \ \ \  \qed
\]
\end{proof}

Thus, Kepler's third law appears to be generalized through the isochrone group.
Kepler's third law is mainly used for mass determination, as in, for example, the post-newtonian approximation
to estimate the mass of black holes.
For a Kepler potential, only one orbit is theoretically necessary to determine the mass of the
central attractive body given by $\mu$. For other isochrone potentials, using~\eqref{thm:KepLawTa}, only two
orbits would be necessary to determine the parameter $b$ and mass $\mu$ described by their isochrone potential.

\subsection{The Bertrand theorem\label{BTh}}

In 1873, J. Bertrand published a fascinating theorem: \emph{There are only
two central potentials for which all orbits with  
an initial velocity below a certain limit are closed, namely
the Keplerian and the harmonic potentials}.
While this fascinating result was
proved more than 140 years ago, the proof of this theorem has been retaining
the attention. According to the most recent reviews~\cite{chin} and works on this topic~\cite[chap.3]{albouy2002lectures}, it
has been proven using very different techniques: \cite{bertrand,Arnold,Lagrange,Jovanovic2015,5}, using global methods, sometimes stemming from the analysis of the precession
rate as initiated in proposition XLV of~\cite{Newton};  \cite{6,Goldstein,8,9}, developing perturbative expansions; \cite{10,11,12}, using inverse transformations methods; \cite{13}, by searching for
additional constants of motion; and \cite{Fejoz}, mainly using Birkhoff invariants along
circular orbits in a generic potential. Furthermore, the original proof does not mention the case
of collision orbits. We will therefore consider the result of Bertrand's theorem under the hypotheses
of orbits that are bounded in position and bounded away from 0. We propose here to show that, in fact,
Bertrand's theorem is a refined property of the isochrone one.

\begin{theorem}
\label{theob}In a given radial potential $\psi$, if all non-circular orbits that are bounded in position and bounded away from 0 are
closed, then $\psi$ is isochrone.
\end{theorem}

\begin{proof}
In a given radial potential $\psi$, if all bounded and bounded away from 0 orbits are closed, the
increment of the azimuthal angle $\Delta \varphi$ during the transfer from
$r_{a}$ to $r_{p}$ is a fractional multiple of $2\pi$, i.e. the quantity
$n_{\varphi}=\frac{\Delta \varphi}{2\pi}\in \mathbb{Q}$. But, for a given radial
potential $\psi \left(  r\right)  $, we have that
\[
n_{\varphi}=-\frac{\partial \mathcal{A}_{r}}{\partial \Lambda}=\frac{1}{\pi}{\int_{r_{p}%
}^{r_{a}}}\frac{\Lambda}{r^{2}\sqrt{2\left[  \xi-\psi \left(  r\right)
\right]  -\dfrac{\Lambda^{2}}{r^{2}}}}dr
\]
is a continuous mapping $\left(  \xi,\Lambda \right)  \mapsto$ $n_{\varphi}$.
By continuity, because the set $\mathbb{R}\setminus \mathbb{Q}$ is dense in
$\mathbb{R}$, one can conclude that in order to only have closed orbits,
$n_{\varphi}=cst\in \mathbb{Q}$. In these conditions we then have
\[
0=\frac{\partial n_{\varphi}}{\partial \xi}.
%=-\frac{\partial^{2}\mathcal{A}_{r}}{\partial \xi \partial \Lambda}
\]
%using Schwarz theorem %%%%% ATTENTION     (mq $\mathcal{A}_{r}$ est $\mathcal{C}^{2}$) 
%we then have
%\[
%0=-2\pi \frac{\partial \tau_{r}}{\partial \Lambda}%
%\]
%which is the general characterization of isochrone potential stated in theorem~\ref{thm:characisopot}
%of appendix~\ref{appendix:isocharac}. 
This characterizes an isochrone potential according to theorem~\ref{thm:characisopot}
of appendix~\ref{appendix:isocharac}.
The potentials of the form $-\frac{\mu}{r^\alpha}$ with $\alpha>2$ are excluded because all orbits that
are bounded in position either collide at the origin or are circular. 
\qed
\end{proof}

Using our study we can go further because we have obtained a geometric and
algebraic description of the whole set of isochrone potentials. More
specifically, we have obtained in table~\eqref{table_tau_n} the explicit value
of $n_{\varphi}$ for all isochrone potentials. The completeness of our
description and this table enable us to claim that Bertrand's theorem is 
a corollary of theorem~\ref{theob}.

\begin{corollary}
The Bertrand Theorem !
There are only
two central potentials for which all non-circular orbits that are bounded in position and bounded away from 0 
are closed, namely
the Keplerian and the harmonic potentials.
\end{corollary}

\begin{proof}
As the quantities $\frac{\Lambda}{2\sqrt{4b\mu+\Lambda^{2}}}$ and $\frac{\Lambda}{\sqrt{\lambda+\Lambda^2}}$ in proposition~\ref{prop:Tniso}
and~\ref{prop:TnaffineTransf} cannot be
rational for each value of $\Lambda$, among all isochrone potentials, only
$\psi_{\mathrm{ke}}$ and $\psi_{\mathrm{ha}}$ have rational $n_{\varphi}$ for all orbits, i.e.
for all values of $\left(  \xi,\Lambda \right)  $. 
%All orbits that collide spiraling into the origin are not closed and are not treated in the original proof. 
\qed

\end{proof}

In a given potential, the fact that all bounded orbits are closed, namely
Bertrand's property, is then a supplementary restriction to the isochrone one.

\section{Conclusion}

In this paper we have revisited the set of isochrone orbits in radial 3D
potentials. These models concern self-organised radial systems with long-range
interactions like gravitation or electrostatics with one kind of electric
charge. Let us summarize the main results we have obtained:

\begin{enumerate}
\item We have clarified the original proof by Michel H\'{e}non~\cite{Henon58}
that isochrone potentials are contained in a branch of a parabola in adapted
coordinates (theorem~\ref{newhenontheo}). These parabolas characterize the property of isochrony.

\item Taking into account very general properties of potentials in physics ---
i.e. invariance under the addition of a constant, conservation of the energy
and angular momentum for isolated radial systems --- we have given a
geometrical characterization and classification of the set of all isochrone orbits/potentials
that we have completed. This
characterization (theorem \ref{theo1}) is based on a subgroup $\mathbb{A}$ of the real affine group.% \textsc{ga}$(\mathbb{R})$.

\item We have shown (theorem~\ref{theo1} and sec.~\ref{subsec:affinegroupaction}) %~\ref{theo2}) 
that under the group action of $\mathbb{A}$, 
any isochrone potential is in the orbit of one of the four fundamental potentials:
Kepler, H\'{e}non, Bounded or Harmonic (definition~\ref{def:redphysgau}).

\item Focusing on orbits, we have proposed a mapping which generalizes the Bohlin transformation to all isochrone potentials.
This mapping, summarized in theorem \ref{prop:Bohlin},  connects any Keplerian elliptic orbit to a particular isochrone radially periodic orbit. Reciprocally, by theorem \ref{thm:orbitboost}, we have shown how to construct the elliptic Keplerian orbit connected to any isochrone periodic orbit. This mapping is based on a particular linear transformation, that we call a bolst, which preserves the orbital differential equation for a given value of the angular momentum.

\item With the set of symmetric bolsts, namely Ibolsts, we have revealed the relative behavior of the isochrone property of orbits/potentials. We have detailed in sec.~\ref{subsec:isochronerelativity} a lot of similarities between the special theory of relativity and the isochrony of orbits in radial potentials. In this view, a given orbit in an isochrone potential is seen as a Keplerian orbit in its special frame. This is the Isochrone Relativity presented in sec. \ref{subsec:isochronerelativity}. The time and energy are relative to each orbit which defines a frame of reference. 
%Hence, we have shown that a given orbit in an isochrone potential is a Keplerian one in an appropriate frame. 
Various examples were presented and illustrated to construct isochrone orbits in isochrone potentials.

\item The explicit expression of the radial ($\tau_r$) and azimuthal ($n_{\varphi}$) periods was calculated for all fundamental isochrone potentials. These results are presented in proposition~\ref{prop:Tniso}. The computation of these periods in physical or gauged isochrones is possible using the results presented in proposition \ref{prop:TnaffineTransf}.

\item We have proposed a generalization of the quadricentennial Kepler's Third Law in theorem~\ref{thm:KepLawTa}. While this classic law involves the semi major axis of \emph{closed} Keplerian orbits, we define characteristic lengths in each isochrone potential that are related to the radial period in the famous 3/2 power equation. This rational value 3/2 is well known to be related to the mechanical similarity involved in the Kepler potential and its $-1$ homogeneity property (e.g. \cite{Landaumeca} p. 22-24). In this view, the generalization of the Kepler's Third Law to any isochrone is not surprising since we have seen that any isochrone is a Kepler in the adequate referential.

\item Noting that both the radial period $\tau_{r}$ and the precession rate
$n_{\varphi}$ are partial derivatives of the same quantity, i.e. the radial
action $\mathcal{A}_{r}$, we observed  
that the famous Bertrand's theorem is a specific
property of isochrones. Once again this property could be interpreted as a consequence of the isochrone relativity.
\end{enumerate}

The essence of isochrony is Keplerian.
As isochrony is characterized by the parabolic property in H\'{e}non's variables, we understand 
the linear transformations that act on these parabolas and are shaped by the bolst $B_{\alpha,\beta}$ play crucial roles. Merging these ideas, we conjecture that a theory of general relativity of radial potentials could be formulated using non-linear transformations. This theory could relate any orbit in a radial potential to an associated orbit in a Kepler potential. 

In a forthcoming paper we will explain the physical importance of the isochrone potential during the formation and evolution process of
self-gravitating systems.

\section*{Appendix}

\appendix

\section{Isochrone characterization}
\label{appendix:isocharac}
Let us recall that the radial action $\mathcal A_r$ gives the radial period $\tau_r$ and the increment of the azimuthal angle $n_\varphi$
through~\eqref{djde} and~\eqref{djdl} in sec.~\ref{subsecHenonparabola}:
$$
   \frac{\partial \mathcal A_r}{\partial \xi} = \frac{\tau_r}{2\pi} 
   \ \text{ and } \
   - \frac{\partial \mathcal A_r}{\partial \Lambda} = \frac{\Delta\varphi}{2\pi} = n_\varphi.
$$
The exclusive $\xi-$dependency of $\tau_r$ is the fundamental isochrone property used by Michel H{\'e}non to define isochrone potentials. 
After his analysis, he remarked the exclusive $\Lambda-$dependency of $n_\varphi$ for his potential. The following theorem
establishes the equivalence of properties which can characterize isochrone potentials as a whole.
\begin{theorem}
\label{thm:characisopot}
Consider a central potential $\psi$.
Then the following properties are equivalent:
\begin{enumerate}
   \item \label{item:Txi} For any orbit $(\xi,\Lambda)$ in $\psi$, $\tau_r$ only depends on $\xi$.
   \item \label {item:nL} For any orbit $(\xi,\Lambda)$ in $\psi$, $n_\varphi$ only depends on $\Lambda$.
   \item \label{item:action} There exist two function $f$ and $g$ such that for any $(\xi,\Lambda)$ the radial action is $\mathcal A_r(\xi,\Lambda)=f(\xi)+g(\Lambda)$.
\end{enumerate}
\end{theorem}

\begin{proof}
   The separation of variables in the radial action expressed in~\ref{item:action} implies the two properties~\ref{item:Txi} and~\ref{item:nL} by direct differentiation with respect to $\xi$ for~\ref{item:Txi} 
   and $\Lambda$ for~\ref{item:nL}. 
   \\
   Assume~\ref{item:nL} is true for any orbit in the central potential $\psi$. Then $\frac{\partial \mathcal A_r}{\partial \xi} = \frac{\tau_r(\xi)}{2\pi} $ and
   by integration there exists a function $g$, constant with respect to $\xi$, such that $\mathcal A_r(\xi, \Lambda) = f(\xi) + g(\Lambda)$, where $f$ is a primitive of $\frac{\tau_r}{2\pi}$. We thus recover~\ref{item:action}.
   \\
   In the same way, assuming~\ref{item:nL} implies~\ref{item:action}. \qed
\end{proof}

\section{Proof of a parabola property\label{appendixa}}

Michel H{\'{e}}non has shown in~\cite{Henon58} the equivalence between the
isochrony of a potential $\psi$ and the parabolic property of the graph
$\mathcal{C}$ of $f:x\rightarrow x\psi$ associated with it. We propose here a
different proof based on the analyticity of the potential.

We call $\left(  \mathcal{P}\right)  $ this parabolic property, and it can
be formulated as follows.

A function $f:I\rightarrow \mathbb{R}$ has the property $\left(  \mathcal{P}%
\right)  $ if and only if :

\begin{enumerate}
\item $f$ is either convex or concave on the real interval $I$, i.e.
$f^{\prime \prime}>0$ or $f^{\prime \prime}<0$ on $I.$

\item For any $P_{0}$ belonging to its graph $\mathcal{C}$, and for any line
$\mathcal{L}$ parallel to the tangent $\mathcal{T}_{P_{0}}\left(  \mathcal{C}\right)
$, the square length of the projected chord $\left \vert x_{a,1}-x_{p,1}\right \vert $ is
proportional to the distance between the chord and the tangent to the curve that is parallel to the chord. %parallel tangent to the curve. 
The proportional relation holds equivalently with the vertical distance $P_{0}I$ 
between $\mathcal{T}_{P_{0}}\left(  \mathcal{C}\right)$ and $\mathcal{L}$.
In figure \ref{parabola} we have $\mathcal{T}_{P_{0}}\left(
\mathcal{C}\right)  :y=\xi x-\Lambda_{0}^{2}$ and $\mathcal{L}:y=\xi
x-\Lambda_{1}^{2}$.
\end{enumerate}
In terms of function, this last point translates as follows:%
\[%
\begin{array}
[c]{cc}%
\left(  \mathcal{P}\right)  \ : & \left \vert
\begin{array}
[c]{l}%
\forall x_{0}\in I,\exists \varpi(x_{0})\in \mathbb{R}_{+} \text{ such that
}\forall \lambda>0\text{, when they exist,}\\
\text{the two solutions }x_{p}\text{ and }x_{a}\text{ of the equation }\\
f(x)-f(x_{0})=\lambda+f^{\prime}(x_{0})\left(  x-x_{0}\right)  \\
\text{satisfy the relation }\left(  x_{a}-x_{p}\right)  =\varpi(x_{0}%
)\sqrt{\lambda}\text{ with }x_{a}%
>x_{p} .%
\end{array}
\right.
\end{array}
\]

Michel H\'{e}non's equivalence then corresponds to the following theorem.

\begin{theorem}
\label{newhenontheo}
Let $f:I\to \mathbb{R}$ be an analytic real function on an interval $I\subset \mathbb{R}$.
Then  the graph of $f$ is a parabola if and only if $f$ has the property $(\mathcal{P})$.
\end{theorem}

The proof of this result will be done in several steps. The first one is a reduction procedure given by the following lemma.

\begin{lemma}
\label{newwwlemm1}Let $g:I\rightarrow \mathbb{R}$ be a real analytic function
satisfying property $\left(  \mathcal{P}\right)  $. Then we have
\begin{enumerate}
\item For any real constant $a\neq0$, $f:=ag$ satisfies $\left(
\mathcal{P}\right)  ;$

\item For any constants $\left(  \varepsilon,\lambda \right)  \in \mathbb{R}%
^{2}$, $f\left(  x\right)  =g\left(  x\right)  +\varepsilon x+\lambda$
satisfies $\left(  \mathcal{P}\right)  ;$

\item For any constants $\left(  \varepsilon,\lambda \right)  \in \mathbb{R}%
^{2}$, with $\varepsilon \neq0$, $f\left(  x\right)  :=g\left(  \varepsilon
x+\lambda \right)  $ satisfies $\left(  \mathcal{P}\right)$.
\end{enumerate}
\end{lemma}

This statement indicates that property $\left(  \mathcal{P}\right)  $ is
stable by affine transformations acting on the graph of the considered
function. Its proof is quite obvious and is left to the reader.

Any graph of a parabola can be obtained by the transformations of lemma~\ref{newwwlemm1} of 
the graphs of $x\mapsto \sqrt{x}$ or $x\mapsto x^2$. %, which have this property $\left(  \mathcal{P}\right)$.
It follows that, if the
graph of $f$ is a parabola, then $f$ satisfies the simple implication of the theorem.

In order to have the converse implication, i.e. $\left(  \mathcal{P}\right)
\implies \mathcal{C}$ is a parabola, we now consider the simple case where, in
figure \ref{parabola}, $\mathcal{T}_{P_{0}}\left(  \mathcal{C}\right)  $ is horizontal.

\begin{lemma}
\label{newlemmm2}If $\varphi:I\rightarrow \mathbb{R}$ is a real analytic
function and if at $x_{0}\in I$ we have $\varphi^{\prime}\left(  x_{0}\right)
=0$ and $\varphi^{\prime \prime}\left(  x_{0}\right)  =2$, then%
\begin{equation}
5\left[  \varphi^{\left(  3\right)  }\left(  x_{0}\right)  \right]
^{2}=6\varphi^{\left(  4\right)  }\left(  x_{0}\right)  \text{.}%
\label{neweqB1}
\end{equation}
\end{lemma}

\begin{proof}
Let $\varphi \left(  x\right)  =g\left(  z\right)  $ with\ $z=x-x_{0}$.
Then, since $\varphi$ is analytic, $g$ has a convergent Taylor expansion at $x_{0}$ 
of the form $g \left(  z\right)  =z^{2}+g_{3}z^{3}+g_{4}z^{4}+\cdots$, such that
\[
g(z)=z^{2}\left(  1+\sum_{n\geq3}g_{n}z^{n-2}\right)  =z^{2}(1+R(z)),
\]
where $R(z)$ is a convergent series that vanishes at $z=0$. Then, we may expand
\[
\sqrt{1+R(z)}=1+\frac{1}{2}R+\frac{1}{2!}\left(  \frac{1}{2}\right)  \left(
\frac{1}{2}-1\right)  R^{2}+\dots
\]
and insert it in
\[
\sqrt{g(z)}=G(z)=z\sqrt{1+R(z)}=z+G_{2}z^{2}+G_{3}z^{3}+\dots
\]
Because $G(0)=0$ and $G^{\prime}(0)=1$, $G$ is locally bijective in the
neighborhood of $z=0$; the analytic inverse function theorem assures that its
inverse $H$ is also a convergent power series $H(z)=z+\sum_{n\geq2}h_{n}z^{n}$. 

The fact that $\varphi$ satisfies ($\mathcal{P}$) means that for any small
enough $\lambda>0$ the two solutions $z_{1}$ and $z_{2}>z_{1}$ of $g\left(
z\right)  =\lambda$ satisfy $z_{2}-z_{1}=\varpi \left(  x_{0}\right)
\sqrt{\lambda}$. However,
\[%
\begin{array}
[c]{cc}%
g\left(  z\right)  =\lambda & \Leftrightarrow G^{2}\left(  z\right)
=\lambda \\
& \Leftrightarrow G\left(  z\right)  =\pm \lambda \\
& \Leftrightarrow z=H\left(  \pm \lambda \right).
\end{array}
\]
More precisely, $z_{2}=H(\sqrt{\lambda})$ and $z_{1}=H(-\sqrt{\lambda
})$ if $\lambda \geq0$ is small enough because $H$ locally increases. The
second condition from $\left(  \mathcal{P}\right)  $ gives $H(\sqrt{\lambda
})-H(-\sqrt{\lambda})=\varpi(x_{0})\sqrt{\lambda}$ for sufficiently small
$\lambda \geq0$ and
\begin{equation}
H(t)-H(-t)=\varpi t,
\label{eqintermmmmm}
\end{equation}
since all members of the previous equation are power series. Inserting the
expression of $H(t)=\sum_{n\geq1}h_{n}t^{n}$ in~\eqref{eqintermmmmm}, noting that the even terms
disappear, one finds $2h_{1}=2=\varpi$ and $h_{2m+1}=0$ if $m\geq1$. In
other words,
\[
H(t)=t+h_{2}t^{2}+h_{4}t^{4}+\dots
\]
Identifying the terms of the equality given by $H\circ G(z)=z$, one specifically
finds $G_{2}=-h_{2}$ and $G_{3}=-2h_{2}G_{2}=2h_{2}^{2}$. Hence the expansion
of $g$ is written as
\[%
\begin{array}
[c]{ccccl}%
g(z) & = & G^{2}(z) & = & z^{2}-2h_{2}z^{3}+5h_{2}^{2}z^{4}+\dots \\
&  &  & = & \frac{1}{2}g^{\prime \prime}(x_{0})z^{2}+\frac{1}{6}g^{(3)}(x_{0})z^{3}+\frac{1}{24}g^{(4)}(x_{0}%
)z^{4}+\dots ,
\end{array}
\]
where the identification between each term leads to $\left(  g^{(3)}(x_{0})\right)  ^{2}=12^{2}h_{2}^{2}=\frac{6}{5}g^{(4)}(x_{0})$ which is exactly~\eqref{neweqB1}.\qed
\end{proof}

We now exploit this particular case to characterize the property
($\mathcal{P}$) in terms of a differential equation.

\begin{lemma}
Let $f:I\rightarrow \mathbb{R}$ be a real analytic function satisfying
($\mathcal{P}$). Then $f$ also satisfies %
\begin{equation}
\forall x_{0}\in I\text{, }5\left[  f^{\left(  3\right)  }\left(
x_{0}\right)  \right]  ^{2}=3f^{\left(  4\right)  }\left(  x_{0}\right)
f^{\prime \prime}\left(  x_{0}\right).
\label{newB2}
\end{equation}
\end{lemma}

\begin{proof}
For any point $x_{0}\in I$ with $f^{\prime \prime}\left(  x_{0}\right)  \neq0$,
the function
\[
\varphi \left(  x\right)  =\frac{2}{f^{\prime \prime}\left(  x_{0}\right)
}\left[  f\left(  x\right)  -f\left(  x_{0}\right)  -f^{\prime}\left(
x_{0}\right)  \left(  x-x_{0}\right)  \right]
\]
satisfies the property ($\mathcal{P}$) according to lemma \ref{newwwlemm1}.
Moreover we have that $\varphi^{\prime}\left(  x_{0}\right)  =0$ and%
\[
\left \{
\begin{array}
[c]{l}%
\varphi^{\prime \prime}\left(  x\right)  =\frac{2f^{\prime \prime}\left(
x\right)  }{f^{\prime \prime}\left(  x_{0}\right)  }\implies \varphi
^{\prime \prime}\left(  x_{0}\right)  =2\\
\varphi^{(3)}\left(  x\right)  =\frac{2f^{(3)}\left(  x\right)  }%
{f^{\prime \prime}\left(  x_{0}\right)  }\implies \varphi^{(3)}\left(
x_{0}\right)  =\frac{2f^{(3)}\left(  x_{0}\right)  }{f^{\prime \prime}\left(
x_{0}\right)  }\\
\varphi^{(4)}\left(  x\right)  =\frac{2f^{(4)}\left(  x\right)  }%
{f^{\prime \prime}\left(  x_{0}\right)  }\implies \varphi^{(4)}\left(
x_{0}\right)  =\frac{2f^{(4)}\left(  x_{0}\right)  }{f^{\prime \prime}\left(
x_{0}\right)  }.%
\end{array}
\right.
\]
As a consequence, $\varphi$ satisfies the assumptions of lemma
\ref{newlemmm2} and therefore (\ref{neweqB1})$\implies$(\ref{newB2}).\qed
\end{proof}

Let us observe that (\ref{newB2}) was obtained under the condition
that $f^{\prime \prime}\left(  x_{0}\right)  \neq0$. By analytic continuation
the relation is still satisfied at the isolated points where $f^{\prime \prime}$ could vanish.

We are therefore led to solve (\ref{newB2}), which is in fact the
\emph{universal differential equation for parabolas}. Setting $w=f^{\prime \prime}$,  (\ref{newB2}) becomes
\begin{equation}
5\left(  w^{\prime}\right)  ^{2}=3w^{\prime \prime}w.\label{newnewnewB3}%
\end{equation}
Two cases may occur:

\begin{enumerate}
\item If $w^{\prime}=f^{\left(  3\right)  }:=0$ on $I$:

then $f^{\prime \prime}$ is constant and $f$ is a second-degree polynomial and its
graph $\mathcal{C}$ is a parabola.

\item If $w^{\prime}=f^{\left(  3\right)  }$ do not vanish everywhere on $I$:

then on any subset where $w^{\prime}\neq0$, equation~\eqref{newnewnewB3} becomes%
\[
\frac{5w^{\prime}}{w}=\frac{3w^{\prime \prime}}{w^{\prime}},%
\]
which gives by integration%
\[
5\ln \left \vert w\right \vert =3\ln \left \vert w^{\prime}\right \vert +cst\implies
w^{\prime}w^{-5/3}=cst.
\]
Hence $w^{-2/3}$ is a linear function of $x$, namely $w\left(  x\right)
=f^{\prime\prime}\left(  x\right)  =\left(  \varepsilon x+\lambda \right)  ^{-3/2}$. By
integrating this equation twice, we get that $f$ is proportional to
$f_{0}\left(  x\right)  =\sqrt{\varepsilon x+\lambda}+ax+b$ whose graph is a
parabola too.
\end{enumerate}
This concludes the proof of theorem \ref{newhenontheo}.

\section{Useful Lemmas \label{appendixlemma}}

Consider
\begin{itemize}
\item a frame $\mathcal{R}_O=\left(  O,\vec{i},\vec{j}\right)  $
with coordinates $\left(  x,y\right)$ for each point $M$;

\item a linear application $L : \mathbb{R}^{2}\rightarrow
\mathbb{R}^{2}$ such that $L\left(  \vec{i}\right)  =\vec{u}$ and $L\left(
\vec{j}\right)  =\vec{v}$;

\item a curve $\mathcal{C}$ of equation $f\left(  x,y\right)  =0$ in
the frame $\mathcal{R}$.
\end{itemize}
The linearity of $L$ ensures the two properties below.

\begin{lemma}
\label{lemme_lineaire} The cartesian equation of curve
$\mathcal{C}^{\prime}=L\left(  C\right)  $ in the frame
$\mathcal{R}_O^{\prime}=\left(  O,\vec{u},\vec{v}\right)  $ remains  $f\left(
x,y\right)  =0 $.
\end{lemma}

\begin{proof}
   Consider $\overrightarrow{OM} = x \vec{i} + y \vec{j}$. Then $M \in \mathcal C \Leftrightarrow f(x,y)=0$.
   But $L(\overrightarrow{OM}) = x L(\vec{i}) + y L(\vec{j}) = x \vec u + y \vec v \in \mathcal C^\prime$ by definition. 
   Thus $f\left(x,y\right)  =0 $ also defines $\mathcal C^\prime$ in $\mathcal{R}_O^{\prime}$.
\end{proof}

Define 
\begin{itemize}
\item $\mathcal{T}_{O}\left(  \mathcal{P}\right)  $ the tangent at the origin $O$ to a parabola $\mathcal{P}$;

\item $\mathcal{S}\left(  \mathcal{P}\right)  $ the symmetry axis of parabola $\mathcal{P}$.
\end{itemize}
Then we have the following lemma: 

\begin{lemma}
\label{lemme_des_droites}If $\mathcal{P}^{\prime}=L\left(  \mathcal{P}\right)
$ then $\mathcal{T}_{O}\left(  \mathcal{P}^{\prime}\right)  =L\left(
\mathcal{T}_{O}\left(  \mathcal{P}\right)  \right)  $ and $\mathcal{S}\left(
\mathcal{P}^{\prime}\right)  =L\left(  \mathcal{S}\left(  \mathcal{P}\right)
\right)  $.
\end{lemma}

\begin{proof}
   According to lemma~\ref{lemme_lineaire}, $\mathcal P$ and $\mathcal P^\prime$ have
   the equation $(ax+by) +e = (cx+dy)^2$ in their respective frames.
   Then the tangent $\mathcal{T}_{O}$ has the direction vector $\vec t = -b \vec i + a \vec j$ and
   the symmetry axis $\mathcal{S}\left(  \mathcal{P}\right)$ has the vector $\vec n = -d \vec i + c \vec j$.
   In the same way, with natural notations, $\vec t^\prime = -b \vec u + a \vec v$ and $\vec n^\prime = -d \vec u + c \vec v$.
   Thus $\vec t^\prime=L(\vec t)$ and $\vec n^\prime = L(\vec n)$.
\end{proof}

\section{Isochrone integrals\label{appendixb}}

\begin{lemma}
\label{lemma:integrals}
The Keplerian and harmonic radial actions are given by 
$$
\mathcal{A}_{r}^{%
\mathrm{ke}}=\frac{\mu }{\sqrt{-2\xi }}-\Lambda 
\ \text{ and } \
\mathcal{A}_{r}^{\mathrm{ha}}=\frac{\xi }{2\omega}-\frac{\Lambda}{2}. 
$$
For any pair of positive real $\left( u_{1},u_{2}\right)$ such that 
$u_{1}<u_{2}$, we have 
\begin{equation*}
\mathcal{I}_{1}\left( u_{1},u_{2}\right) =\int_{u_{1}}^{u_{2}}\frac{\sqrt{%
\left( u-u_{1}\right) \left( u_{2}-u\right) }}{u}\,du=\tfrac{\pi }{2}\left(
u_{1}+u_{2}-2\sqrt{u_{1}u_{2}}\right) 
\end{equation*}%
and 
\begin{equation*}
\begin{array}{ccl}
\mathcal{I}_{2}\left( u_{1},u_{2}\right) &=&\displaystyle \int_{u_{1}}^{u_{2}}\frac{\left(
u-1\right) \sqrt{(u-u_{1})\left( u_{2}-u\right) }}{u\left( u-2\right) }%
\,du\\
&=&\left\{ 
\begin{array}{cl}
\frac{\pi }{2}\left[ u_{1}+u_{2}-\sqrt{u_{1}u_{2}}-\sqrt{\left(
u_{1}-2\right) \left( u_{2}-2\right) }-2\right]  & \text{if\ \ }2<u_{1} \\ 
\frac{\pi }{2}\left[ u_{1}+u_{2}-\sqrt{u_{1}u_{2}}+\sqrt{\left(
u_{1}-2\right) \left( u_{2}-2\right) }-2\right]  & \text{if\ }\ u_{2}<2.%
\end{array}%
\right. 
\end{array}
\end{equation*}
\end{lemma}

The result of $\mathcal{I}_{1}$ can be obtained by a direct meticulous computation; instead, we propose to deduce it 
from the physical computation of the Keplerian radial action.

In a second step, we will deduce $\mathcal{I}_{2}$ from $\mathcal{I}_{1}$.

%The derivation of $\mathcal{I}_{1}$ is a classic calculus, which could be done by two ways: a direct
%one and another one using basics physics knowledge.
%We give the two
%options below and deduce the expression of $\mathcal{I}_{2}$.

%\subsection{Direct calculation of $\mathcal{I}_{1}$}
%\label{subsec:directcalculintegral}
%We first reorganize the integrand %under the radical 
%to get 
%\begin{equation*}
%\mathcal{I}_{1}=\int_{\alpha }^{\beta }\frac{\sqrt{A^{2}-(u-B)^{2}}}{u}%
%\,du\;\;\text{with}\;A=\frac{\alpha +\beta }{2}\;\text{and}\;B=\frac{\alpha
%-\beta }{2}.
%\end{equation*}%
%Then, introducing $s=\frac{u-B}{A}$ we get 
%\begin{equation*}
%\mathcal{I}_{1}=Bk^{2}\int_{-1}^{+1}\frac{\sqrt{1-s^{2}}}{1+ks}\,ds\;\;\text{%
%with}\;k=\frac{A}{B}=\frac{\beta -\alpha }{\beta +\alpha }<1.
%\end{equation*}%
%Hence with $s=\sin (t)$ we get 
%\begin{equation*}
%\mathcal{I}_{1}=Bk^{2}\int_{-\frac{\pi }{2}}^{+\frac{\pi }{2}}\frac{\cos
%^{2}(t)}{1+k\sin (t)}\,dt.
%\end{equation*}%
%This last integral does not resist to the last change of variable $\tan (%
%\frac{t}{2})=y$ already used page \pageref{Bioche} for the calculus of the
%integral in the generalization of the Bohlin transformation. It gives 
%\begin{equation*}
%\mathcal{I}_{1}=\pi B\left( 1-\sqrt{1-k^{2}}\right) 
%\end{equation*}%
%which is the result after a little reorganization and implementation of the
%values of $k$ and $B$.

\subsection{Computation of $\mathcal{A}_r^\mathrm{ke}$, $\mathcal{A}_r^\mathrm{ha}$ and physical deduction of $\mathcal{I}_{1}$}
\label{subsec:physdeducintegral}
The radial action for an orbit of negative energy $\xi $ and momentum $%
\Lambda $ in a Keplerian potential $\psi _{\mathrm{ke}}\left( r\right) =-\frac{\mu}{r}$ is given by
\begin{eqnarray}
\mathcal{A}_{r}^{\mathrm{ke}} &=&\frac{1}{\pi }{\int_{r_{p}}^{r_{a}}}\sqrt{2%
\left[ \xi -\psi _{\mathrm{ke}}\left( r\right) \right] -\dfrac{\Lambda ^{2}}{%
r^{2}}}dr  \label{AKEeq1} \\
&=&\frac{\sqrt{-2\xi }}{\pi }{\int_{r_{p}}^{r_{a}}}\dfrac{\sqrt{\left(
r-r_{p}\right) \left( r_{a}-r\right) }}{r}dr\text{ with\ }\left\{ 
\begin{array}{l}
r_{p}+r_{a}=-\frac{\mu }{\xi } \\ 
r_{p}r_{a}=-\frac{\Lambda ^{2}}{2\xi }%
\end{array}%
\right.   \label{AKEeq2}
\end{eqnarray}%
as in~\eqref{eq:RAke}. The radial period and the azimuthal precession are just partial derivatives
of the radial action according to~\eqref{djde} and~\eqref{djdl}:
\begin{equation*}
\frac{\partial \mathcal{A}_{r}}{\partial \xi }=\frac{\tau _{r}}{2\pi }\text{
and\ }\frac{\partial \mathcal{A}_{r}}{\partial \Lambda }=-\frac{\Delta
\varphi }{2\pi }=-n_{\varphi }.
\end{equation*}%
For a negative energy, the Kepler orbit is an ellipse with semi-major axis $a=%
\frac{1}{2}\left( r_{a}+r_{p}\right)$, where $r_{a}$ and $r_{p}$ are
respectively the apoastron and the periastron of the trajectory (hence $%
r_{a}\geq r_{p}$). For this Keplerian ellipse we trivially have  $\Delta \varphi =2\pi $
and then $n_{\varphi }=1$. By integration, one gets in this case%
\begin{equation*}
\frac{\partial \mathcal{A}_{r}^{\mathrm{ke}}}{\partial \Lambda }=-1\
\implies \ \mathcal{A}_{r}^{\mathrm{ke}}=-\Lambda +f\left( \xi \right) \text{%
.}
\end{equation*}%
The unknown function $f\left( \xi \right) $ could be expressed in terms of
the radial period through the relation%
\begin{equation*}
\tau _{r}=2\pi \frac{\partial \mathcal{A}_{r}^{\mathrm{ke}}}{\partial \xi }%
=2\pi f^{\prime }\left( \xi \right) .
%~\implies f\left( \xi \right) =\frac{1}{%
%2\pi }\int \tau _{r}d\xi +c
\end{equation*}
%where $c$ is a constant. From classic physics knowledge, we know the
From the classical Kepler's third law, we have $\tau _{r}=\frac{\pi \mu }{\sqrt{2}\left( -\xi
\right) ^{3/2}}$, which gives
\begin{equation}
f\left( \xi \right) =\frac{\mu }{2\sqrt{2}}\int \left( -\xi \right)
^{-3/2}d\xi +c=\frac{\mu }{\sqrt{-2\xi }}+c\ \implies \mathcal{A}_{r}^{%
\mathrm{ke}}=\frac{\mu }{\sqrt{-2\xi }}-\Lambda +c,   \label{justeavant}
\end{equation}
where $c$ is a constant.
On the one hand, for a circular Keplerian orbit we have $%
r_{a}=r_{p}$, so that $\mathcal{A}_{r}^{\mathrm{ke}}$ given by~\eqref{AKEeq2} vanishes in this case.  On the other hand, 
a circular Keplerian orbit is
characterized by $\Lambda =\frac{\mu }{\sqrt{-2\xi }}$. Combining these two
remarks in $\left( \text{\ref{justeavant}}\right) $ gives $c=0$. Plugging this
result into $\left( \text{\ref{AKEeq2}}\right)$, one obtains%
%From (\ref{eq:RAke}), we then have%
\begin{equation}
\mathcal{I}_{1}\left( r_{p},r_{a}\right) =\frac{\pi }{\sqrt{-2\xi }}\mathcal{%
A}_{r}^{\mathrm{ke}}=\frac{\pi }{2}\left( \frac{\mu }{\left( -\xi \right) }-%
\frac{2 \Lambda}{\sqrt{-2\xi }}\right), 
\label{intermdddd}
\end{equation}%
%It is clear from the definition of $\mathcal{I}_1$ in $\left( \text{\ref{defI1alphabeta}}\right) $ that $\mathcal{%
%I}_{1}\left( r_{p},r_{a}\right) =0$ for a circular orbit for which $%
%r_{a}=r_{p}$, but we kwow too from classic physics knwoledge that a circular
%Keplerian orbit is characterized by the relation $\Lambda =\frac{\mu }{\sqrt{%
%-2\xi }}$. Plugin this two facts into $\left( \text{\ref{intermdddd}}\right) 
%$ gives $c=0$. Finally, using the expression of the energy and the angular momentum
%in terms of the sum and the product of the periastron and the apoastron of
%the Keplerian ellipse given in (\ref{eq:RAke}) we obtain%
%\begin{equation*}
%\mathcal{I}_{1}\left( r_{p},r_{a}\right) =\frac{\pi }{2}\left( \frac{\mu }{%
%\left( -\xi \right) }-2\frac{\Lambda }{\sqrt{-2\xi }}\right) =\frac{\pi }{2}%
%\left( r_{p}+r_{a}-2\sqrt{r_{a}r_{p}}\right) 
%\end{equation*}%
%which is the general result for $\mathcal{I}_{1}$ usable for the calculus
%of the radial action in the general isochrone case.
where we recognize the values of the sum and the product
of $r_{a}$ and $r_{p}$ given by~\eqref{AKEeq2}. Hence,
%which provides the expected result owing to the values of the sum an the product
%of $r_{a}$ and $r_{p}$ in a Kepler potential given in $\left( \text{\ref%
%{AKEeq2}}\right) $. 
\begin{equation*}
\mathcal{I}_{1}\left( r_p,r_a\right) =\tfrac{\pi }{2}\left(
r_p+r_a-2\sqrt{r_pr_a}\right) .
\end{equation*}%
Since the above formula holds for any arbitrary positive numbers $r_p \leq r_a$,
we deduce the explicit expression of $\mathcal{I}_{1}$ given in the lemma.
%This completes the proof for $\mathcal{I}_{1}\left(
%r_{p},r_{a}\right) $ for a Keplerian potential.

In the same way, given $\tau_r=\frac{\pi}{\omega}$ and $n_\varphi=\frac{1}{2}$ to compute the radial action
with~\eqref{djde} and~\eqref{djdl}, the
proof can be done similarly for a harmonic potential.

\subsection{Proof for the expression of $\mathcal{I}_{2}$}

The result for  $\mathcal{I}_{2}\left( u_{1},u_{2}\right) $ simply comes from
the relation%
\begin{equation*}
\frac{2\left( u-1\right) }{u\left( u-2\right) }=\frac{1}{u}+\frac{1}{u-2}
\end{equation*}%
from which we have
\begin{equation}
2\mathcal{I}_{2}\left( u_{1},u_{2}\right) =\mathcal{I}_{1}\left(
u_{1},u_{2}\right) +\int_{u_{1}}^{u_{2}}\frac{\sqrt{\left( u-u_{1}\right)
\left( u_{2}-u\right) }}{u-2} du.  \label{lastintegral}
\end{equation}%
Two cases are of interest:

\begin{enumerate}
\item If $2<u_{1}<u_{2}$, then plugging $v=u-2$ into the last integral of 
\eqref{lastintegral}, we get\ $2\mathcal{I}_{2}\left(
u_{1},u_{2}\right) =\mathcal{I}_{1}\left( u_{1},u_{2}\right) +\mathcal{I}%
_{1}\left( u_{1}-2,u_{2}-2\right) $ which gives%
\begin{equation*}
\mathcal{I}_{2}\left( u_{1},u_{2}\right) =\frac{\pi }{2}\left[ u_{1}+u_{2}-%
\sqrt{u_1 u_{2}}-\sqrt{\left( u_{1}-2\right) \left( u_{2}-2\right) }-2%
\right] .
\end{equation*}

\item If $0<u_{1}<u_{2}<2$, then plugging $v=2-u$ into the last integral of 
\eqref{lastintegral}, we get\ $2\mathcal{I}_{2}\left(
u_{1},u_{2}\right) =\mathcal{I}_{1}\left( u_{1},u_{2}\right) -\mathcal{I}%
_{1}\left( 2-u_{1},2-u_{2}\right) $ which gives%
\begin{equation*}
\mathcal{I}_{2}\left( u_{1},u_{2}\right) =\frac{\pi }{2}\left[ u_{1}+u_{2}-%
\sqrt{u_{1}u_{2}}+\sqrt{\left( u_{1}-2\right) \left( u_{2}-2\right) }-2%
\right] .
\end{equation*}
\end{enumerate}

This completes the proof for $\mathcal{I}_{2}$.

\bigskip

\textbf{Acknowledgment} This work is supported by the ``IDI 2015" project
funded by the IDEX Paris-Saclay, ANR-11-IDEX-0003-02. JP especially thanks
Jean-Baptiste Fouvry for helpful discussions about Bertrand's theorem.
ASP especially thanks Alain Albouy for his great remarks on Bertrand's theorem
and for sharing his deep historical knowledge. The authors are grateful to Faisal Amlani 
for his detailed copy-editing of the paper and thank the referees of the article for their 
helpful comments and fruitful suggestions.

\end{document}